\documentclass[11pt]{article}

\usepackage[margin=1in]{geometry}
\usepackage{setspace}
\usepackage{multirow}
\usepackage{makecell}
\usepackage{booktabs}
\usepackage{afterpage}
\usepackage[bottom]{footmisc}

\usepackage{epigraph}
\usepackage{graphicx}
\usepackage{tikz}
\usetikzlibrary{arrows,automata,patterns,decorations.pathreplacing}
\usepackage{enumitem}
\usepackage{float}
\usepackage{hyperref}
\usepackage{xurl}
\usepackage{lipsum}

\usepackage[linesnumbered,ruled,vlined,boxed,algo2e]{algorithm2e}
\usepackage{nicefrac}
\usepackage{amsmath}
\usepackage{mathtools}
\usepackage{amsthm}
\usepackage{fullpage}
\usepackage{amsfonts}
\usepackage{clipboard}
\usepackage{thm-restate}
\usepackage{refcount}
\usepackage[nameinlink]{cleveref}
\newtheorem{theorem}{Theorem}

\newtheorem{conjecture}{Conjecture}

\newtheorem{lemma}[theorem]{Lemma}

\newtheorem{claim}{Claim}[theorem]

\newtheorem{question}{Question}

\newcommand{\claimsof}[1]{%
  \setcounter{claim}{0}%
  \renewcommand{\theclaim}{\getrefnumber{#1}.\arabic{claim}}%
  \def\theHclaim{#1.\arabic{claim}}%
}

\newtheorem{problem}{Problem}

\newtheorem{definition}{Definition}

\newtheorem{remark}{Remark}

\crefname{claim}{claim}{claims}
\Crefname{claim}{Claim}{Claims}
\crefname{remark}{Remark}{Remarks}
\Crefname{remark}{Remark}{Remarks}
\crefname{problem}{Problem}{Problems}
\Crefname{problem}{Problem}{Problems}
\crefname{definition}{Definition}{Definitions}
\Crefname{definition}{Definition}{Definitions}

\Crefname{enumi}{(item)}{(items)}

\usepackage{xcolor}
\usepackage{colortbl}
\definecolor{DarkGreen}{RGB}{1,50,32}
\usepackage{pgfplots}
\pgfplotsset{compat=1.18}

\usepackage{silence}
\usepackage{chngcntr}
\counterwithin{equation}{section}

\begin{document}
\ActivateWarningFilters[pdftoc]
\newcommand{\defeq}{:=}
\newcommand{\eps}{\varepsilon}

\newcommand{\liam}[1]{{\color{red} \textbf{Liam}: #1}} 
\newcommand{\avi}[1]{{\color{purple} \textbf{Avi}: #1}} 
\newcommand{\chatgpt}[1]{{\color{blue!65!black} \textbf{ChatGPT}: #1}}
\newcommand{\claude}[1]{{\color{orange!80!black}\fontencoding{T1}\selectfont \textbf{Claude}: #1}}
\newcommand{\new}[1]{{\color{blue} #1}}

\newcommand{\blue}[1]{{\color{blue}#1}}

\newcommand{\ReturnCode}{\textbf{return}}
\newcommand{\codestyle}[1]{\texttt{#1}\xspace}
\newcommand{\Initialize}{\mbox{\codestyle{Initialize}}}
\newcommand{\Cycle}{\mbox{\codestyle{Cycle}}}
\newcommand{\ADO}{\mbox{\codestyle{ADO}}}
\newcommand{\hADO}{\mbox{\codestyle{hADO}}}
\newcommand{\Query}{\mbox{\codestyle{.Query}}}
\newcommand{\ADOQuery}{\mbox{\codestyle{ADO.Query}}}
\newcommand{\ConstructADO}{\codestyle{ConstructADO}}
\newcommand{\FastAPSP}{\codestyle{FastAPSP}}
\newcommand{\Construct}{\mbox{\codestyle{Construct}}}
\newcommand{\Spanner}{\codestyle{Spanner}}
\newcommand{\RT}{\mbox{\codestyle{RT}}}
\renewcommand{\L}{\mbox{\codestyle{L}}}
\newcommand{\CycleOdd}{\codestyle{CycleOdd}}
\newcommand{\BallOrCycle}{\codestyle{BallOrCycle}}
\newcommand{\ClusterOrCycleBounded}{\codestyle{ClusterOrCycleBounded}}
\newcommand{\ClusterOrCycle}{\codestyle{ClusterOrCycle}}
\newcommand{\SimpleCycle}{\codestyle{SimpleCycle}}
\newcommand{\Next}{\codestyle{Next}}
\newcommand{\Sample}{\codestyle{Sample}}
\newcommand{\Dijkstra}{\codestyle{Dijkstra}}
\newcommand{\BundleDijkstra}{\codestyle{BundleDijkstra}}
\newcommand{\ActiveBellmanFord}{\codestyle{ActiveBellmanFord}}
\newcommand{\Preprocess}{\codestyle{Preprocess}}
\newcommand{\HashTable}{\codestyle{HashTable}}
\newcommand{\Heap}{\codestyle{Heap}}
\newcommand{\RelaxNext}{\codestyle{RelaxNext}}

\newcommand{\PreprocessGraph}{\codestyle{Initialize}}
\newcommand{\Route}{\codestyle{Route}}
\newcommand{\TreeRoute}{\codestyle{TreeRoute}}
\newcommand{\N}{\mathbb{N}}
\newcommand{\MinCycle}{\codestyle{MinCycle}}
\newcommand{\Ball}{\codestyle{Ball}}
\newcommand{\Bundle}{\codestyle{Bundle}}
\newcommand{\DistanceOracle}{\codestyle{TZ-DistanceOracle}}
\newcommand{\SparseOrCycle}{\codestyle{SparseOrCycle}}
\newcommand{\Intersection}{\codestyle{Intersection}}
\newcommand{\CycleAdditive}{\codestyle{CycleAdditive}}
\newcommand{\GenerateSi}{\codestyle{ComputeS}}

\newcommand{\codeNull}{\codestyle{null}}
\newcommand{\codeYes}{\codestyle{Yes}}
\newcommand{\codeNo}{\codestyle{No}}
\newcommand{\codeAnd}{~ \mathrm{and} ~}
\newcommand{\codeOr}{~ \mathrm{or} ~}
\newcommand{\wt}{\ell}
\newcommand{\Cl}{CL}
\newcommand{\CL}{CL}
\newcommand{\cl}{c\ell}

\newcommand{\EQ}{\;=\;}
\newcommand{\GE}{\;\ge\;}
\newcommand{\Ot}{\tilde{O}}
\newcommand{\nPSP}{$n$-PSP\xspace}
\newcommand{\stactri}{\stackrel\triangle}

\newcommand{\EE}{\mathbb{E}}
\newcommand{\RR}{\mathbb{R}}

\DeclarePairedDelimiter{\ceil}{\lceil}{\rceil}
\DeclarePairedDelimiter{\floor}{\lfloor}{\rfloor}
\DeclarePairedDelimiter{\pair}{\langle}{\rangle}

\newcommand{\dG}{\delta_G}

\title{A Tight Space Lower Bound for Pivot-Based Distance Oracles}
\date{}
\hypersetup{pdftitle={A Tight Space Lower Bound for Pivot-Based Distance Oracles}, pdfauthor={Avi Kadria, Liam Roditty}}
\author{Avi Kadria\thanks{Department of Computer Science, Bar-Ilan University, Ramat Gan 5290002, Israel.\\ E-mail: \texttt{avi.kadria3@gmail.com}.}
\and Liam Roditty\thanks{Department of Computer Science, Bar-Ilan University, Ramat Gan 5290002, Israel.\\ E-mail: \texttt{liam.roditty@biu.ac.il}.}}
\pagenumbering{roman}
\maketitle

\thispagestyle{empty}

\begin{abstract}
An approximate distance oracle (ADO) is a compact data structure that answers distance queries within a multiplicative \emph{stretch} $k$. The classical oracle of Thorup and Zwick~[STOC'01, JACM'05] has stretch $k=2\ell-1$ and uses $O(\ell n^{1+1/\ell})$ space, for every integer $\ell\ge1$. This trade-off is optimal for dense graphs under the Erd\H{o}s girth conjecture. However, for sparse graphs ($m=O(n)$), their lower bound yields nothing beyond $\Omega(m)$ space.

The main result of this paper is a nearly tight lower bound for pivot-based $1$-hop ADOs in sparse graphs.  
A pivot-based $1$-hop ADO gives every vertex $\tilde{O}(1)$ pivots and returns as an estimate the length of a walk of the form $u\to p_u \to p_v \to v$, where $p_u$ is a pivot of $u$, $p_v$ is a pivot of $v$, and $p_u\to p_v$ is a stored distance.
We observe that all the state-of-the-art ADOs with $\tilde{O}(1)$ query time are in fact pivot-based $1$-hop ADOs. 

We prove that every \emph{pivot-based} $1$-hop ADO with \emph{real} stretch $k$ requires $\widetilde\Omega(m^{1+2/(k+1)})$ space, under the Erd\H{o}s girth conjecture for graphs with constant degree. This matches the best known ADOs with $\tilde{O}(1)$ query time, at every known stretch, up to polylogarithmic factors.


We also extend our lower bound to pivot-based ADOs with any constant number $\beta$ of hops, showing that they require $\widetilde\Omega(m^{1+2/(\beta(k+1))})$ space. In addition, we obtain a tight lower bound for two-hop hopsets and improved lower bounds for general $\beta$-hop hopsets.

\end{abstract}

\clearpage
\pagenumbering{arabic}

\section{Introduction}

Let $G=(V,E)$ be an undirected graph with $n$ vertices and $m$ edges, and let $\dG(u,v)$ denote the distance between vertices $u,v\in V$. For a real parameter $k\ge 1$, a \emph{$k$-stretch approximate distance oracle} (ADO) is a data structure that, given two query vertices $u$ and $v$, returns an estimate $\widehat d(u,v)$ satisfying
$\dG(u,v)\le \widehat d(u,v)\le k\cdot \dG(u,v)$.

The central goal in the study of distance oracles is to understand the trade-off between stretch, space, and query time~\cite{Sommer14}. 
Simply storing the graph takes $O(m)$ space, but answering a query by a shortest-path search may require $O(m)$ time in the worst case. The research on ADOs focuses on answering queries much faster and ideally in $\Ot(1)$ time. Therefore, we focus on ADOs with $\Ot(1)$ query time for sparse graphs ($m=O(n)$),\footnote{The notation $\Ot(\cdot)$ suppresses factors polylogarithmic in $n$.} and study the trade-off between their stretch and space.

The seminal oracle of Thorup and Zwick~\cite{DBLP:journals/jacm/ThorupZ05} achieves stretch $2\ell-1$, space $O(\ell n^{1+1/\ell})$, and query time $O(\ell)$, for every integer $\ell\ge1$. 
They also prove a matching lower bound for dense graphs, assuming the Erd\H{o}s girth conjecture.\footnote{The Erd\H{o}s girth conjecture states that for every integer $\ell\ge1$ there are graphs with girth at least $2\ell+2$ and $\Omega(n^{1+1/\ell})$ edges.}
However, in terms of $m$, this lower bound yields only $\Omega(m)$, leaving open whether sparse graphs admit substantially smaller oracles with $\Ot(1)$ query time.

For sparse graphs, the best known ADOs with $\Ot(1)$ query time attain the following stretch--space trade-off for various stretches:
\[
  \text{stretch }k
  \qquad\longleftrightarrow\qquad
  \text{space }\Ot\!\left(m^{1+\frac{2}{k+1}}\right).
\]
These ADOs include the Thorup--Zwick oracle~\cite{DBLP:journals/jacm/ThorupZ05} for odd integer stretches, the oracle of P\v{a}tra\c{s}cu and Roditty~\cite{PatrascuR14} for stretch $2$, and the oracles of P\v{a}tra\c{s}cu, Roditty, and Thorup~\cite{PatrascuRT12} for every even integer stretch (see also~\cite{DBLP:conf/wdag/AbrahamG11}) and infinitely many fractional stretches. We refer to this trade-off as the \emph{stretch--space curve}. Whether this curve is optimal for ADOs with $\Ot(1)$ query time on sparse graphs is a central open question. To approach it, we look at how the oracles on the curve work.

These oracles, and in fact all the best known ADOs with $\Ot(1)$ query time,\footnote{This observation also applies to ADOs with efficient construction times~\cite{DBLP:conf/icalp/RodittyTZ05,DBLP:journals/talg/BaswanaS06,DBLP:journals/siamcomp/BaswanaK10,DBLP:conf/soda/Wulff-Nilsen12,KadriaR25faster} and to those with faster query times~\cite{DBLP:conf/soda/Wulff-Nilsen13,DBLP:conf/stoc/Chechik14,DBLP:conf/stoc/Chechik15}.} share a common structure: each vertex has $\Ot(1)$ fixed \emph{pivots}, and upon a query, the oracle returns the length of a walk $u\to p_u \to p_v \to v$, where $p_u$ is a pivot of $u$ and $p_v$ is a pivot of $v$. We call $u\to p_u$ and $p_v \to v$ \emph{pivot jumps}, and $p_u \to p_v$ a \emph{hop}.

We call such oracles \emph{pivot-based} $1$-hop ADOs. (See \Cref{fig:pivot-model}\textbf{(a)} for an illustration, \Cref{def:piv-ado} for the formal definition, and 
\Cref{sec:models} and \Cref{tab:conversions} for an explanation of why these oracles fit this model.)\footnote{Notice that every pivot-based $1$-hop ADO trivially yields an ADO with $\Ot(1)$ query time: a query simply looks up each of the $\Ot(1)$ pairs $(p_u,p_v)$ in a hash table of the stored hops and returns the minimum total walk weight.}
The shared structure of these constructions leads to a natural question:
\begin{question}
  Is the stretch--space curve optimal for pivot-based $1$-hop ADOs? That is, does every pivot-based $1$-hop ADO with stretch $k$ require $\widetilde{\Omega}(m^{1+\frac{2}{k+1}})$ space?
\end{question}

Assuming the Erd\H{o}s girth conjecture for constant-degree graphs (\Cref{conj:moore}), we answer this question in the affirmative: every pivot-based $1$-hop ADO with \emph{real} stretch $k\ge1$ requires $\widetilde\Omega(m^{1+\frac{2}{k+1}})$ space, matching the stretch--space curve. Formally, we prove the following theorem, where $r$ is the maximum number of pivots per vertex.\footnote{We measure space in words, assuming that vertex identifiers and stored distances fit in a word. In the pivot model, we count stored weighted edges, each requiring $O(1)$ words.}

\begin{restatable}{theorem}{thmlbpivot}\label{thm:lb-pivot}
Fix a real $k\ge1$. Under \Cref{conj:moore}, for arbitrarily large $n$ there are constant-degree $n$-vertex graphs $G$ with $m=\Theta(n)$ edges such that, for every integer $r\ge1$, every $r$-pivot $1$-hop $k$-stretch ADO for $G$ requires $\Omega\!\left(\frac{1}{r^2\log n}\cdot m^{1+\frac{2}{k+1}}\right)$ space.
\end{restatable}

Using the graphs of~\cite{DoubleAI26}, we obtain a weaker unconditional version of the same bound, namely $\widetilde\Omega(m^{1+\frac{8}{5(k+1)}})$ space. All the graphs in our lower bounds are unweighted, so the bounds hold even for unweighted graphs.

\colorlet{chop}{blue!65!black}       
\colorlet{cpiv}{green!50!black}      
\colorlet{cpath}{gray!55}            

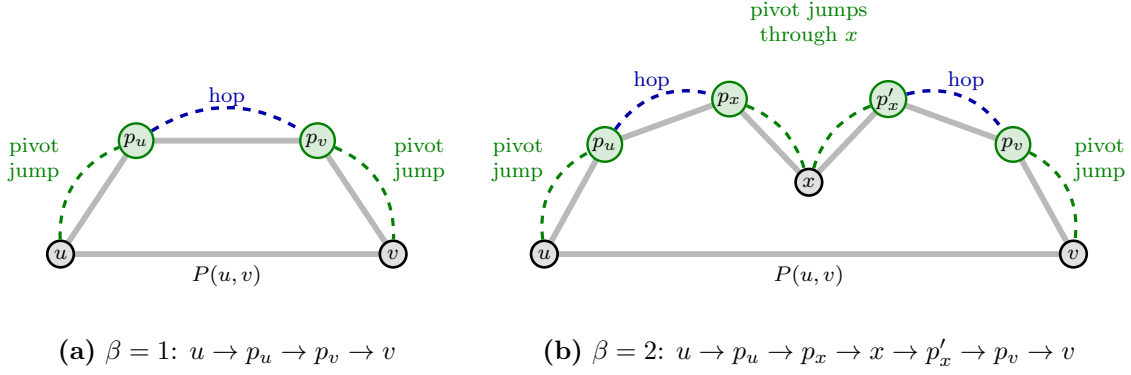
\begin{figure}[t]
\centering
\begin{tikzpicture}[
    vert/.style   = {circle, draw=black, fill=white, minimum size=5pt, inner sep=0pt, line width=0.7pt},
    bvert/.style  = {circle, draw=black, fill=black!12, minimum size=10pt, inner sep=0pt, line width=1pt, font=\scriptsize},
    pvert/.style  = {circle, draw=cpiv, fill=cpiv!15, minimum size=13pt, inner sep=0pt, line width=0.9pt, font=\scriptsize},
    hvert/.style  = {circle, draw=chop, fill=chop!12, minimum size=10pt, inner sep=0pt, line width=0.9pt, font=\scriptsize},
    gpath/.style  = {cpath, line width=2.2pt, line cap=round},
    hop/.style    = {chop, dashed, line width=1.3pt},
    piv/.style    = {cpiv, dashed, line width=1.2pt},
]

\begin{scope}
  \draw[gpath] (0,0) -- (4.4,0);
  \draw[gpath] (0,0) -- (1.0,1.5) -- (3.4,1.5) -- (4.4,0);

  \node[bvert] (uA) at (0,0)   {$u$};
  \node[bvert] (vA) at (4.4,0) {$v$};
  \node[pvert] (aA) at (1.0,1.5) {$p_u$};
  \node[pvert] (bA) at (3.4,1.5) {$p_v$};

  \draw[piv] (uA) to[bend left=35] (aA);
  \draw[piv] (bA) to[bend left=35] (vA);
  \draw[hop] (aA) to[bend left=32] (bA);

  \node[chop, font=\scriptsize] at (2.2,2.08) {hop};
  \node[cpiv, font=\scriptsize, align=right, anchor=east] at (0.12,1.25) {pivot\\jump};
  \node[cpiv, font=\scriptsize, align=left, anchor=west] at (4.28,1.25) {pivot\\jump};
  \node[font=\scriptsize] at (2.2,-0.28) {$P(u,v)$};
  \node[font=\small] at (2.2,-1.3) {\textbf{(a)} $\beta=1$: $u\to p_u\to p_v\to v$};
\end{scope}

\begin{scope}[xshift=6.4cm]
  \draw[gpath] (0,0) -- (7.0,0);
  \draw[gpath] (0,0) -- (0.8,1.45) -- (2.45,2.05) -- (3.5,0.95) -- (4.55,2.05) -- (6.2,1.45) -- (7.0,0);

  \node[bvert] (uB)  at (0,0)       {$u$};
  \node[bvert] (vB)  at (7.0,0)     {$v$};
  \node[pvert] (puB) at (0.8,1.45)  {$p_u$};
  \node[pvert, inner sep=0.5pt] (p1B) at (2.45,2.05) {$p_x$};
  \node[bvert] (xB)  at (3.5,0.95)  {$x$};
  \node[pvert, inner sep=0.5pt] (p2B) at (4.55,2.05) {$p'_x$};
  \node[pvert] (pvB) at (6.2,1.45)  {$p_v$};

  \draw[piv] (uB)  to[bend left=35] (puB);
  \draw[hop] (puB) to[bend left=35] (p1B);
  \draw[piv] (p1B) to[bend left=25] (xB);
  \draw[piv] (xB) to[bend left=25] (p2B);
  \draw[hop] (p2B) to[bend left=35] (pvB);
  \draw[piv] (pvB) to[bend left=35] (vB);

  \node[cpiv, font=\scriptsize, align=right, anchor=east] at (0.12,1.25) {pivot\\jump};
  \node[chop, font=\scriptsize] at (1.42,2.3) {hop};
  \node[cpiv, font=\scriptsize, align=center] at (3.5,3.05) {pivot jumps\\through $x$};
  \node[chop, font=\scriptsize] at (5.58,2.3) {hop};
  \node[cpiv, font=\scriptsize, align=left, anchor=west] at (6.88,1.25) {pivot\\jump};
  \node[font=\scriptsize] at (3.5,-0.28) {$P(u,v)$};
  \node[font=\small] at (3.5,-1.3) {\textbf{(b)} $\beta=2$: $u\to p_u\to p_x\to x\to p'_x\to p_v\to v$};
\end{scope}

\end{tikzpicture}
\caption{%
  The pivot model. \textbf{(a)} A walk of an $r$-pivot $1$-hop ADO (\Cref{def:piv-ado}). \textbf{(b)} A walk with two hops of an $r$-pivot $\beta$-hop ADO (\Cref{def:piv-ado-beta}).
}
\label{fig:pivot-model}
\end{figure}

\subsection{Beyond Polylogarithmic Query Time: More Hops}
While ADOs with $\Ot(1)$ query time have been the main focus, ADOs with larger query time, such as $\Ot(n^{1/\ell})$, have also been studied extensively~\cite{DBLP:conf/infocom/AgarwalGH11,DBLP:journals/algorithmica/PoratR13,DBLP:conf/soda/AgarwalG13,DBLP:conf/esa/Agarwal14,DBLP:journals/corr/abs-2307-11677,KadriaR25new}.
On sparse graphs, the larger query time makes it possible to go below the stretch--space curve. For example, for every integer $\ell\ge3$, the oracle of Kadria and Roditty~\cite{KadriaR25new} has stretch $2\ell-3$ and uses $\Ot(m^{1+1/\ell})$ space, which is below the curve.

We observe that several of these ADOs~\cite{DBLP:conf/infocom/AgarwalGH11,DBLP:conf/esa/Agarwal14,KadriaR25new} also have a pivot-based structure, but their estimates may combine more than one hop (stored distance). For example, the stretch-$2$ ADO of Agarwal, Godfrey, and Har-Peled~\cite{DBLP:conf/infocom/AgarwalGH11} may return the estimate $\dG(u,x)+\dG(x,v)$, where both distances are stored and $x$ need not be a pivot of $u$ or $v$. Such an estimate uses two hops. Hence, we extend the model to \emph{$r$-pivot $\beta$-hop ADOs}, whose estimates use at most $\beta$ hops and may take pivot jumps between consecutive hops. We call $\beta$ the \emph{hopbound}.
For example, for $\beta=2$, an estimate can be the length of the walk
\[
  u\to p_u\to p_x\to x\to p'_x\to p_v\to v,
\]
where $(p_u,p_x)$ and $(p'_x,p_v)$ are the two hops, and $p_u,p_x,p'_x,p_v$ are pivots of $u,x,x,v$, respectively. The remaining steps are pivot jumps. See \Cref{def:piv-ado-beta} for a formal definition and \Cref{fig:pivot-model}\textbf{(b)} for an illustration.



The ADOs of~\cite{DBLP:conf/infocom/AgarwalGH11,KadriaR25new} with $\Ot(n^{1+1/\ell})$ space and $\Ot(n^{1/\ell})$ query time use at most two hops in the pivot-based model (\Cref{sec:models}). Moreover, the $n$-pairs shortest paths (\nPSP)\footnote{Given an $n$-vertex graph $G$ and $n$ vertex pairs $(s_i,t_i)$, $1\le i\le n$, the $n$-pairs shortest paths problem asks for the distances $\dG(s_i,t_i)$ for all $i$.} algorithms of~\cite{DBLP:conf/focs/DalirrooyfardJW22,DBLP:conf/soda/ChechikHL25,KadriaRW26} compute their main estimates using only two hops. With larger query time, an oracle may use more hops. For example, for every integer $t\ge1$, the stretch-$(1+1/t)$ oracle of Agarwal~\cite{DBLP:conf/esa/Agarwal14} uses at most $2t$ hops.


\begin{restatable}{theorem}{thmlbpivotbeta}\label{thm:lb-pivot-beta}
Fix a real $k\ge1$ and an integer $\beta\ge1$. Under \Cref{conj:moore}, for arbitrarily large $n$ there are constant-degree $n$-vertex graphs $G$ with $m=\Theta(n)$ edges such that, for every integer $r\ge1$, every $r$-pivot $\beta$-hop $k$-stretch ADO for $G$ requires $\Omega\!\left(\frac{1}{r^2\log^{6-4/\beta}n}\cdot m^{1+\frac{2}{\beta(k+1)}}\right)$ space. For $\beta=2$, the lower bound improves to $\Omega\!\left(\frac{1}{r^2\log^3n\,\log(r+1)}\cdot m^{1+\frac{1}{k+1}}\right)$.
\end{restatable}

For $r=\Ot(1)$, the bound is $\widetilde\Omega(m^{1+2/(\beta(k+1))})$, so the required space remains polynomially superlinear for every fixed $\beta$ and $k$.

Two-hop pivot-based constructions achieve stretch $\ell$ using $\Ot(n^{1+1/\ell})$ space for $\ell=2,3$~\cite{DBLP:conf/infocom/AgarwalGH11} and $\ell=4,5$~\cite{KadriaR25new}. 
By \Cref{thm:lb-pivot-beta}, every two-hop pivot-based ADO with $r=\Ot(1)$ and $\Ot(n^{1+1/\ell})$ space has stretch at least $\ell-1$, so our bound is tight up to an additive $1$ in the stretch for these values of $\ell$.

\subsection{Hopsets and Approximate Hub Labelings}

Hopsets are a central tool for computing approximate shortest paths in parallel, distributed, and dynamic settings, and they have been studied extensively~\cite{Cohen00,MillerPVX15Hopsets,HenzingerKN18Hopsets,ElkinN19Hopsets,HuangP19,BenLevyP20,NeimanS22}. 

A \emph{$(k,\beta)$-hopset} is a set $F$ of edges weighted by their exact distances in $G$, such that every pair has a path in $G\cup F$ of at most $\beta$ edges and length at most $k$ times its distance (\Cref{def:hopset}). Its \emph{size}, which we also refer to as its space, is $|F|$.\footnote{Since the edges of $G$ are not counted in $|F|$, the space of the corresponding ADO is at most $m+|F|$. Our lower bounds are superlinear in $m$, so this additive $m$ does not affect them.}

Since every $(k,\beta)$-hopset is a $1$-pivot $\beta$-hop $k$-stretch ADO, \Cref{thm:lb-pivot-beta} already implies that $(k,\beta)$-hopsets require polynomially more space than the input graph. 
Hopsets are more restricted than pivot-based ADOs, since every edge of a path in $G\cup F$ counts as a hop and there are no pivot jumps. This allows stronger lower bounds for hopsets. We obtain the following tight lower bound for hopsets with hopbound $2$:

\begin{restatable}{theorem}{thmdensetransfer}\label{thm:dense-transfer}
Fix an integer $\ell\ge3$. Under the Erd\H{o}s girth conjecture, for arbitrarily large $n$ there are $n$-vertex graphs with $m=\Theta(n)$ edges on which every $(2\ell-3,2)$-hopset has size
\[
  \Omega\!\left(m^{1+\frac1\ell}\right).
\]
\end{restatable}

For every integer $\ell\ge3$, the oracles of Kadria and Roditty~\cite{KadriaR25new} give $(2\ell-3,2)$-hopsets of size $O(m^{1+1/\ell})$ on sparse graphs (see \Cref{sec:models} for more details), so this bound is tight. For a general hopbound $\beta$, we obtain the following bound.

\begin{restatable}{theorem}{thmbetadensetransfer}\label{thm:beta-dense-transfer}
Fix integers $\beta\ge2$ and $\ell>\beta$. Under the Erd\H{o}s girth conjecture, for arbitrarily large $n$ there are $n$-vertex graphs with $m=\Theta(n)$ edges on which every $\left(\frac{2\ell-\beta-1}{\beta-1},\beta\right)$-hopset has size
\[
  \Omega\!\left(m^{1+\frac1\ell}\right).
\]
\end{restatable}

Unlike for $\beta=2$, no matching upper bound is known for $\beta\ge3$ (see \Cref{prob:hopset-beta}).

A closely related problem is \emph{hub labeling}, in which each vertex stores distances to selected hubs, and a query is answered through a hub common to the labels of its endpoints~\cite{CohenHKZ03,AlstrupDKP16,GawrychowskiKU16}. Stretch-$k$ hub labelings correspond to $(k,2)$-hopsets, up to constant factors and an additive $O(n+m)$ term in total size (see \Cref{sec:discussion}). Hence, by \Cref{thm:dense-transfer}, for every integer $\ell\ge3$, every stretch-$(2\ell-3)$ hub labeling of the graphs of that theorem has total label size $\Omega\!\left(m^{1+1/\ell}\right)$. \Cref{tab:hopset-lb} compares the hopset bounds with previous results.

\begin{table}[t]
\centering
\footnotesize
\renewcommand{\arraystretch}{1.3}
\setlength{\tabcolsep}{3.4pt}
\begin{tabular}{@{}>{\raggedright\arraybackslash}p{1.65cm}>{\raggedright\arraybackslash}p{3.05cm}>{\raggedright\arraybackslash}p{1.2cm}>{\raggedright\arraybackslash}p{1.75cm}>{\raggedright\arraybackslash}p{2.4cm}>{\raggedright\arraybackslash}p{4.5cm}@{}}
\toprule
\textbf{Reference} & \textbf{Size} & \textbf{Stretch} & \textbf{Hopbound} & \textbf{Conjecture} & \textbf{Comment}\\
\midrule
\cite{KosowskiUV19} & $\Omega(m^{2-o(1)})$ & $1$ & $2$ & None & \\
\cite{NeimanS22} & $\Omega(m)$ & $k$ & $\beta$ & None & \mbox{$m=\Theta\big(n^{1+\frac{1}{\ceil{(k+1)(\beta+1)}+1}}\big)$} \\
\midrule
\rowcolor{black!7}\Cref{thm:dense-transfer} & $\Omega(m^{1+\frac{1}{\ell}})$ & $2\ell-3$ & $2$ & Erd\H{o}s girth & Matches the upper bound \\
\rowcolor{black!7}\Cref{thm:beta-dense-transfer} & $\Omega(m^{1+\frac{1}{\ell}})$ & $\frac{2\ell-\beta-1}{\beta-1}$ & $\beta$ & Erd\H{o}s girth &  \\
\midrule
\cite{KadriaR25new} & $O(m^{1+\frac{1}{\ell}})$ & $2\ell-3$ & $2$ & None & Matches the lower bound \\
\bottomrule
\end{tabular}
\caption{Size bounds for hopsets with stretch $k$ and hopbound $\beta$. All rows except the last are lower bounds, and $\ell>\beta$ is an integer. Unless stated otherwise, the bounds hold for graphs with $m=\Theta(n)$.}
\label{tab:hopset-lb}
\end{table}

\subsection{Related Work}\label{sec:related}

\paragraph{Known lower bounds for distance oracles.}
Sommer, Verbin, and Yu~\cite{DBLP:conf/focs/SommerVY09} prove a cell-probe space lower bound of $n^{1+\Omega(1/(kt))}/\log n$ for stretch $k$ and query time $t$. The lower bound is far from the known upper bounds, and does not rule out even linear-space constant-stretch ADOs with logarithmic query time.

Under the set-intersection hypotheses, P\v{a}tra\c{s}cu and Roditty~\cite{PatrascuR14} and P\v{a}tra\c{s}cu, Roditty, and Thorup~\cite{PatrascuRT12} give space lower bounds for ADOs with constant query time and stretches below $3$ (see also~\cite{CohenP10} and \Cref{tab:lb}).\footnote{Kopelowitz, Korin, and Roditty~\cite{DBLP:conf/icalp/KopelowitzKR24} show that stretch $2-\eps$ requires $\widetilde\Omega(n^2)$ space, even with any constant additive error and on graphs of polylogarithmic maximum degree.} The lower bound of P\v{a}tra\c{s}cu, Roditty, and Thorup~\cite{PatrascuRT12} is $\widetilde\Omega(m^{1+\frac{2}{k+1}})$, the same as ours, at the stretches $k=3-\frac{2}{\ell}$ for an integer $\ell$, and it applies only to stretches below $3$, whereas our lower bound holds for every real $k\ge1$. Moreover, their lower bounds apply only to oracles with $O(1)$ query time, whereas ours hold for every pivot-based ADO, regardless of its query time.

\begin{table}[!t]
\centering
\footnotesize
\renewcommand{\arraystretch}{1.3}
\setlength{\tabcolsep}{3.4pt}
\begin{tabular}{@{}>{\raggedright\arraybackslash}p{1.65cm}>{\raggedright\arraybackslash}p{3.3cm}>{\raggedright\arraybackslash}p{1.7cm}>{\raggedright\arraybackslash}p{1.0cm}>{\raggedright\arraybackslash}p{2.4cm}>{\raggedright\arraybackslash}p{5.0cm}@{}}
\toprule
\textbf{Reference} & \textbf{Space LB} & \textbf{Stretch} & \textbf{Query time} & \textbf{Conjecture} & \textbf{Comment}\\
\midrule
\cite{DBLP:journals/jacm/ThorupZ05} & $\Omega(m)$ & $k=2\ell-1$ & Any & Erd\H{o}s girth & $m=\Theta(n^{1+\frac{2}{k+1}})$, $\ell\in \N$\\
\midrule
\rowcolor{black!7}\Cref{thm:lb-pivot} & $\widetilde\Omega(m^{1+\frac{2}{k+1}})$ & $k$ & Any & Erd\H{o}s girth\textsuperscript{*} & Pivot-based $1$-hop ADO \\
\rowcolor{black!7}\Cref{thm:lb-pivot-beta} & $\widetilde\Omega(m^{1+\frac{1}{k+1}})$ & $k$ & Any & Erd\H{o}s girth\textsuperscript{*} & Pivot-based $2$-hop ADO \\
\rowcolor{black!7}\Cref{thm:lb-pivot-beta} & $\widetilde\Omega(m^{1+\frac{2}{\beta(k+1)}})$ & $k$ & Any & Erd\H{o}s girth\textsuperscript{*} & Pivot-based $\beta$-hop ADO \\
\midrule
\cite{DBLP:conf/focs/SommerVY09} & $\widetilde\Omega(m^{1+\Omega(\frac{1}{kt})})$ & $k$ & $t$ & None & Cell-probe model \\
\cite{PatrascuR14} & $\widetilde\Omega(m^{2})$ & $2-\eps$ & $O(1)$ & Set-intersection & \\
\cite{PatrascuRT12} & $\widetilde\Omega(m^{1+\frac{2}{k+1}})$ & $k=3-\frac{2}{\ell}$ & $O(1)$ & Set-intersection & $\ell\in\N$ \\
\bottomrule
\end{tabular}
\caption{Space lower bounds for approximate distance oracles for various stretches. Unless stated otherwise, the bound holds for graphs with $m=\widetilde{\Theta}(n)$. The `*' denotes the conjecture for constant-degree graphs (\Cref{conj:moore}). The $\widetilde{\Theta}$ and $\widetilde{\Omega}$ notations omit polylogarithmic factors.}
\label{tab:lb}
\end{table}

\paragraph{Hopsets and labels.}
Earlier lower bounds address exact and near-exact distances, both for hopsets~\cite{DBLP:journals/siamcomp/AbboudBP18,HuangP21,KoganP22focs,BodwinH23,HoppenworthXX25} and for hub labelings and distance labelings~\cite{GavoilleKKPP01,CohenHKZ03,DBLP:journals/jal/GavoillePPR04,AlstrupGHP16,KosowskiUV19}.

For arbitrary stretch, Neiman and Shabat~\cite{NeimanS22} show that a hopset of size $O(n^{1+1/\ell})$ must satisfy $k\beta=\Omega(\ell)$. For fixed stretch and hopbound, they achieve an $\Omega(m)$ lower bound on the number of added edges in their hard graphs. In contrast, our lower bounds are polynomially superlinear in $m$, even on constant-degree graphs.
Moreover, even in terms of $n$, our lower bounds are stronger. For example, for stretch $3$ and hopbound $2$, their lower bound gives $\Omega(n^{1+1/13})$, while \Cref{thm:dense-transfer} gives $\Omega(n^{4/3})$.

\paragraph{Organization.}
\Cref{sec:prelim} gives the preliminaries, and in \Cref{sec:overview} we give a technical overview. We develop the charging framework on constant-degree graphs in \Cref{sec:robust-hopsets}, beginning with hopsets, and extend it to pivot-based ADOs in \Cref{sec:pivot-lower-bounds}. \Cref{sec:dense-transfer} uses a separate construction to obtain stronger hopset lower bounds. \Cref{sec:discussion} discusses applications, girth assumptions, and open questions. \Cref{sec:models} shows how known oracles fit the pivot model.

\section{Preliminaries}\label{sec:prelim}

Let $G=(V,E)$ be an undirected, unweighted graph on $n=|V|$ vertices and $m=|E|$ edges. We denote by
$\dG(x,y)$ the shortest-path distance between $x$ and $y$ in $G$.
We denote by $P(u,v)$ a fixed shortest path between $u$ and $v$, chosen so that $P(v,u)$ is the reverse of $P(u,v)$. We write $z\in P(u,v)$ if $z$ is a vertex of $P(u,v)$. The \emph{length} $|P(u,v)|$ of the path is its number of edges, so $|P(u,v)|=\dG(u,v)$.

For $u\in V$ and $R\in\N$, we
denote by $B_R(u)$ the ball around $u$ with radius $R$, i.e., $B_R(u)=\{x\in V:\dG(u,x)\le R\}$.
The \emph{girth} $g$ is the length of a shortest cycle in $G$.
For a set $S\subseteq V$, we denote by $G[S]$ the subgraph of $G$ \emph{induced} by $S$, i.e., the graph whose vertex set is $S$ and whose edge set consists of all the edges of $E$ with both endpoints in $S$.

We now define hopsets.
\begin{definition}[$(k,\beta)$-hopset]\label{def:hopset}
Let $\beta\ge 1$ be an integer, and let $k\ge 1$ be a real number. A \emph{$(k,\beta)$-hopset} for $G$ is a set $F$ of edges on $V$, where each edge $(x,y)\in F$ has weight $\dG(x,y)$, such that for all $u,v\in V$
\[
  d_{G\cup F}^{\le\beta}(u,v)\;\le\;k\cdot\dG(u,v),
\]
where $d_{G\cup F}^{\le\beta}(u,v)$ is the minimum weight of a path between $u$ and $v$ in $G\cup F$ that uses at most $\beta$ edges. The \emph{size} of the hopset is $|F|$.
\end{definition}

Next, we define the \emph{hop graph} $H=(V,E_H,w_H)$ of $F$, where $E_H=E\cup F$ and $w_H(x,y)=\dG(x,y)$ for every $(x,y)\in E_H$, and we write $d_H^{\le\beta}=d_{G\cup F}^{\le\beta}$. We call the edges of $H$ \emph{hops}. Since hop weights never underestimate distances, the triangle inequality gives
\[
  \dG(u,v)\;\le\;d_H^{\le\beta}(u,v)\;\le\;k\cdot\dG(u,v)\qquad\text{for all } u,v\in V.
\]
Our proofs use only these inequalities and the fact that $w_H(x,y)\ge\dG(x,y)$ for every hop, so they also apply to ADOs that store overestimates of the distances rather than exact distances.\footnote{For example, the ADO of Baswana and Kavitha~\cite{DBLP:journals/siamcomp/BaswanaK10}.}

Our proofs give a lower bound on $|E_H|$. Throughout the paper, we consider sparse graphs with $m=O(n)$ edges. Since $|F|\ge|E_H|-m$, every superlinear lower bound on $|E_H|$ is also a lower bound on the size $|F|$ of the hopset.

Hops are undirected. We write a hop as an ordered pair $(u,x)$ when the endpoint $u$ matters (as
when charging pairs to hops in our proofs). Thus each unordered edge $(u,x)\in E_H$ yields the two ordered
hops $(u,x)$ and $(x,u)$, and there are $2|E_H|$ ordered hops in total.

\subsection{High-Girth Graphs}

The lower bounds of \Cref{sec:robust-hopsets,sec:pivot-lower-bounds} use $d$-regular graphs whose girth is large.
By the \emph{Moore bound}, a $d$-regular graph of girth $g\ge 2\ell+1$ on $n$ vertices satisfies $n\ge (d-1)^{\ell}$. For constant $d$ this gives
$g\le 2\log_{d-1}n+O(1)$.
We need graphs whose girth almost matches this bound. We assume the following conjecture.

\begin{conjecture}[Erd\H{o}s girth conjecture for constant-degree graphs~\cite{ExooJ08}]\label{conj:moore}
There are constants $d\ge 3$ and $C\ge0$ such that for infinitely many $n$ some $d$-regular graph
on $n$ vertices has girth $g\ge 2\log_{d-1}(n)-C$.
\end{conjecture}
Unconditionally, Erd\H{o}s and Sachs~\cite{ErdosSachs63} proved that $d$-regular graphs of girth
$\log_{d-1}n$ exist for every $d\ge3$, and the Ramanujan graphs of
Lubotzky, Phillips, and Sarnak~\cite{LPS88} achieve girth $\tfrac43\log_{d-1}n$. The best construction to date is due to Double Agent~\cite{DoubleAI26} and gets four fifths of the way to the Moore bound:

\begin{lemma}[Explicit high-girth graphs~\cite{DoubleAI26}]\label{lem:explicit-girth}
Let $d=q^3+1$ for a prime power $q$. There is a constant $C'=C'(d)\ge0$ such that for arbitrarily large $n$ there is an
explicit $d$-regular graph on $n$ vertices with girth $g\ge\tfrac{8}{5}\log_{d-1}n-C'$.
\end{lemma}

The lower bounds of \Cref{sec:robust-hopsets,sec:pivot-lower-bounds} are stated under \Cref{conj:moore}. Running the same proofs on the graphs of \Cref{lem:explicit-girth} instead
yields unconditional lower bounds, with the exponent above $1$ multiplied by $\tfrac{4}{5}$.

On any such $d$-regular high-girth graph, our arguments rely on one simple structural property: every ball of radius up to roughly half the girth is a tree. Consequently, shortest paths are locally unique, and the number of vertices at each distance depends only on the degree.

\begin{lemma}[Local tree structure]\label{lem:tree}
Let $G$ be a $d$-regular graph of girth $g$, let $h=\floor{g/2}-1$, let $u\in V$, and let
$0\le R\le h$. Then $B_R(u)$ induces a tree. Consequently, if $\dG(u,v)=R$, then $P(u,v)$ is the unique shortest path from $u$ to $v$, and the number of vertices at distance exactly $R$ from $u$ is $d(d-1)^{R-1}$ for $R\ge1$.
Moreover, for vertices $a\ne b$ and an integer $R\ge0$ with $\dG(b,a)+R\le h$, the number of vertices $v$ with $\dG(a,v)=R$ and $a\in P(b,v)$ (the vertices at distance $R$ from $a$ in the subtree hanging off $a$ away from $b$) is $(d-1)^{R}$.
\end{lemma}

\begin{proof}
Let $T_u$ be a BFS tree of the subgraph induced by $B_h(u)$, rooted at $u$. Suppose that there is an edge $\{x,y\}$ of $G[B_h(u)]$ that does not belong to $T_u$. Then this edge, together with the paths from $u$ to $x$ and from $u$ to $y$ in $T_u$, contains a cycle of length at most
\[
\dG(u,x)+\dG(u,y)+1 \le 2h+1 < g,
\]
contradicting the assumption that $G$ has girth $g$. Hence $G[B_h(u)]$ is a tree, and therefore every shortest path of length at most $h$ is unique.

Next, in this tree the root has $d$ children, while every other vertex at depth less than $h$ has exactly $d-1$ children, since one incident edge leads to its parent. Thus, for every $1\le R\le h$, the number of vertices at distance exactly $R$ from $u$ is
$d(d-1)^{R-1}$.

For the last part, apply the above with root $b$. A vertex $v$ with $\dG(a,v)=R$ and $a\in P(b,v)$ satisfies $\dG(b,v)=\dG(b,a)+R\le h$, so it is a descendant of $a$ at depth $R$ below $a$ in the tree $G[B_h(b)]$, and conversely. Since $a\ne b$, the vertex $a$ and each of its descendants at depth less than $h$ have exactly $d-1$ children, so there are exactly $(d-1)^R$ such vertices.
\end{proof}

\section{Technical Overview}\label{sec:overview}

\paragraph{The charging framework.}
The proofs of all our lower bounds follow a common charging framework on high-girth graphs with $m=O(n)$ edges. We fix a distance $c$ with $(k+1)c<g$, where $g$ is the girth of the graph, and consider a family of vertex pairs at distance exactly $c$. 
Then, for such a pair $\{u,v\}$, replacing each hop $(x,y)$ by $P(x,y)$ in a path that approximates $\dG(u,v)$ with stretch $k$ gives a walk in $G$ of length at most $kc$.

Together with $P(u,v)$, whose length is $c$, this walk is too short to contain a cycle, so it traverses every edge of $P(u,v)$. We say that the walk \emph{traces} $P(u,v)$, and use the fact that some hop covers a long segment of $P(u,v)$.
Each proof has three steps:
\begin{enumerate}
  \item We charge each such pair to one hop of a path with at most $\beta$ hops that approximates its distance.
  \item We bound the number of pairs that can be charged to any fixed hop.
  \item We divide the number of pairs by this bound, which gives a lower bound on the number of hops, that is, on $|E_H|$.
\end{enumerate}
Since $m=\Theta(n)$ on these graphs, bounds in terms of $n$ and of $m$ are equivalent. Our main lower bound is for pivot-based ADOs, but, for the sake of simplicity, we first present the charging framework for hopsets.

\paragraph{Warmup: exact two-hop hopsets.}
To illustrate the charging argument in its simplest form, we start with the exact case, $k=1$, in \Cref{sec:warmup}. Here the condition $(k+1)c<g$ allows $c\approx g/2$. Consider a pair $\{u,v\}$ with $\dG(u,v)=c$, and let its distance be realized by two hops through a vertex $x$. Since the path is exact and hop weights never underestimate distances, we have $\dG(u,v)=\dG(u,x)+\dG(x,v)$. Thus $x$ lies on the shortest path $P(u,v)$, which is unique since $c<g/2$.

Assume, without loss of generality, that $\dG(u,x)\ge \dG(x,v)$. Then $\dG(u,x)\ge c/2$, and we charge the pair $\{u,v\}$ to the ordered hop $(u,x)$. Since $x\in P(u,v)$, every pair charged to $(u,x)$ must have its second endpoint $v$ at distance exactly $c-\dG(u,x)\le c/2$ from $x$, in the subtree hanging off $x$ in the tree rooted at $u$ (see \Cref{fig:charge-exact}). By the local tree structure of \Cref{lem:tree}, there are at most $(d-1)^{c/2}$ such choices for $v$. On the other hand, there are $\Theta(n(d-1)^c)$ pairs at distance $c$. Therefore:
\[|E_H|=\Omega\left(\frac{n(d-1)^c}{(d-1)^{c/2}}\right)=\Omega(n(d-1)^{c/2})=\Omega(n^{3/2}),\]
where the last equality uses $c\approx g/2\approx \log_{d-1}(n)$.

\paragraph{Two-hop hopsets with stretch.}
In \Cref{sec:2hop}, we prove a lower bound for $(k,2)$-hopsets on constant-degree graphs.
The first step is the same as in the exact case: if the pair $\{u,v\}$ is approximated using the hops $(u,x)$ and $(x,v)$, and $\dG(u,x)\ge \dG(x,v)$, then we charge the pair $\{u,v\}$ to the ordered hop $(u,x)$. However, allowing stretch introduces a new difficulty in the second step: the intermediate vertex $x$ is not necessarily on the shortest path $P(u,v)$, so it cannot anchor the counting as in the exact case.

The main challenge is to show that the hop $(u,x)$ is still charged by only a few pairs. Let $\tau\in P(u,v)$ be the midpoint of the shortest path, so that $\dG(u,\tau)=\dG(\tau,v)=c/2$.
We prove that, even when $x$ is not on $P(u,v)$, the high girth of the graph forces $\tau$ to be on $P(u,x)$ whenever $\dG(u,x)\ge \dG(x,v)$. Indeed, the paths $P(u,x)$, $P(x,v)$, and $P(u,v)$ have fewer than $g$ edges together, so their union is a tree. In this tree, branching toward $x$ before $\tau$ would make $x$ farther from $v$ than from $u$, which contradicts the fact that $\dG(u,x)\ge \dG(x,v)$. See \Cref{fig:charge-stretch}\textbf{(a)} for an illustration.

Thus, for a fixed hop $(u,x)$, the midpoint $\tau$ is fixed. Every vertex $v$ whose pair $\{u,v\}$ charges $(u,x)$ must then lie at distance exactly $c/2$ from this fixed vertex $\tau$, in the subtree hanging off $\tau$ away from $u$, since $\tau\in P(u,v)$ (see \Cref{fig:charge-stretch}\textbf{(b)}). By \Cref{lem:tree}, there are at most $(d-1)^{c/2}$ such choices for $v$.
Since there are $\Theta(n(d-1)^c)$ pairs at distance $c$, we obtain
\[|E_H|=\Omega\left(\frac{n(d-1)^c}{(d-1)^{c/2}}\right)=\Omega(n(d-1)^{c/2})=\Omega(n^{1+\frac{1}{k+1}}),\]
where the last equality uses $c\approx g/(k+1)\approx 2\log_{d-1}(n)/(k+1)$.
This is the same counting as in the warmup. The only difference is that stretch $k$ forces a smaller $c$, namely $c\approx g/(k+1)$ instead of $c\approx g/2$, and this is where the exponent $\frac{1}{k+1}$ comes from.

\paragraph{Hopsets with general hopbound.}
In \Cref{sec:betahop}, we extend the argument to $(k,\beta)$-hopsets. When $\beta>2$, the charged hop may be incident to neither $u$ nor $v$, so fixing it no longer fixes either query endpoint. Moreover, the hop no longer determines the midpoint of $P(u,v)$, so the midpoint cannot anchor the counting as in the two-hop case.

The key is that all the hops together must trace $P(u,v)$, so one hop $(x,y)$ covers a contiguous segment of length at least $\sigma=c/\beta$. We select a segment of $\sigma$ edges, called a \emph{window}, and charge the pair to the hop together with this window.
Fixing the window and its position in $P(u,v)$ limits the choices for both endpoints: $u$ and $v$ lie in the subtrees hanging off the two ends of the window, and only the $c-\sigma$ edges outside the window remain free, leaving at most $(d-1)^{c-\sigma}$ pairs (see \Cref{fig:beta-trace}). Since $P(x,y)$ has at most $kc=O(\log n)$ edges, a fixed hop has $O(\log n)$ choices of the window, and the window has $O(\log n)$ choices of its position in $P(u,v)$. Comparing with the $\Theta(n(d-1)^c)$ pairs at distance $c$ gives:
\[
|E_H|
=
\Omega\left(\frac{n(d-1)^c}{\log^2 n\,(d-1)^{c-\sigma}}\right)
=
\Omega\left(\frac{n(d-1)^\sigma}{\log^2 n}\right)
=
\Omega\left(\frac{n^{1+\frac{2}{\beta(k+1)}}}{\log^2 n}\right),
\]
where $\sigma=c/\beta\approx 2\log_{d-1}(n)/(\beta(k+1))$.

\paragraph{Pivot-based ADOs with one hop.}
In \Cref{sec:pivot}, we prove our main theorem, a lower bound for pivot-based $1$-hop ADOs in sparse graphs that matches the stretch--space curve.
Recall that a pivot-based $1$-hop ADO estimates $\dG(u,v)$ by a walk $u\to p_u\to p_v\to v$, which makes a pivot jump from $u$ to one of its pivots $p_u\in\mathcal{P}(u)$, takes a single hop $(p_u,p_v)$, and makes a pivot jump from a pivot $p_v$ of $v$ to $v$ (see \Cref{def:piv-ado} and \Cref{fig:pivot-model}\textbf{(a)}).
As in the charging framework, this walk traces $P(u,v)$, and we charge the pair to its single hop $(p_u,p_v)$.

The new difficulty is that pivot jumps are not hops, so the jump $u\to p_u$ may cover a long prefix of $P(u,v)$ and leave less of it to the hop. However, each vertex has only $r$ pivots, and a pivot $p_u$ of $u$ lets the walk skip a long prefix of $P(u,v)$ only if $P(u,p_u)$ follows that prefix. Hence we add a preliminary step, which discards the few pairs with such a pivot. We say that a pivot $p_u$ of $u$ is \emph{aligned with} $v$ if $P(u,p_u)$ follows the first $t$ edges of $P(u,v)$, where $t=O(\log(r+1))$ is chosen such that $(d-1)^t=\Theta(r)$ (see \Cref{fig:pivot}\textbf{(a)}). By the tree structure, a single pivot $p_u$ of $u$ is aligned only with the vertices $v$ in the subtree hanging off the vertex at distance $t$ from $u$ on $P(u,p_u)$. These are at most a $\frac{1}{4r}$ fraction of the vertices at distance $c$ from $u$ (see \Cref{fig:aligned-tree}). Summing over the at most $r$ pivots of $u$ and of $v$ shows that fewer than half of the pairs have an aligned pivot. We discard these pairs.

For every remaining pair, the hop $(p_u,p_v)$ covers all of $P(u,v)$ except for fewer than $t$ edges at each end (see \Cref{fig:pivot}\textbf{(b)}). Hence the hop determines the middle part of $P(u,v)$ up to its $O(\log n)$ possible positions on $P(p_u,p_v)$, and then there are at most $(d-1)^{t}$ choices for each of $u$ and $v$. Thus a hop is charged by $O(\log n\,(d-1)^{2t})=O(r^2\log n)$ pairs, which gives
\[
|E_H|
=
\Omega\left(\frac{n(d-1)^c}{r^2\log n}\right)
=
\Omega\left(\frac{n^{1+\frac{2}{k+1}}}{r^2\log n}\right),
\]
where we use $c\approx g/(k+1)\approx 2\log_{d-1}(n)/(k+1)$.
Unlike in the two-hop argument, the single hop covers almost all of $P(u,v)$, and not only half of it. Hence, for $r=\Ot(1)$, a hop is charged by only polylogarithmically many pairs, and this doubles the exponent from $\frac{1}{k+1}$ to $\frac{2}{k+1}$.

\paragraph{Pivot-based ADOs with general hopbound.}
In \Cref{sec:pivot-beta}, we extend the result to $r$-pivot $\beta$-hop ADOs by adapting the window argument to pivots. A walk now makes up to $2\beta$ pivot jumps (see \Cref{def:piv-ado-beta} and \Cref{fig:pivot-model}\textbf{(b)}). The jumps at intermediate vertices may start anywhere in the graph and cover a long segment in the middle of $P(u,v)$ without using a hop, so discarding pairs with an aligned endpoint pivot no longer suffices.
To handle these jumps, we extend the charging step using the window argument: fixing a long segment shared by a pivot jump and $P(u,v)$ limits the choices for both $u$ and $v$.
Since $|\mathcal{P}(x)|\le r$ for every vertex $x$, there are at most $nr$ pivot jumps. As for hops, a pivot jump has $O(\log n)$ choices of a window of $\rho$ edges, the window has $O(\log n)$ choices of its position in $P(u,v)$, and fixing both leaves at most $(d-1)^{c-\rho}$ pairs. Hence only $O(nr\log^2 n\,(d-1)^{c-\rho})$ pairs have a pivot jump that shares at least $\rho$ edges with $P(u,v)$. We choose the threshold $\rho$ such that $(d-1)^\rho=\Theta(r\log^2 n)$, so these pairs are only a small fraction of the $\Theta(n(d-1)^c)$ pairs at distance $c$. We discard them, along with the pairs that have an aligned endpoint pivot.

For every remaining pair, the pivot jumps cover only short portions of $P(u,v)$, leaving the rest to the at most $\beta$ hops. Thus, some hop covers a window of length $\sigma\approx(c-2t-2(\beta-1)\rho)/\beta$, and we charge the pair to this hop and window as before (see \Cref{fig:pivot-beta}). This gives
\begin{align*}
|E_H|
&=
\Omega\left(\frac{n(d-1)^\sigma}{\log^2 n}\right)
=
\Omega\left(\frac{n^{1+\frac{2}{\beta(k+1)}}}{r^2\log^{6-4/\beta} n}\right),
\end{align*}
where $\sigma$ is now smaller than $c/\beta$ because of the portions covered by pivot jumps, and this contributes the factors $r^2$ and $\log^{4-4/\beta}n$ in the denominator.

\paragraph{Stronger hopset bounds via dense base graphs.}
In \Cref{sec:dense-transfer}, we prove hopset lower bounds that are stronger than those of \Cref{sec:robust-hopsets}. Under the Erd\H{o}s girth conjecture, for fixed integers $\beta\ge2$ and $\ell>\beta$, there are sparse graphs on which every $(k,\beta)$-hopset with $k=\frac{2\ell-\beta-1}{\beta-1}$ has size $\Omega(m^{1+1/\ell})$. For $\beta=2$, the stretch is $2\ell-3$, and the bound is tight, since the stretch-$(2\ell-3)$ oracle of Kadria and Roditty~\cite{KadriaR25new} yields a $(2\ell-3,2)$-hopset of size $O(m^{1+1/\ell})$ (\Cref{sec:models}). 

The construction starts from a dense base graph $B$ with $N$ vertices, $M=\Omega(N^{1+1/(\ell-1)})$ edges, all degrees $\Theta(D)$ for $D=\Theta(M/N)$, and girth at least $2\ell$. We replace every base edge by a path of constant length $L$, and attach $D$ private leaves to every base vertex (see \Cref{fig:dense-graph}, with $q=D$). The resulting graph is sparse, with $m=\Theta(M)$, and its girth is at least $2\ell L$.

Let $s=\beta-1$. For every path $a_0,\dots,a_s$ of $s$ base edges, the demands are the $D^2$ pairs of a leaf $u$ of $a_0$ and a leaf $v$ of $a_s$ (see \Cref{fig:dense-graph-beta}). The natural exact route $u\to a_0\to\dots\to a_s\to v$ uses $\beta+1$ hops, one more than allowed. There are $\Omega(ND^{s+2})$ such demands, each at distance $c=sL+2$.

Since $(k+1)s=2\ell-2<2\ell$, for a large enough constant $L$ we have $(k+1)c<2\ell L$. Hence, as in the charging framework, the hops answering a demand $\{u,v\}$ trace $P(u,v)$, and together they cover its $s$ subdivided base edges.

There are at most $\beta=s+1$ hops. If the first or the last hop enters a base edge, then it fixes its leaf and the base edge next to it, leaving $O(D^{s})$ choices for the other $s-1$ base vertices and the other leaf (see \Cref{fig:two-hop-transfer}). Otherwise, the at most $s-1$ internal hops cover the $s$ base edges, so one of them passes through a junction $a_i$ and fixes the two base edges around it, even if it traverses only part of each (see \Cref{fig:beta-transfer}). This leaves $O(D^{s})$ choices for the other $s-2$ base vertices and the two leaves. Hence every hop is charged by $O(D^{s})$ demands. Dividing the number of demands by this bound gives
\[
  |E_H|=\Omega\!\left(\frac{ND^{s+2}}{D^{s}}\right)
      =\Omega(ND^2)
      =\Omega\!\left(\frac{M^2}{N}\right)
      =\Omega\!\left(m^{1+\frac1\ell}\right),
\]
where the last step uses $m=\Theta(M)$ and $N=O(M^{(\ell-1)/\ell})$, which follows from $M=\Omega(N^{1+1/(\ell-1)})$.

Unlike the arguments above, this construction does not extend to pivot-based ADOs. If every leaf has its parent as a pivot, then all the selected demands have exact answers with $\beta-1$ hops, using only linear space (\Cref{rem:transfer-pivots}).

\section{Robust Hopset Lower Bounds}\label{sec:robust-hopsets}
In this section, we develop our charging framework on hopsets, where it is simplest, and in \Cref{sec:pivot-lower-bounds} we extend it to pivot-based ADOs and prove our main results. We start with exact two-hop hopsets, then allow stretch, and then allow a general hopbound.

\subsection{Warmup: Exact Two-Hop Hopsets}\label{sec:warmup}
We begin with the exact case, namely $k=1$. The resulting lower bound is not new and is in fact far weaker than the near-quadratic lower bound of Kosowski, Uzna\'nski, and Viennot~\cite{KosowskiUV19} for exact hub labelings on sparse graphs. We include it because its short proof is the skeleton of all our later proofs.

\colorlet{chop}{blue!65!black}      
\colorlet{cpath}{black}             
\colorlet{csub}{orange!80!black}    

\begin{figure}[t]
\centering
\begin{tikzpicture}[
    vert/.style  = {circle, draw=black, fill=white, minimum size=5pt, inner sep=0pt, line width=0.7pt},
    bvert/.style = {circle, draw=black, fill=black!12, minimum size=11pt, inner sep=0pt, line width=1pt, font=\scriptsize},
    pvert/.style = {circle, draw=black, fill=white, minimum size=10pt, inner sep=0pt, line width=1pt, font=\scriptsize},
    tedge/.style = {gray!55, line width=0.9pt},
    pedge/.style = {cpath, line width=1.6pt},
]

\coordinate (u)  at (0,0);
\coordinate (m1) at (0.5,0);
\coordinate (m2) at (1.0,0);
\coordinate (m3) at (1.5,0);
\coordinate (x)  at (2.0,0);

\coordinate (a1) at (1.60,-1.05);
\coordinate (a2) at (2.40,-1.05);
\coordinate (b1) at (1.15,-2.10);
\coordinate (b2) at (1.75,-2.10);
\coordinate (b3) at (2.25,-2.10);
\coordinate (b4) at (2.85,-2.10);
\coordinate (v)  at (0.90,-3.15);   
\coordinate (c2) at (1.65,-3.15);
\coordinate (c3) at (2.35,-3.15);
\coordinate (c4) at (3.10,-3.15);

\fill[csub!15] (x) -- (0.10,-3.50) -- (3.90,-3.50) -- cycle;
\draw[csub!65, line width=0.8pt] (x) -- (0.10,-3.50) -- (3.90,-3.50) -- cycle;

\draw[tedge] (x) -- (a2);
\draw[tedge] (a1) -- (b2);
\draw[tedge] (a2) -- (b3); \draw[tedge] (a2) -- (b4);
\draw[tedge] (b2) -- (c2); \draw[tedge] (b3) -- (c3); \draw[tedge] (b4) -- (c4);
\draw[pedge] (x) -- (a1) -- (b1) -- (v);
\draw[pedge] (u) -- (m1) -- (m2) -- (m3) -- (x);

\draw[chop, dashed, line width=1.4pt] (u) to[bend left=50] (x);

\node[vert] at (m1) {}; \node[vert] at (m2) {}; \node[vert] at (m3) {};
\node[vert] at (a1) {}; \node[vert] at (a2) {};
\node[vert] at (b1) {}; \node[vert] at (b2) {}; \node[vert] at (b3) {}; \node[vert] at (b4) {};
\node[vert] at (c2) {}; \node[vert] at (c3) {}; \node[vert] at (c4) {};
\node[bvert] (un) at (u)  {$u$};
\node[pvert] (xn) at (x)  {$x$};
\node[bvert] (vn) at (v)  {$v$};

\node[chop, font=\scriptsize, align=center] at (1.0,1.2)
   {hop $(u,x)\in E_H$, $\dG(u,x)\ge c/2$};
\draw[<->, gray!65] (4.15,0) -- (4.15,-3.15);
\node[gray!55!black, font=\scriptsize] at (4.55,-1.57) {$\le c/2$};
\node[csub!75!black, font=\scriptsize] at (2.0,-3.90)
   {There are at most $(d-1)^{c/2}$ choices for $v$ at the fixed depth $c-\dG(u,x)\le c/2$.};

\end{tikzpicture}
\caption{%
The charging argument of \Cref{lem:exact}: at most $(d-1)^{c/2}$ pairs can charge the hop $(u,x)$.
}
\label{fig:charge-exact}
\end{figure}

\begin{lemma}[Exact $2$-hopsets]\label{lem:exact}
Under the Erd\H{o}s girth conjecture for constant-degree graphs (\Cref{conj:moore}), for arbitrarily large $n$ there are constant-degree $n$-vertex
graphs $G$ with $m=\Theta(n)$ edges such that every $(1,2)$-hopset for $G$, i.e., every \emph{exact} $2$-hopset, requires $\Omega(m^{3/2})$ space.
\end{lemma}

\begin{proof}
Let $G$ be the graph guaranteed by \Cref{conj:moore}. Thus $G$ has constant degree $d$, girth $g\ge 2\log_{d-1}(n)-C$, and $m=\Theta(n)$ edges. Let $c=\floor{g/2}-1$. Note that:
\begin{equation}\label{eq:warmup-1}
    c=\floor{g/2}-1\ge\frac g2-\frac32\ge\log_{d-1}(n)-\frac{C+3}{2}.
\end{equation}

Let $F$ be a $(1,2)$-hopset for $G$, and let $H$ be its hop graph. We show that $|E_H|=\Omega(n^{3/2})$, which implies $|F|\ge|E_H|-m=\Omega(m^{3/2})$.

The proof contains three steps, which are also the steps of all our later proofs. First, for a pair $\{u,v\}$ at distance $c$, we locate the path that realizes its distance estimate, relative to $P(u,v)$, and charge the pair to one of its hops. Second, we bound the number of pairs that charge a single hop. Third, we count the pairs and divide by this bound.

Let $\{u,v\}$ be a pair with $\dG(u,v)=c$. Consider a path with at most two hops in $H$ that realizes $d_H^{\le 2}(u,v)$, and let $x$ be the intermediate vertex (if this path consists of the single hop $(u,v)$, set $x=v$). 
First, we show that $x$ is on the shortest path $P(u,v)$, and charge the pair to a hop.
Since hop weights never underestimate distances, and since $F$ is a $(1,2)$-hopset, we have $\dG(u,x)+\dG(x,v)\le d_H^{\le 2}(u,v)=\dG(u,v)$.
By the triangle inequality, we get that $\dG(u,x)+\dG(x,v)=\dG(u,v)$. Hence $x$ lies on a shortest path between $u$ and $v$.
Since $c\le \floor{g/2}-1$, by \Cref{lem:tree} this path is $P(u,v)$.
Assume w.l.o.g.\ that $\dG(u,x)\ge\dG(x,v)$. Since $\dG(u,x)+\dG(x,v)=c$, we have $c/2\le\dG(u,x)\le c$. We charge the pair $\{u,v\}$ to the ordered hop $(u,x)$.

Second, we show that each ordered hop is charged by at most $(d-1)^{c/2}$ pairs.
Fix an ordered hop $(u,x)$ with $\dG(u,x)\ge c/2$, and let $v$ be a vertex such that $\{u,v\}$ charges it.
Then, since $\dG(u,v)=c$, $\dG(u,x)\ge c/2$, and $x\in P(u,v)$, we have that $\dG(x,v)\le c/2$. See \Cref{fig:charge-exact} for an illustration.
Since $x\ne u$, $x\in P(u,v)$, and $\dG(u,v)=c\le\floor{g/2}-1$, by \Cref{lem:tree} there are at most $(d-1)^{c-\dG(u,x)}\le(d-1)^{c/2}$ choices for $v$ (the vertices at distance $c-\dG(u,x)$ from $x$ in the subtree hanging off $x$ away from $u$).

Third, since $c=\floor{g/2}-1$, by \Cref{lem:tree} there are $\tfrac12\,n\,d(d-1)^{c-1}$ pairs at distance $c$, each pair charged to one of $2|E_H|$ hops (two directions per hop). Since each hop is charged at most $(d-1)^{c/2}$ times, we get:
\begin{align*}
  2|E_H|\cdot (d-1)^{c/2} &\ge \tfrac12\,n\,d(d-1)^{c-1}
  \\|E_H| &\ge \frac{n\,d(d-1)^{c-1}}{4\,(d-1)^{c/2}} = \frac{n\,d}{4}(d-1)^{\frac c2-1}
  \stackrel{(\ref{eq:warmup-1})}{\ge} \frac{n\,d}{4}\cdot\frac{\sqrt n}{(d-1)^{\frac{C+3}{4}+1}} \stackrel{d,C=O(1)}{=}\Omega(n^{3/2}),
\end{align*}
as required.
\end{proof}

\subsection{Two Hops with Stretch $k>1$}\label{sec:2hop}
Next, we discuss how our framework handles stretch $k>1$. We prove that on high-girth constant-degree graphs, every
$(k,2)$-hopset requires $\Omega(m^{1+1/(k+1)})$ space. The same exponent also holds for two-hop pivot-based ADOs (\Cref{thm:lb-pivot-beta}).

Notice that unlike in the warmup, since stretch $k>1$ is allowed, $x$ need not lie on the shortest path between $u$ and $v$. Instead, we show that the longer of the two hops still ``goes through'' the midpoint $\tau$ of $P(u,v)$, that is, if $\dG(u,x)\ge\dG(x,v)$, then $\tau$ lies on $P(u,x)$ (\Cref{cl:charge-stretch}).

\begin{lemma}\label{lem:lb}
Fix a real $k\ge1$. Under \Cref{conj:moore}, for arbitrarily large $n$ there are constant-degree $n$-vertex graphs $G$ with $m=\Theta(n)$ edges such that every $(k,2)$-hopset for $G$ requires $\Omega\!\left(m^{1+\frac{1}{k+1}}\right)$ space.
\end{lemma}

\begin{proof}
\claimsof{lem:lb}%
Let $G$ be the graph guaranteed by \Cref{conj:moore}. Thus $G$ has constant degree $d$, girth $g\ge 2\log_{d-1}(n)-C$, and $m=\Theta(n)$ edges. Set $c=2\floor{\frac{g}{2(k+1)}}-2$ and $R=\floor{\frac{(k+1)c}{2}}$. Then, for $n$ large enough, $c$ is an even positive integer and $(k+1)c\le g-2(k+1)<g$. Consequently, $c\le R\le\frac g2-(k+1)\le\floor{g/2}-1$. Hence, by \Cref{lem:tree}, balls of radius at most $R$ induce trees, and shortest paths between vertices at distance at most $R$ are unique. Note that:
\begin{equation}\label{eq:main-1}
    \frac c2=\floor{\frac{g}{2(k+1)}}-1\ge\frac{g}{2(k+1)}-2\ge\frac{\log_{d-1}(n)}{k+1}-\frac{C}{2(k+1)}-2.
\end{equation}

Let $F$ be a $(k,2)$-hopset for $G$, and let $H$ be its hop graph. We show that $|E_H|=\Omega(n^{1+1/(k+1)})$, which implies $|F|\ge|E_H|-m=\Omega(m^{1+1/(k+1)})$.

The proof contains the same three steps as the proof of \Cref{lem:exact}, and only the first two change. The vertex $x$ may now be outside $P(u,v)$, so in the first step we only bound $\dG(u,x)$, and in the second step we use the midpoint of $P(u,v)$.
Let $\{u,v\}$ be a pair with $\dG(u,v)=c$. Consider a path with at most two hops in $H$ that realizes $d_H^{\le 2}(u,v)$, and let $x$ be the intermediate vertex (if this path consists of the single hop $(u,v)$, set $x=v$). Assume w.l.o.g.\ that $\dG(u,x)\ge\dG(x,v)$. First, we show that $c/2\le\dG(u,x)\le R$.
Since hop weights never underestimate distances, we have $\dG(u,x)+\dG(x,v)\le d_H^{\le 2}(u,v)\le kc$.
By the triangle inequality, $c=\dG(u,v)\le\dG(u,x)+\dG(x,v)\le 2\dG(u,x)$, and we get that $\dG(u,x)\ge c/2$.
For the upper bound, the triangle inequality gives $\dG(u,x)\le\dG(u,v)+\dG(v,x)=c+\dG(x,v)$. Together with $\dG(u,x)+\dG(x,v)\le kc$, we get that $2\dG(u,x)\le(k+1)c$, i.e., $\dG(u,x)\le\frac{(k+1)c}{2}$.
 Since distances are integers, $\dG(u,x)\le\floor{\frac{(k+1)c}{2}}=R$.
Thus, $(u,x)$ is a hop with $c/2\le\dG(u,x)\le R$, and we charge the pair $\{u,v\}$ to it.

Second, we bound the number of pairs that charge a single hop.

\colorlet{chop}{blue!65!black}       
\colorlet{cpath}{black}              
\colorlet{cbranch}{orange!85!black}  
\colorlet{cxv}{cyan!55!black}        
\colorlet{csub}{orange!80!black}     

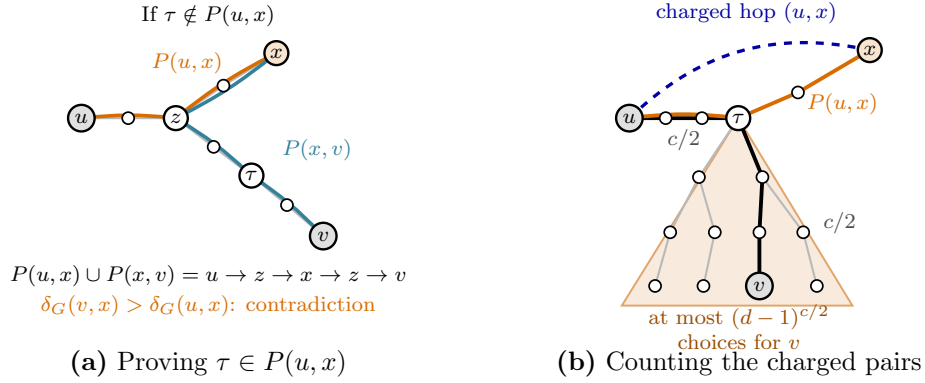
\begin{figure}[t]
\centering
\begin{tikzpicture}[
    vert/.style  = {circle, draw=black, fill=white, minimum size=4.5pt, inner sep=0pt,
                    line width=0.65pt},
    bvert/.style = {circle, draw=black, fill=black!12, minimum size=10pt, inner sep=0pt,
                    line width=0.9pt, font=\scriptsize},
    pvert/.style = {circle, draw=black, fill=white, minimum size=9pt, inner sep=0pt,
                    line width=0.9pt, font=\scriptsize},
    tedge/.style = {gray!55, line width=0.8pt},
    pedge/.style = {cpath, line width=1.45pt},
]

\begin{scope}
  \coordinate (uA)   at (0,0);
  \coordinate (a1)   at (0.62,0);
  \coordinate (zA)   at (1.25,0);
  \coordinate (b1)   at (1.88,0.43);
  \coordinate (xA)   at (2.58,0.86);
  \coordinate (a2)   at (1.75,-0.38);
  \coordinate (tauA) at (2.25,-0.76);
  \coordinate (a3)   at (2.72,-1.15);
  \coordinate (vA)   at (3.20,-1.56);

  \draw[tedge] (uA) -- (a1) -- (zA) -- (b1) -- (xA);
  \draw[tedge] (zA) -- (a2) -- (tauA) -- (a3) -- (vA);

  \draw[cbranch, line width=1.35pt, ->] (uA) to[bend left=6] (zA);
  \draw[cbranch, line width=1.35pt, ->] (zA) to[bend left=6] (xA);
  \draw[cxv, line width=1.35pt, ->] (xA) to[bend left=6] (zA);
  \draw[cxv, line width=1.35pt, ->] (zA) to[bend left=6] (tauA);
  \draw[cxv, line width=1.35pt, ->] (tauA) to[bend left=6] (vA);

  \node[vert] at (a1) {};
  \node[vert] at (a2) {};
  \node[vert] at (a3) {};
  \node[vert] at (b1) {};
  \node[bvert] at (uA) {$u$};
  \node[pvert] at (zA) {$z$};
  \node[pvert] at (tauA) {$\tau$};
  \node[pvert, fill=cbranch!18] at (xA) {$x$};
  \node[bvert] at (vA) {$v$};

  \node[font=\scriptsize, cbranch] at (1.40,0.73) {$P(u,x)$};
  \node[font=\scriptsize, cxv] at (3.12,-0.42) {$P(x,v)$};
  \node[font=\scriptsize, align=center] at (1.68,1.36)
       {If $\tau\notin P(u,x)$};
  \node[font=\scriptsize] at (1.68,-2.08)
       {$P(u,x)\cup P(x,v)=u\to z\to x\to z\to v$};
  \node[font=\scriptsize, cbranch] at (1.68,-2.48)
       {$\dG(v,x)>\dG(u,x)$: contradiction};
  \node[font=\small] at (1.68,-3.24) {\textbf{(a)} Proving $\tau\in P(u,x)$};
\end{scope}

\begin{scope}[xshift=7.25cm]
  \coordinate (uB)   at (0,0);
  \coordinate (m1)   at (0.48,0);
  \coordinate (m2)   at (0.96,0);
  \coordinate (tauB) at (1.44,0);
  \coordinate (s1)   at (0.92,-0.78);
  \coordinate (s2)   at (1.76,-0.78);
  \coordinate (q1)   at (0.54,-1.51);
  \coordinate (q2)   at (1.13,-1.51);
  \coordinate (q3)   at (1.72,-1.51);
  \coordinate (q4)   at (2.30,-1.51);
  \coordinate (r1)   at (0.34,-2.22);
  \coordinate (r2)   at (0.98,-2.22);
  \coordinate (vB)   at (1.72,-2.22);
  \coordinate (r4)   at (2.48,-2.22);
  \coordinate (cm1)  at (2.23,0.34);
  \coordinate (xB)   at (3.18,0.90);

  \fill[csub!13] (tauB) -- (-0.10,-2.48) -- (2.95,-2.48) -- cycle;
  \draw[csub!58, line width=0.75pt] (tauB) -- (-0.10,-2.48) -- (2.95,-2.48) -- cycle;

  \draw[tedge] (tauB) -- (s1);
  \draw[tedge] (s1) -- (q1); \draw[tedge] (s1) -- (q2);
  \draw[tedge] (s2) -- (q4);
  \draw[tedge] (q1) -- (r1); \draw[tedge] (q2) -- (r2); \draw[tedge] (q4) -- (r4);
  \draw[pedge] (uB) -- (m1) -- (m2) -- (tauB);
  \draw[pedge] (tauB) -- (s2) -- (q3) -- (vB);

  \draw[cbranch, line width=1.45pt] (uB) to[bend left=7] (tauB);
  \draw[cbranch, line width=1.45pt] (tauB) -- (cm1) -- (xB);
  \draw[chop, dashed, line width=1.25pt] (uB) to[bend left=28] (xB);

  \node[vert] at (m1) {}; \node[vert] at (m2) {};
  \node[vert] at (s1) {}; \node[vert] at (s2) {};
  \node[vert] at (q1) {}; \node[vert] at (q2) {}; \node[vert] at (q3) {}; \node[vert] at (q4) {};
  \node[vert] at (r1) {}; \node[vert] at (r2) {}; \node[vert] at (r4) {};
  \node[vert] at (cm1) {};
  \node[bvert] at (uB) {$u$};
  \node[pvert] at (tauB) {$\tau$};
  \node[pvert, fill=cbranch!18] at (xB) {$x$};
  \node[bvert] at (vB) {$v$};

  \node[chop, font=\scriptsize] at (1.55,1.40) {charged hop $(u,x)$};
  \node[font=\scriptsize, gray!60!black] at (0.72,-0.28) {$c/2$};
  \node[font=\scriptsize, cbranch] at (2.82,0.19) {$P(u,x)$};
  \node[font=\scriptsize, gray!60!black] at (2.78,-1.38) {$c/2$};
  \node[csub!75!black, font=\scriptsize, align=center] at (1.45,-2.76)
       {at most $(d-1)^{c/2}$\\choices for $v$};
  \node[font=\small] at (1.45,-3.24) {\textbf{(b)} Counting the charged pairs};
\end{scope}

\end{tikzpicture}
\caption{%
The two parts of the charging argument in \Cref{cl:charge-stretch}.
\textbf{(a)} If $\tau\notin P(u,x)$, the two shortest paths follow
$u\to z\to x\to z\to v$, so $\dG(v,x)>\dG(u,x)$, a contradiction.
\textbf{(b)} Hence $\tau$ is fixed by the hop, leaving at most $(d-1)^{c/2}$ choices for $v$.%
}
\label{fig:charge-stretch}
\end{figure}

\begin{claim}\label{cl:charge-stretch}
  Each hop $(u,x)$ such that $c/2 \le \dG(u,x)\le R$ is charged at most $(d-1)^{c/2}$ times.
\end{claim}
\begin{proof}
Let $v$ be a vertex such that the pair $\{u,v\}$ charges the hop $(u,x)$. Thus $\dG(u,v)=c$, $\dG(u,x)+\dG(x,v)\le kc$, and $\dG(u,x)\ge\dG(x,v)$. Since $c$ is even, $P(u,v)$ has a midpoint $\tau$, i.e., the vertex with $\dG(u,\tau)=\dG(\tau,v)=c/2$. The proof contains two steps.

First, we show that $\tau\in P(u,x)$. See \Cref{fig:charge-stretch}\textbf{(a)}. Assume for the sake of contradiction that $\tau\notin P(u,x)$. Both $P(u,\tau)$ and $P(u,x)$ start at $u$ and, since $\dG(u,\tau)=c/2\le R$ and $\dG(u,x)\le R$, lie in the ball $B_R(u)$, which induces a tree by \Cref{lem:tree}. Hence they share a common prefix, whose last vertex we denote by $z$. As $\tau\notin P(u,x)$, we have $z\ne\tau$, so $\dG(u,z)<\dG(u,\tau)=c/2$. Since $z\in P(u,\tau)\subseteq P(u,v)$, we also have $\dG(z,v)=c-\dG(u,z)>c/2$.

Because $z$ lies on both $P(u,x)$ and $P(u,v)$, the walk $P(v,z)\cup P(z,x)$ has length
$\dG(v,z)+\dG(z,x)=(c-\dG(u,z))+(\dG(u,x)-\dG(u,z))\le c+\dG(u,x)$.
We claim that this walk is a shortest $v$--$x$ path. Suppose otherwise, i.e., $\dG(v,x)<\dG(v,z)+\dG(z,x)$. Then the walk together with a shortest $v$--$x$ path contains a cycle, so
\[
g\le\dG(v,x)+\dG(v,z)+\dG(z,x)\le\dG(v,x)+c+\dG(u,x).
\]
Since $g>(k+1)c$, this gives $\dG(u,x)+\dG(x,v)\ge g-c>kc$, contradicting $\dG(u,x)+\dG(x,v)\le kc$. Hence $\dG(v,x)=\dG(v,z)+\dG(z,x)$, and since $\dG(u,z)<c/2<\dG(v,z)$,
\[
  \dG(x,v)=\dG(v,z)+\dG(z,x)>\dG(u,z)+\dG(z,x)=\dG(u,x),
\]
contradicting $\dG(u,x)\ge\dG(x,v)$. Therefore $\tau\in P(u,x)$.

Second, we bound the number of vertices $v$ that charge $(u,x)$. See \Cref{fig:charge-stretch}\textbf{(b)}. Since $\dG(u,x)\ge c/2$, the path $P(u,x)$ has a unique vertex at distance $c/2$ from $u$. By the first step this vertex is $\tau$, which is therefore the same for every $v$ that charges $(u,x)$. Each such $v$ satisfies $\tau\in P(u,v)$ and $\dG(\tau,v)=c/2$, so by \Cref{lem:tree} (applied with $\dG(u,\tau)+c/2=c\le\floor{g/2}-1$) it is one of the $(d-1)^{c/2}$ vertices at distance $c/2$ from $\tau$ in the subtree hanging off $\tau$ away from $u$. Hence $(u,x)$ is charged at most $(d-1)^{c/2}$ times, as required.
\end{proof}

Third, since $c\le\floor{g/2}-1$, by \Cref{lem:tree} there are $\tfrac12\,n\,d(d-1)^{c-1}$ pairs at distance $c$, each pair charged to one of $2|E_H|$ hops (two directions per hop). Since each hop is charged at most $(d-1)^{c/2}$ times (\Cref{cl:charge-stretch}), we get:
\begin{align*}
  2|E_H|\cdot (d-1)^{c/2} &\ge \tfrac12\,n\,d(d-1)^{c-1}
  \\|E_H| &\ge \frac{n\,d(d-1)^{c-1}}{4\,(d-1)^{c/2}} = \frac{n\,d}{4}(d-1)^{\frac c2-1}
  \\&\stackrel{(\ref{eq:main-1})}{\ge} \frac{n\,d}{4}\cdot\frac{n^{1/(k+1)}}{(d-1)^{\frac{C}{2(k+1)}+3}} \stackrel{d,C,k=O(1)}{=} \Omega\!\left(n^{1+\frac1{k+1}}\right),
\end{align*}
as required.
\end{proof}

\subsection{General Hopbound}\label{sec:betahop}
Next, we generalize our techniques to every hopbound $\beta >2$. 
Notice that unlike in the case $\beta=2$, a path with $\beta$ hops may contain hops that are not incident to $u$ or to $v$, so we can no longer charge a pair to a hop that starts at one of its endpoints. Instead, we show that the path traces $P(u,v)$, and hence one of its hops $(x,y)$ covers at least $\sigma=c/\beta$ consecutive edges of $P(u,v)$. 
We charge the pair to this hop together with a \emph{window}, which is a subpath of $P(x,y)$ with $\sigma$ edges, taken with a direction that tells which of its ends is closer to $u$. We adapt this window argument to pivot jumps in \Cref{sec:pivot-lower-bounds}.
\colorlet{clong}{orange!85!black}   
\colorlet{cwin}{red!75!black}       
\colorlet{csub}{orange!82!black}    
\colorlet{cpath}{black}             

\begin{figure}[t]
\centering
\begin{tikzpicture}[
    xscale=1.05,
    vert/.style  = {circle, draw=black, fill=white, minimum size=5pt, inner sep=0pt, line width=0.7pt},
    bvert/.style = {circle, draw=black, fill=black!12, minimum size=11pt, inner sep=0pt, line width=1pt, font=\scriptsize},
    wvert/.style = {circle, draw=cwin, fill=white, minimum size=9pt, inner sep=0pt, line width=1.1pt, font=\scriptsize},
    pedge/.style = {cpath, line width=1.5pt},
]

\foreach \i in {0,...,8} {\coordinate (n\i) at (\i,0);}
\coordinate (uu) at (-1,0);   
\draw[pedge] (uu) -- (n8);
\foreach \i in {0,1,3,4,6,7} {\node[vert] at (n\i) {};}

\coordinate (xhi) at (1.3,1.6);
\coordinate (yhi) at (6.7,1.6);
\draw[clong, line width=1.8pt]
     (xhi) to[out=-40,in=150] (n2) -- (n6) to[out=30,in=-140] (yhi);
\fill[clong] (xhi) circle (2.1pt);
\fill[clong] (yhi) circle (2.1pt);
\node[clong, font=\scriptsize] at (0.95,1.86) {$x$};
\node[clong, font=\scriptsize] at (7.05,1.86) {$y$};
\node[clong, font=\scriptsize, align=center] at (4,2.2)
   {hop $(x,y)$: $P(x,y)$ covers $\ge c/\beta=\sigma$ edges of $P(u,v)$};

\draw[cwin, line width=3pt, line cap=round] (n2) -- (n5);
\node[wvert] (a) at (n2) {$q$};
\node[wvert] (b) at (n5) {$q'$};
\node[cwin, font=\scriptsize] at (3.5,-0.5) {window $\pi$, $|\pi|=\sigma$};

\node[bvert] at (uu) {$u$};
\node[bvert] at (n8) {$v$};

\fill[csub!16] (a) -- (-0.6,-1.7) -- (2.0,-1.7) -- cycle;
\draw[csub!70, line width=0.8pt] (a) -- (-0.6,-1.7) -- (2.0,-1.7) -- cycle;
\fill[csub!16] (b) -- (5.0,-1.7) -- (7.6,-1.7) -- cycle;
\draw[csub!70, line width=0.8pt] (b) -- (5.0,-1.7) -- (7.6,-1.7) -- cycle;
\node[csub!75!black, font=\scriptsize, align=center] at (0.7,-2.05)
   {$\le(d-1)^{\dG(u,q)}$\\choices for $u$};
\node[csub!75!black, font=\scriptsize, align=center] at (6.3,-2.05)
   {$\le(d-1)^{c-\sigma-\dG(u,q)}$\\choices for $v$};

\end{tikzpicture}
\caption{%
The charging argument of \Cref{lem:lb-beta} (\Cref{cl:perwindow}): at most $(c+1)(d-1)^{c-\sigma}$ pairs can be charged to the hop $(x,y)$ through the window $\pi$.
}
\label{fig:beta-trace}
\end{figure}
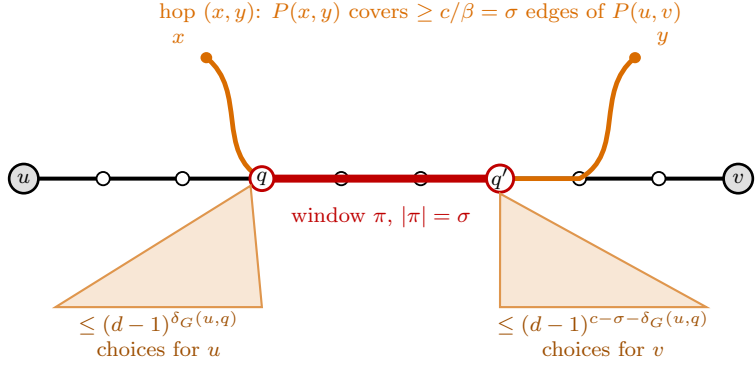

\begin{lemma}\label{lem:lb-beta}
Fix a real $k\ge1$ and an integer $\beta\ge2$. Under \Cref{conj:moore}, for arbitrarily large $n$ there are constant-degree $n$-vertex graphs $G$ with $m=\Theta(n)$ edges such that every $(k,\beta)$-hopset for $G$ requires $\Omega\!\left(\frac{1}{\log^2n}\cdot m^{1+\frac{2}{\beta(k+1)}}\right)$ space.
\end{lemma}
\begin{proof}
\claimsof{lem:lb-beta}%
Let $G$ be the graph guaranteed by \Cref{conj:moore}. Thus $G$ has constant degree $d$, girth $g\ge 2\log_{d-1}(n)-C$, and $m=\Theta(n)$ edges. Set
$\sigma=\Big\lfloor\tfrac{g}{\beta(k+1)}\Big\rfloor-1$ and $c=\beta\sigma$.
Then, for $n$ large enough, $\sigma\ge1$ and $(k+1)c=(k+1)\beta\sigma\le g-\beta(k+1)<g$. In particular, $c\le\frac g{k+1}-\beta\le\frac g2-2\le\floor{g/2}-1$, so by \Cref{lem:tree} balls of radius at most $c$ induce trees, and shortest paths between vertices at distance at most $c$ are unique. Note that:
\begin{equation}\label{eq:beta-1}
    \sigma=\floor{\tfrac{g}{\beta(k+1)}}-1\ge\frac{g}{\beta(k+1)}-2\ge\frac{2\log_{d-1}(n)-C}{\beta(k+1)}-2.
\end{equation}

Let $F$ be a $(k,\beta)$-hopset for $G$, and let $H$ be its hop graph. We show that $|E_H|=\Omega(n^{1+\frac{2}{\beta(k+1)}}/\log^2 n)$, which implies $|F|\ge|E_H|-m=\Omega(m^{1+\frac{2}{\beta(k+1)}} \cdot \tfrac{1}{\log^2 n})$.

The proof contains the same three steps as the previous proofs. In the first step we show that the hops trace $P(u,v)$, and charge the pair to a hop together with a window. In the second step we bound the number of windows of a hop, and the number of pairs charged to it through each window.
Let $\{u,v\}$ be a pair with $\dG(u,v)=c$. Let $(u=x_0,x_1),(x_1,x_2),\dots,(x_{j-1},x_j=v)$, with $j\le\beta$, be the hops of a path realizing $d_H^{\le\beta}(u,v)$, and denote by $P_H^{\le \beta}(u,v)$ the walk from $u$ to $v$ obtained by concatenating the shortest paths $P(x_0,x_1),P(x_1,x_2),\dots,P(x_{j-1},x_j)$ of these hops. Since hop weights never underestimate distances, this walk has at most $\sum_{i=1}^{j}\dG(x_{i-1},x_i)\le\sum_{i=1}^{j}w_H(x_{i-1},x_i)=d_H^{\le\beta}(u,v)\le kc$ edges.
First, we show that the walk $P_H^{\le\beta}(u,v)$ traces $P(u,v)$, that is, it covers every edge of $P(u,v)$.

\begin{claim}\label{cl:trace}
$P(u,v)\subseteq P_H^{\le\beta}(u,v)$.
\end{claim}
\begin{proof}
The edge set $P_H^{\le\beta}(u,v)\cup P(u,v)$ has at most $kc+c=(k+1)c<g$ edges. Hence it contains no cycle, so it forms a tree $T$. In a tree there is a unique simple $u$--$v$ path, and $P(u,v)$ is a simple $u$--$v$ path contained in $T$, so it is that unique path. Finally, every $u$--$v$ walk in a tree traverses every edge of the unique simple $u$--$v$ path. Since $P_H^{\le\beta}(u,v)$ is a walk from $u$ to $v$ inside $T$, it traverses every edge of $P(u,v)$. Thus $P(u,v)\subseteq P_H^{\le\beta}(u,v)$, as required.
\end{proof}

By \Cref{cl:trace}, the shortest paths $P(x_0,x_1),\dots,P(x_{j-1},x_j)$ of the hops together cover all $c$ edges of $P(u,v)$. Moreover, the edges of $P(u,v)$ covered by the shortest path of a single hop form a \emph{contiguous} segment of $P(u,v)$, since $P(u,v)$ and the shortest path of the hop are both simple paths in the tree $T$ of \Cref{cl:trace}, and the intersection of two paths in a tree is a path. As $j\le\beta$, some hop $(x,y)\in E_H$ covers a contiguous segment of $P(u,v)$ of length at least $c/\beta=\sigma$. Note that $\dG(x,y)\le kc$, since the walk $P_H^{\le\beta}(u,v)$ has at most $kc$ edges. Let $\pi$ be the first $\sigma$ edges of that segment, ordered from $u$, so that $|\pi|=\sigma$ and $\pi\subseteq P(u,v)\cap P(x,y)$. We charge the pair $\{u,v\}$ to the hop $(x,y)$ through the window $\pi$.

Second, we bound the number of pairs that can charge a single hop, in two parts: a hop is charged through only a few windows, and through each window it is charged by only a few pairs. We first bound the windows.

\begin{claim}\label{cl:windows}
Each hop $(x,y)$ with $\dG(x,y)\le kc$ is charged through at most $2(kc+1)$ distinct windows $\pi$.
\end{claim}
\begin{proof}
A window charged to $(x,y)$ is a subpath with $\sigma$ edges of the shortest path $P(x,y)$, whose length
is $\dG(x,y)\le kc$. Such a subpath of a fixed path of length
at most $kc$ is determined by its starting offset and its direction, giving at most
$2(kc-\sigma+1)\le 2(kc+1)$ choices.
\end{proof}

Next, we bound the number of pairs that are charged to a hop through a fixed window.
\begin{claim}\label{cl:perwindow}
Every hop $(x,y)$ is charged through every window $\pi$ by at most $(c+1)(d-1)^{c-\sigma}$ pairs.
\end{claim}
\begin{proof}
Let $q$ and $q'$ be the first and the last vertices of $\pi$, so $\pi=P(q,q')$ and $\dG(q,q')=\sigma$.
If $\{u,v\}$ is charged to $(x,y)$ through $\pi$, then $\pi\subseteq P(u,v)$, and
$P(u,v)=P(u,q)\cup\pi\cup P(q',v)$ with $\dG(u,q)+\sigma+\dG(q',v)=c$.

Since $\sigma+\dG(u,q)\le c\le\floor{g/2}-1$, by \Cref{lem:tree}, there are at most $(d-1)^{\dG(u,q)}$ choices for $u$ (the vertices at distance $\dG(u,q)$ from $q$ in the subtree hanging off $q$ away from $q'$). Symmetrically, there are at most $(d-1)^{\dG(q',v)}$ choices for $v$. 
As $\dG(u,q)+\dG(q',v)=c-\sigma$, for each fixed value of $\dG(u,q)$ the number of charging pairs is at most
\[(d-1)^{\dG(u,q)} \cdot (d-1)^{\dG(q',v)}=(d-1)^{c-\sigma}.\]
Since $\dG(u,q)$ ranges over $\{0,\dots,c-\sigma\}$, the total is at most $(c-\sigma+1)(d-1)^{c-\sigma}\le(c+1)(d-1)^{c-\sigma}$, as required. See \Cref{fig:beta-trace} for an illustration.
\end{proof}

Third, since $c\le\floor{g/2}-1$, by \Cref{lem:tree} there are $\tfrac12\,n\,d(d-1)^{c-1}$ pairs at distance $c$, each charged to one of the $|E_H|$ hops through one of at most $2(kc+1)$ windows (\Cref{cl:windows}), and each hop is charged through each window by at most $(c+1)(d-1)^{c-\sigma}$ pairs (\Cref{cl:perwindow}).
Since $c=\beta\sigma\le\frac{g}{k+1}=O(\log n)$, both $kc+1$ and $c+1$ are $O(\log n)$, so we get:
\begin{align*}
  2(kc+1)(c+1)(d-1)^{c-\sigma}\,|E_H| &\ge \tfrac12\,n\,d(d-1)^{c-1}
  \\|E_H| &\ge \frac{n\,d(d-1)^{c-1}}{4\,(kc+1)(c+1)(d-1)^{c-\sigma}} = \frac{n\,d}{4\,(kc+1)(c+1)(d-1)}\,(d-1)^{\sigma}
  \\&\stackrel{(\ref{eq:beta-1})}{\ge} \frac{n\,d}{4\,(kc+1)(c+1)(d-1)}\cdot\frac{n^{2/(\beta(k+1))}}{(d-1)^{\frac{C}{\beta(k+1)}+2}}
  \\&\stackrel{d,C,k,\beta=O(1)}{=}\Omega\!\left(n^{1+\frac{2}{\beta(k+1)}} \cdot \tfrac{1}{\log^2 n}\right),
\end{align*}
as required.
\end{proof}

\section{Lower Bounds for Pivot-Based ADOs}\label{sec:pivot-lower-bounds}
We now extend the charging arguments of \Cref{sec:robust-hopsets} to pivot-based ADOs. We first prove a tight lower bound for one hop, and then generalize it to any hopbound.

\subsection{One Hop}\label{sec:pivot}
In this subsection, we prove the main result of the paper, a tight lower bound for pivot-based $1$-hop ADOs. A pivot-based $1$-hop ADO estimates $\dG(u,v)$ by the weight of a walk $u\to p_u \to p_v \to v$, where $p_u$ is a pivot of $u$, $p_v$ is a pivot of $v$, and $p_u\to p_v$ is a stored \emph{hop}. See \Cref{fig:pivot-model}\textbf{(a)}. The steps $u\to p_u$ and $p_v\to v$ are \emph{pivot jumps}: they do not count as hops. Formally:

\begin{definition}[$r$-pivot $1$-hop $k$-stretch ADO]\label{def:piv-ado}
Let $r\ge1$ be an integer, and let $k\ge1$. An \emph{$r$-pivot $1$-hop $k$-stretch ADO} for $G$ consists of a weighted undirected graph $H=(V,E_H,w_H)$ with $w_H(x,y)\ge\dG(x,y)$ for every $(x,y)\in E_H$, and a set $\mathcal{P}(v)\subseteq V$ of at most $r$ \emph{pivots} for every $v\in V$, with $v\in\mathcal{P}(v)$, such that for all $u,v\in V$:
\begin{equation}\label{eq:pivot-estimate}
  \dG(u,v)\;\le\;d_{H,\mathcal{P}}^{\le1}(u,v)\;:=\;\min_{p_u\in\mathcal{P}(u),\,p_v\in\mathcal{P}(v)}\bigl(w_H(u,p_u)+w_H(p_u,p_v)+w_H(p_v,v)\bigr)\;\le\;k\,\dG(u,v),
\end{equation}
where $w_H(a,a)=0$, and $w_H(x,y)=\infty$ if $x\ne y$ and $(x,y)\notin E_H$. The \emph{space} of the oracle is $|E_H|$.
\end{definition}

Note that the oracle is not required to compute the minimum in \eqref{eq:pivot-estimate}. We only require that a walk $u\to p_u\to p_v\to v$ of weight at most $k\dG(u,v)$ exists. Hence, our lower bound holds no matter how the query algorithm searches for such a walk or how much time the search takes.

We remark that this model captures the constructions attaining the best known stretch--space trade-offs for general sparse graphs with $\Ot(1)$ query time (see \Cref{sec:models}). For example, the Thorup--Zwick oracle~\cite{DBLP:journals/jacm/ThorupZ05} with stretch $2\ell-1$ is an $\ell$-pivot $1$-hop ADO that uses $O(\ell n^{1+1/\ell})$ space. On sparse graphs with $m=\Theta(n)$, \Cref{thm:lb-pivot}, with $k=2\ell-1$ and $r=\ell$, shows that this is optimal within a factor of $O(\ell^3\log n)$.

\thmlbpivot*
\colorlet{cpiv}{green!50!black}      
\colorlet{chop}{blue!65!black}       
\colorlet{clong}{orange!85!black}    
\colorlet{cwin}{red!75!black}        
\colorlet{csub}{orange!82!black}     
\colorlet{cpath}{black}              

\begin{figure}[t]
\centering
\begin{tikzpicture}[
    vert/.style  = {circle, draw=black, fill=white, minimum size=4.5pt, inner sep=0pt, line width=0.65pt},
    bvert/.style = {circle, draw=black, fill=black!12, minimum size=10pt, inner sep=0pt, line width=0.9pt, font=\scriptsize},
    pvert/.style = {circle, draw=black, fill=white, minimum size=9pt, inner sep=0pt, line width=0.9pt, font=\scriptsize},
    ivert/.style = {circle, draw=cpiv, fill=cpiv!15, minimum size=9pt, inner sep=0pt, line width=0.9pt, font=\scriptsize},
    tedge/.style = {gray!55, line width=0.8pt},
    pedge/.style = {cpath, line width=1.45pt},
    pivedge/.style = {cpiv, dashed, line width=1.1pt},
]

\begin{scope}
  \foreach \i in {0,...,4} {\coordinate (s\i) at (0.75*\i,0);}
  \fill[cpiv!10] (-0.25,-0.28) rectangle (1.62,0.28);
  \draw[decorate, decoration={brace, amplitude=4pt, mirror}, gray!60!black]
       (0,-0.36) -- (1.5,-0.36) node[midway, below=3pt, font=\scriptsize] {$t$};

  \coordinate (br3) at (2.60,0.80);  \coordinate (a3) at (2.95,1.75);

  \draw[pedge] (s0) -- (s4);
  \draw[clong, line width=1.35pt] (s0) to[bend left=8] (s2) to[bend left=8] (s3) -- (br3) -- (a3);

  \draw[pivedge] (s0) to[out=95, in=175] (a3);

  \node[vert] at (s1) {}; \node[vert] at (s3) {};
  \node[vert] at (br3) {};
  \node[bvert] at (s0) {$u$};
  \node[bvert] at (s4) {$v$};
  \node[pvert] at (s2) {$w$};
  \node[ivert, draw=clong, fill=clong!18] at (a3) {};
  \node[clong, font=\scriptsize] at (2.95,2.07) {$p_u$};
  \node[cpiv, font=\scriptsize, align=center] at (-0.55,1.3) {pivot\\jump};
  \node[clong, font=\scriptsize] at (3.6,0.95) {$P(u,p_u)$};

  \node[font=\small] at (1.40,-3.35) {\textbf{(a)} an aligned pivot};
\end{scope}

\begin{scope}[xshift=6.3cm]
  \foreach \i in {0,...,7} {\coordinate (n\i) at (0.9*\i,0);}
  \coordinate (aa) at (0.55,1.25);  \coordinate (ma) at (0.80,0.62);
  \coordinate (bb) at (5.75,1.25);  \coordinate (mb) at (5.50,0.62);

  \fill[cpiv!10] (-0.25,-0.22) rectangle (1.8,0.22);
  \fill[cpiv!10] (4.5,-0.22) rectangle (6.55,0.22);

  \fill[csub!13] (n2) -- (0.35,-1.95) -- (2.85,-1.95) -- cycle;
  \draw[csub!58, line width=0.75pt] (n2) -- (0.35,-1.95) -- (2.85,-1.95) -- cycle;
  \fill[csub!13] (n5) -- (3.45,-1.95) -- (5.95,-1.95) -- cycle;
  \draw[csub!58, line width=0.75pt] (n5) -- (3.45,-1.95) -- (5.95,-1.95) -- cycle;

  \draw[pedge] (n0) -- (n7);
  \draw[tedge] (n1) -- (ma) -- (aa);
  \draw[tedge] (n6) -- (mb) -- (bb);

  \draw[clong, line width=1.35pt]
       (aa) -- (ma) -- (n1) to[bend left=10] (n6) -- (mb) -- (bb);

  \draw[cwin, line width=3pt, line cap=round] (n2) -- (n5);

  \draw[chop, dashed, line width=1.25pt] (aa) to[bend left=22] (bb);
  \draw[pivedge] (n0) to[bend left=30] (aa);
  \draw[pivedge] (bb) to[bend left=30] (n7);

  \node[vert] at (n3) {}; \node[vert] at (n4) {};
  \node[vert] at (ma) {}; \node[vert] at (mb) {};
  \node[bvert] at (n0) {$u$};
  \node[bvert] at (n7) {$v$};
  \node[pvert] at (n1) {};
  \node[pvert] at (n6) {};
  \node[pvert, draw=cwin] at (n2) {};
  \node[pvert, draw=cwin] at (n5) {};
  \node[ivert] at (aa) {$p_u$};
  \node[ivert] at (bb) {$p_v$};

  \node[font=\scriptsize] at (0.55,0.32) {$z_u$};
  \node[font=\scriptsize] at (5.78,0.32) {$z_v$};
  \node[cwin, font=\scriptsize] at (1.80,-0.32) {$q_u$};
  \node[cwin, font=\scriptsize] at (4.50,-0.32) {$q_v$};
  \node[cwin, font=\scriptsize] at (3.15,-0.32) {$Q$};
  \node[gray!60!black, font=\scriptsize] at (0.75,-0.36) {$t$};
  \node[gray!60!black, font=\scriptsize] at (5.55,-0.36) {$t$};
  \node[chop, font=\scriptsize] at (3.15,2.05) {charged hop $(p_u,p_v)$};
  \node[clong, font=\scriptsize] at (3.15,0.55) {$P(p_u,p_v)$};
  \node[cpiv, font=\scriptsize] at (-0.25,1.05) {pivot};
  \node[cpiv, font=\scriptsize] at (6.55,1.05) {pivot};

  \node[csub!75!black, font=\scriptsize, align=center] at (1.60,-2.35)
       {$\le(d-1)^{t}$\\choices for $u$};
  \node[csub!75!black, font=\scriptsize, align=center] at (4.70,-2.35)
       {$\le(d-1)^{t}$\\choices for $v$};
  \node[font=\small] at (3.15,-3.35) {\textbf{(b)} charging a non-aligned pair};
\end{scope}

\end{tikzpicture}
\caption{%
The charging argument of \Cref{thm:lb-pivot}. \textbf{(a)} An aligned pivot (\Cref{cl:piv-good}). \textbf{(b)} The hop $(p_u,p_v)$ covers $Q$ (\Cref{cl:piv-cover,cl:piv-count}).%
}
\label{fig:pivot}
\end{figure}
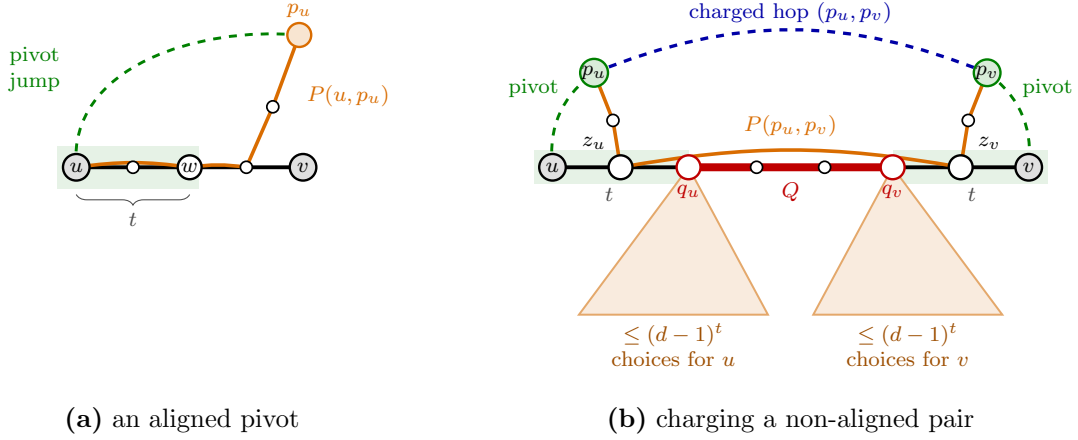

\begin{proof}
\claimsof{thm:lb-pivot}%
Let $G$ be the graph guaranteed by \Cref{conj:moore}. Thus $G$ has constant degree $d$, girth $g\ge 2\log_{d-1}(n)-C$, and $m=\Theta(n)$ edges. Set $c=\floor{\frac{g}{k+1}}-1$ and $t=1+\ceil{\log_{d-1}(4r)}$.
 Then $(k+1)c\le g-(k+1)<g$. In particular $c\le\frac g2-1$, so $c\le\floor{g/2}-1$, and by \Cref{lem:tree} shortest paths between vertices at distance at most $c$ are unique. Note that:
\begin{align}
    c&\ge\frac{g}{k+1}-2\ge\frac{2\log_{d-1}(n)-C}{k+1}-2, \qquad\text{and}\qquad kc+1\le g=O(\log n), \label{eq:piv-c}\\
    4r&\le(d-1)^{t-1}\qquad\text{and}\qquad (d-1)^{t}\le 4r(d-1)^{2}.\label{eq:piv-t}
\end{align}
Let $(H,\mathcal{P})$ be an $r$-pivot $1$-hop $k$-stretch ADO for $G$. We show that
$|E_H|=\Omega\big(n^{1+\frac{2}{k+1}}/(r^2\log n)\big)$. Since $m=\Theta(n)$, this implies the theorem.
The proof contains the same three steps as the previous proofs, preceded by a new preliminary step, which discards the pairs with an aligned pivot. The three steps are then carried out only for the remaining pairs.
We may assume $c\ge2t+1$. Otherwise, \eqref{eq:piv-c} and \eqref{eq:piv-t} give $n^{2/(k+1)}=O((d-1)^{2t})=O(r^2)$, so the claimed bound is $O(n/\log n)$ and follows from $|E_H|\ge n/2$, since every vertex must be incident to an edge of $H$.

In the preliminary step, we discard the few pairs for which a pivot shortcuts a long prefix (or suffix) of $P(u,v)$. Let $\{u,v\}$ be a pair with $\dG(u,v)=c$. We say that a pivot $p_u\in\mathcal{P}(u)$ is \emph{aligned with $v$} if $P(u,p_u)$ and $P(u,v)$ share their first $t$ edges, and that the pair $\{u,v\}$ is \emph{non-aligned} if no pivot of $u$ is aligned with $v$ and no pivot of $v$ is aligned with $u$. See \Cref{fig:pivot}\textbf{(a)}.

\colorlet{cpvt}{red!75!black}        
\colorlet{csub}{orange!82!black}     
\colorlet{cgood}{green!50!black}     

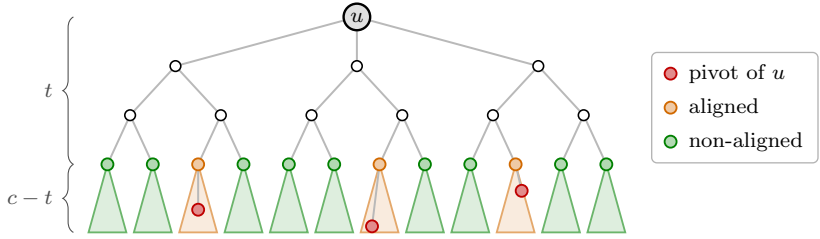
\begin{figure}[t]
\centering
\begin{tikzpicture}[
    vert/.style  = {circle, draw=black, fill=white, minimum size=4pt, inner sep=0pt, line width=0.6pt},
    bvert/.style = {circle, draw=black, fill=black!12, minimum size=10pt, inner sep=0pt, line width=0.9pt, font=\scriptsize},
    avert/.style = {circle, draw=csub, fill=csub!35, minimum size=4.5pt, inner sep=0pt, line width=0.7pt},
    gvert/.style = {circle, draw=cgood, fill=cgood!30, minimum size=4.5pt, inner sep=0pt, line width=0.7pt},
    ivert/.style = {circle, draw=cpvt, fill=cpvt!45, minimum size=4.5pt, inner sep=0pt, line width=0.7pt},
    tedge/.style = {gray!55, line width=0.75pt},
]

\coordinate (u) at (0,0);
\foreach \j/\x in {1/-2.4, 2/0, 3/2.4} {\coordinate (a\j) at (\x,-0.65);}
\foreach \j/\x in {1/-3.0, 2/-1.8, 3/-0.6, 4/0.6, 5/1.8, 6/3.0} {\coordinate (b\j) at (\x,-1.3);}
\foreach \i/\x in {1/-3.3, 2/-2.7, 3/-2.1, 4/-1.5, 5/-0.9, 6/-0.3, 7/0.3, 8/0.9, 9/1.5, 10/2.1, 11/2.7, 12/3.3} {\coordinate (w\i) at (\x,-1.95);}

\foreach \i in {1,2,4,5,6,8,9,11,12} {
  \fill[cgood!13] (w\i) -- ++(-0.25,-0.9) -- ++(0.5,0) -- cycle;
  \draw[cgood!58, line width=0.7pt] (w\i) -- ++(-0.25,-0.9) -- ++(0.5,0) -- cycle;
}
\foreach \i in {3,7,10} {
  \fill[csub!13] (w\i) -- ++(-0.25,-0.9) -- ++(0.5,0) -- cycle;
  \draw[csub!58, line width=0.7pt] (w\i) -- ++(-0.25,-0.9) -- ++(0.5,0) -- cycle;
}

\foreach \j in {1,2,3} {\draw[tedge] (u) -- (a\j);}
\foreach \j/\a in {1/1, 2/1, 3/2, 4/2, 5/3, 6/3} {\draw[tedge] (a\a) -- (b\j);}
\foreach \i/\b in {1/1, 2/1, 3/2, 4/2, 5/3, 6/3, 7/4, 8/4, 9/5, 10/5, 11/6, 12/6} {\draw[tedge] (b\b) -- (w\i);}

\coordinate (p1) at (-2.1,-2.55);
\coordinate (p2) at (0.2,-2.77);
\coordinate (p3) at (2.18,-2.3);
\draw[tedge] (w3) -- (p1);
\draw[tedge] (w7) -- (p2);
\draw[tedge] (w10) -- (p3);

\foreach \j in {1,2,3} {\node[vert] at (a\j) {};}
\foreach \j in {1,...,6} {\node[vert] at (b\j) {};}
\foreach \i in {1,2,4,5,6,8,9,11,12} {\node[gvert] at (w\i) {};}
\foreach \i in {3,7,10} {\node[avert] at (w\i) {};}
\node[bvert] at (u) {$u$};
\foreach \i in {1,2,3} {\node[ivert] at (p\i) {};}

\draw[decorate, decoration={brace, amplitude=4pt}, gray!60!black]
     (-3.75,-1.95) -- (-3.75,0) node[midway, left=4pt, font=\scriptsize] {$t$};
\draw[decorate, decoration={brace, amplitude=4pt}, gray!60!black]
     (-3.75,-2.85) -- (-3.75,-1.95) node[midway, left=4pt, font=\scriptsize] {$c-t$};

\begin{scope}[shift={(4.15,-0.75)}]
  \draw[gray!60, rounded corners=2pt, line width=0.5pt] (-0.25,0.28) rectangle (1.95,-1.18);
  \node[ivert] at (0,0) {};     \node[right=3pt, font=\scriptsize] at (0,0) {pivot of $u$};
  \node[avert] at (0,-0.45) {}; \node[right=3pt, font=\scriptsize] at (0,-0.45) {aligned};
  \node[gvert] at (0,-0.9) {};  \node[right=3pt, font=\scriptsize] at (0,-0.9) {non-aligned};
\end{scope}

\end{tikzpicture}
\caption{Each pivot of $u$ is aligned with vertices $v$ in at most one subtree (\Cref{cl:piv-good}).}
\label{fig:aligned-tree}
\end{figure}

\begin{claim}\label{cl:piv-good}
There are at least $\tfrac14\,n\,d(d-1)^{c-1}$ non-aligned pairs.
\end{claim}
\begin{proof}
By \Cref{lem:tree}, there are $N=\tfrac12\,n\,d(d-1)^{c-1}$ pairs at distance $c$.
Fix a vertex $u$ and a pivot $p_u\in\mathcal{P}(u)$ with $\dG(u,p_u)\ge t$, and let $w$ be the vertex of $P(u,p_u)$ at distance $t$ from $u$. If $p_u$ is aligned with $v$, then $w\in P(u,v)$ and $\dG(w,v)=c-t$, so by \Cref{lem:tree} there are at most $(d-1)^{c-t}$ such vertices $v$ (the vertices at distance $c-t$ from $w$ in the subtree hanging off $w$ away from $u$). See \Cref{fig:aligned-tree}.
Summing over all choices of an endpoint and one of its at most $r-1$ pivots other than itself, and using \eqref{eq:piv-t}, the number of pairs with an aligned pivot at either endpoint is at most
\[
  n(r-1)(d-1)^{c-t}\le\frac{n(r-1)(d-1)^{c-1}}{4r}<\tfrac14\,n\,(d-1)^{c-1}\le\frac N2.
\]
Hence at least $N/2=\tfrac14\,n\,d(d-1)^{c-1}$ pairs are non-aligned, as required.
\end{proof}

First, we charge every non-aligned pair to a hop. Let $\{u,v\}$ be a non-aligned pair. 
By \Cref{def:piv-ado}, $d_{H,\mathcal{P}}^{\le1}(u,v)$ is realized by a walk $u\to p_u\to p_v\to v$ in $H$, where $p_u\in\mathcal{P}(u)$ and $p_v\in\mathcal{P}(v)$. Since the weights of $H$ never underestimate distances and the oracle has stretch $k$, we have:
\begin{equation}\label{eq:piv-walk}
  \dG(u,p_u)+\dG(p_u,p_v)+\dG(p_v,v)\le d_{H,\mathcal{P}}^{\le1}(u,v)\le kc.
\end{equation}
Hence the edge set $T=P(u,p_u)\cup P(p_u,p_v)\cup P(p_v,v)\cup P(v,u)$ has at most $kc+c=(k+1)c<g$ edges. Since $G$ has girth $g$, $T$ contains no cycle, and since it is connected, it is a tree. Therefore, each of the four paths is the unique path in $T$ between its endpoints.
 In particular, $P(u,p_u)$ and $P(u,v)$ share a prefix and are disjoint after it. Let $z_u$ be the last vertex of this prefix, and define $z_v$ similarly for $P(v,p_v)$ and $P(v,u)$. Since $\{u,v\}$ is non-aligned, $\dG(u,z_u)\le t-1$ and $\dG(v,z_v)\le t-1$.
Let $q_u,q_v\in P(u,v)$ be the vertices with $\dG(u,q_u)=t$ and $\dG(q_v,v)=t$, and let $Q=P(q_u,q_v)$ be the subpath of $P(u,v)$ between them. As $c\ge 2t+1$, we have $|Q|=c-2t\ge1$.

\begin{claim}\label{cl:piv-cover}
$Q\subseteq P(p_u,p_v)$.
\end{claim}
\begin{proof}
Since $\dG(u,z_u)<t$ and $\dG(v,z_v)<t$, the vertex $z_u$ precedes $q_u$, and $z_v$ follows $q_v$ on $P(u,v)$. By the choice of $z_u$ and $z_v$, the subpath of $P(u,p_u)$ from $z_u$ to $p_u$ meets $P(u,v)$ only at $z_u$, and the subpath of $P(v,p_v)$ from $z_v$ to $p_v$ meets $P(u,v)$ only at $z_v$. These two subpaths are disjoint, as otherwise, together with the subpath of $P(u,v)$ between $z_u$ and $z_v$, they would close a cycle in $T$. Hence the walk that goes from $p_u$ to $z_u$, then along $P(u,v)$ to $z_v$, and then to $p_v$, is a simple path in $T$, and thus it is $P(p_u,p_v)$. This path contains $Q$, as required.
\end{proof}

We charge the non-aligned pair $\{u,v\}$ to the hop $(p_u,p_v)$. See \Cref{fig:pivot}\textbf{(b)}.

Second, we bound the number of non-aligned pairs that charge a single hop.

\begin{claim}\label{cl:piv-count}
Each hop $(p_u,p_v)\in E_H$ is charged by at most $2(kc+1)(d-1)^{2t}$ non-aligned pairs.
\end{claim}
\begin{proof}
Fix the hop $(p_u,p_v)$. We count the non-aligned pairs $\{u,v\}$ that charge it.
We first bound the number of choices for $q_u$ and $q_v$, and then, given $q_u$ and $q_v$, the number of choices for $u$ and $v$.

By \Cref{cl:piv-cover}, $P(q_u,q_v)=Q\subseteq P(p_u,p_v)$, so $q_u$ and $q_v$ are two vertices of $P(p_u,p_v)$. 
By \eqref{eq:piv-walk}, $P(p_u,p_v)$ has at most $kc$ edges, so there are at most $kc+1$ choices for $q_u$, and given $q_u$, at most two choices for $q_v$ (one in each direction along $P(p_u,p_v)$).

Now fix $q_u$ and $q_v$. Since $q_u\in P(q_v,u)$ and $\dG(q_u,u)=t$, by \Cref{lem:tree} there are at most $(d-1)^{t}$ choices for $u$ (the vertices at distance $t$ from $q_u$ in the subtree hanging off $q_u$ away from $q_v$). Symmetrically, there are at most $(d-1)^t$ choices for $v$.
Hence $(p_u,p_v)$ is charged at most
\[
  \underbrace{2(kc+1)}_{\text{choices of $q_u,q_v$}}\cdot\underbrace{(d-1)^{t}}_{\text{choices of $u$}}\cdot\underbrace{(d-1)^{t}}_{\text{choices of $v$}}
  \;=\;2(kc+1)(d-1)^{2t}
\]
times, as required.
\end{proof}

Third, by \Cref{cl:piv-good}, there are at least $\tfrac14\,n\,d(d-1)^{c-1}$ non-aligned pairs, each charged to one of the $|E_H|$ hops, and by \Cref{cl:piv-count} each hop is charged at most $2(kc+1)(d-1)^{2t}$ times. Since $kc<g=O(\log n)$, we get:
\begin{align*}
  2(kc+1)(d-1)^{2t}\,|E_H| &\ge \tfrac14\,n\,d(d-1)^{c-1} = \Omega\big(n(d-1)^{c}\big)
  \\|E_H| &= \Omega\!\left(\frac{n(d-1)^{c}}{(kc+1)(d-1)^{2t}}\right)
  \stackrel{(\ref{eq:piv-t})}{=}
  \Omega\!\left(\frac{n(d-1)^{c}}{(kc+1)\,r^2}\right)
  \\&\stackrel{(\ref{eq:piv-c})}{=}
  \Omega\!\left(\frac{n^{1+\frac{2}{k+1}}}{r^2\log n}\right),
\end{align*}
as required.
\end{proof}

\subsection{General Hopbound}\label{sec:pivot-beta}
In this subsection, we generalize \Cref{thm:lb-pivot} to pivot-based $\beta$-hop ADOs. A pivot-based $\beta$-hop ADO estimates $\dG(u,v)$ by the weight of a walk
\[
  u\to p_u\to p_{x_1}\to x_1\to p'_{x_1}\to\dots\to p_{x_{j-1}}\to x_{j-1}\to p'_{x_{j-1}}\to p_v\to v
\]
with $j\le\beta$ stored hops $p_u\to p_{x_1}$, $p'_{x_1}\to p_{x_2}$, \dots, $p'_{x_{j-1}}\to p_v$, where $p_{x_i}$ and $p'_{x_i}$ are pivots of an intermediate vertex $x_i$. See \Cref{fig:pivot-model}\textbf{(b)}. All the other steps are pivot jumps, and they do not count as hops. Formally:

\begin{definition}[$r$-pivot $\beta$-hop $k$-stretch ADO]\label{def:piv-ado-beta}
Let $r,\beta\ge1$ be integers, and let $k\ge1$. An \emph{$r$-pivot $\beta$-hop $k$-stretch ADO} for $G$ consists of a weighted undirected graph $H=(V,E_H,w_H)$ with $w_H(x,y)\ge\dG(x,y)$ for every $(x,y)\in E_H$, and a set $\mathcal{P}(v)\subseteq V$ of at most $r$ \emph{pivots} for every $v\in V$, with $v\in\mathcal{P}(v)$, such that for all $u,v\in V$:
\[
  \dG(u,v)\;\le\;d_{H,\mathcal{P}}^{\le\beta}(u,v)\;\le\;k\,\dG(u,v),
\]
where $d_{H,\mathcal{P}}^{\le\beta}(u,v)$ is the minimum of
\begin{equation}\label{eq:pivot-estimate-beta}
  w_H(u,p_u)+\sum_{i=1}^{j}w_H\bigl(p'_{x_{i-1}},p_{x_i}\bigr)+\sum_{i=1}^{j-1}\Bigl(w_H\bigl(p_{x_i},x_i\bigr)+w_H\bigl(x_i,p'_{x_i}\bigr)\Bigr)+w_H(p_v,v)
\end{equation}
over all $1\le j\le\beta$, all vertices $x_1,\dots,x_{j-1}\in V$, and all pivots $p_u\in\mathcal{P}(u)$, $p_v\in\mathcal{P}(v)$, and $p_{x_i},p'_{x_i}\in\mathcal{P}(x_i)$, where $p'_{x_0}=p_u$ and $p_{x_j}=p_v$. Here $w_H(a,a)=0$, and $w_H(x,y)=\infty$ if $x\ne y$ and $(x,y)\notin E_H$. The \emph{space} of the oracle is $|E_H|$.
\end{definition}

Every walk with at most $\beta$ hops makes at most $2\beta$ pivot jumps. For $\beta=1$, \Cref{def:piv-ado-beta} coincides with \Cref{def:piv-ado}.

\thmlbpivotbeta*

For $\beta=1$, \Cref{thm:lb-pivot} gives a stronger bound, and for $r=1$, \Cref{thm:lb-pivot-beta} recovers \Cref{lem:lb-beta} up to polylogarithmic factors.

The proof follows that of \Cref{lem:lb-beta}. We discard the few pairs $\{u,v\}$ for which some pivot jump covers a long segment of $P(u,v)$. Since there are at most $nr$ pivot jumps, few pairs are discarded. For every remaining pair, the at most $2\beta$ pivot jumps of the walk cover only a few edges of $P(u,v)$, so one of its at most $\beta$ hops covers almost a $1/\beta$ fraction of $P(u,v)$, and we charge the pair to that hop together with a window, as in the proof of \Cref{lem:lb-beta}. A pivot jump at an intermediate vertex covers fewer than $\rho=O(\log\log n+\log(r+1))$ edges of $P(u,v)$, and, as in the proof of \Cref{thm:lb-pivot}, a pivot jump at an endpoint covers fewer than $t=O(\log(r+1))$ edges. For $\beta=2$, every hop touches a pivot of a query endpoint, so its window starts close to that endpoint, and this saves another logarithmic factor.

\colorlet{cpiv}{green!50!black}      
\colorlet{chop}{blue!65!black}       
\colorlet{clong}{orange!85!black}    
\colorlet{cwin}{red!75!black}        
\colorlet{csub}{orange!82!black}     
\colorlet{cpath}{black}              

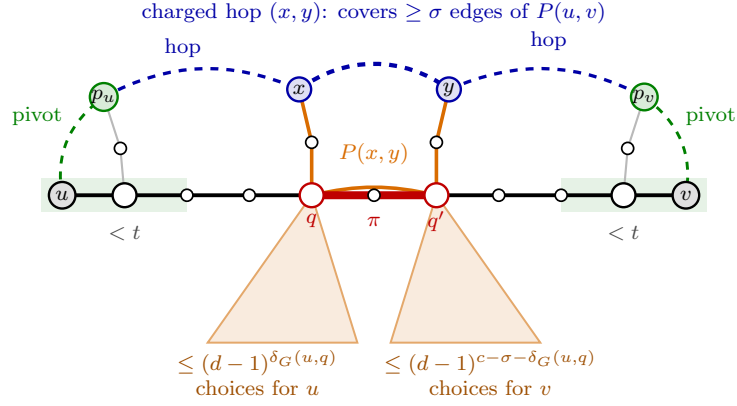
\begin{figure}[t]
\centering
\begin{tikzpicture}[
    xscale=1.1,
    vert/.style  = {circle, draw=black, fill=white, minimum size=4.5pt, inner sep=0pt, line width=0.65pt},
    bvert/.style = {circle, draw=black, fill=black!12, minimum size=10pt, inner sep=0pt, line width=0.9pt, font=\scriptsize},
    pvert/.style = {circle, draw=black, fill=white, minimum size=9pt, inner sep=0pt, line width=0.9pt, font=\scriptsize},
    ivert/.style = {circle, draw=cpiv, fill=cpiv!15, minimum size=9pt, inner sep=0pt, line width=0.9pt, font=\scriptsize},
    hvert/.style = {circle, draw=chop, fill=chop!12, minimum size=9pt, inner sep=0pt, line width=0.9pt, font=\scriptsize},
    tedge/.style = {gray!55, line width=0.8pt},
    pedge/.style = {cpath, line width=1.45pt},
    pivedge/.style = {cpiv, dashed, line width=1.1pt},
    hopedge/.style = {chop, dashed, line width=1.25pt},
]

\foreach \i in {0,...,10} {\coordinate (n\i) at (0.75*\i,0);}

\coordinate (ma) at (0.70,0.62);  \coordinate (aa) at (0.50,1.30);
\coordinate (mx) at (3.00,0.70);  \coordinate (xx) at (2.85,1.40);
\coordinate (my) at (4.50,0.70);  \coordinate (yy) at (4.65,1.40);
\coordinate (mb) at (6.80,0.62);  \coordinate (bb) at (7.00,1.30);

\fill[cpiv!10] (-0.25,-0.22) rectangle (1.5,0.22);
\fill[cpiv!10] (6.0,-0.22) rectangle (7.75,0.22);

\fill[csub!13] (n4) -- (1.75,-1.95) -- (3.55,-1.95) -- cycle;
\draw[csub!58, line width=0.75pt] (n4) -- (1.75,-1.95) -- (3.55,-1.95) -- cycle;
\fill[csub!13] (n6) -- (3.95,-1.95) -- (5.75,-1.95) -- cycle;
\draw[csub!58, line width=0.75pt] (n6) -- (3.95,-1.95) -- (5.75,-1.95) -- cycle;

\draw[pedge] (n0) -- (n10);
\draw[tedge] (n1) -- (ma) -- (aa);
\draw[tedge] (n4) -- (mx) -- (xx);
\draw[tedge] (n6) -- (my) -- (yy);
\draw[tedge] (n9) -- (mb) -- (bb);

\draw[clong, line width=1.35pt] (xx) -- (mx) -- (n4) to[bend left=14] (n6) -- (my) -- (yy);

\draw[cwin, line width=3pt, line cap=round] (n4) -- (n6);

\draw[hopedge] (aa) to[bend left=28] (xx);
\draw[hopedge, line width=1.6pt] (xx) to[bend left=40] (yy);
\draw[hopedge] (yy) to[bend left=28] (bb);
\draw[pivedge] (n0) to[bend left=30] (aa);
\draw[pivedge] (bb) to[bend left=30] (n10);

\foreach \i in {2,3,5,7,8} {\node[vert] at (n\i) {};}
\node[vert] at (ma) {}; \node[vert] at (mx) {};
\node[vert] at (my) {}; \node[vert] at (mb) {};
\node[bvert] at (n0) {$u$};
\node[bvert] at (n10) {$v$};
\node[pvert] at (n1) {};
\node[pvert] at (n9) {};
\node[pvert, draw=cwin] at (n4) {};
\node[pvert, draw=cwin] at (n6) {};
\node[ivert] at (aa) {$p_u$};
\node[ivert] at (bb) {$p_v$};
\node[hvert] at (xx) {$x$};
\node[hvert] at (yy) {$y$};

\node[gray!60!black, font=\scriptsize] at (0.75,-0.50) {$<t$};
\node[gray!60!black, font=\scriptsize] at (6.75,-0.50) {$<t$};
\node[cwin, font=\scriptsize] at (3.00,-0.34) {$q$};
\node[cwin, font=\scriptsize] at (4.50,-0.34) {$q'$};
\node[cwin, font=\scriptsize] at (3.75,-0.34) {$\pi$};
\node[clong, font=\scriptsize] at (3.75,0.55) {$P(x,y)$};
\node[chop, font=\scriptsize, align=center] at (3.75,2.42)
     {charged hop $(x,y)$: covers $\ge\sigma$ edges of $P(u,v)$};
\node[chop, font=\scriptsize] at (1.45,1.95) {hop};
\node[chop, font=\scriptsize] at (5.85,2.05) {hop};
\node[cpiv, font=\scriptsize] at (-0.30,1.05) {pivot};
\node[cpiv, font=\scriptsize] at (7.80,1.05) {pivot};

\node[csub!75!black, font=\scriptsize, align=center] at (2.35,-2.35)
     {$\le(d-1)^{\dG(u,q)}$\\choices for $u$};
\node[csub!75!black, font=\scriptsize, align=center] at (5.15,-2.35)
     {$\le(d-1)^{c-\sigma-\dG(u,q)}$\\choices for $v$};

\end{tikzpicture}
\caption{%
The charging argument of \Cref{thm:lb-pivot-beta} with $\beta=3$: the hop $(x,y)$ covers a window $\pi$ of $\sigma$ edges of $P(u,v)$.
}
\label{fig:pivot-beta}
\end{figure}

\begin{proof}
\claimsof{thm:lb-pivot-beta}%
Let $G$ be the graph guaranteed by \Cref{conj:moore}. Thus $G$ has constant degree $d$, girth $g\ge 2\log_{d-1}(n)-C$, and $m=\Theta(n)$ edges. Set $c=\floor{\frac{g}{k+1}}-1$, $t=1+\ceil{\log_{d-1}(4r)}$, $\Lambda=8r(kc+1)(c+1)$, $\rho=1+\ceil{\log_{d-1}\Lambda}$, and $\sigma=\ceil{\frac{c-2(t-1)-2(\beta-1)(\rho-1)}{\beta}}$.
Then $(k+1)c\le g-(k+1)<g$. In particular $c\le\frac g2-1$, so $c\le\floor{g/2}-1$, and by \Cref{lem:tree} shortest paths between vertices at distance at most $c$ are unique. Note that:
\begin{align}
    c&\ge\frac{g}{k+1}-2\ge\frac{2\log_{d-1}(n)-C}{k+1}-2, \qquad\text{and}\qquad kc+1\le g=O(\log n), \label{eq:pivb-c}\\
    4r&\le(d-1)^{t-1}\qquad\text{and}\qquad (d-1)^{t}\le 4r(d-1)^{2},\label{eq:pivb-t}\\
    \Lambda &\le(d-1)^{\rho-1}\le \Lambda(d-1), \qquad\text{and}\qquad \Lambda=O\big(r\log^2 n\big), \label{eq:pivb-r}\\
    \sigma&\ge\frac{c}{\beta}-\frac{2(t-1)}{\beta}-\frac{2(\beta-1)(\rho-1)}{\beta},\label{eq:pivb-sigma}
\end{align}
where the bound on $\Lambda$ holds since $c+1\le kc+1\le g=O(\log n)$ by \eqref{eq:pivb-c}.
Let $(H,\mathcal{P})$ be an $r$-pivot $\beta$-hop $k$-stretch ADO for $G$. We show that
$|E_H|=\Omega\big(n^{1+\frac{2}{\beta(k+1)}}/(r^2\log^{6-4/\beta} n)\big)$, and that for $\beta=2$, $|E_H|=\Omega\big(n^{1+\frac{1}{k+1}}/(r^2\log^3 n\log(r+1))\big)$. Since $m=\Theta(n)$, this implies the theorem.
We may assume $\sigma\ge1$. Otherwise $c\le2(t-1)+2(\beta-1)(\rho-1)$, and \eqref{eq:pivb-c}, \eqref{eq:pivb-t}, and \eqref{eq:pivb-r} give
\[
 n^{2/(\beta(k+1))}=O((d-1)^{c/\beta})=O(r^2\log^{4-4/\beta}n).
\]
The general bound, and the improved bound when $\beta=2$, are then $O(n/\log n)$ and follow from $|E_H|\ge n/2$.

The proof contains the same preliminary step and three steps as the proof of \Cref{thm:lb-pivot}.
In the preliminary step, we discard the few pairs for which a single pivot jump covers a long segment of $P(u,v)$. We denote the pivot jump between a vertex $x$ and one of its pivots $a\in\mathcal{P}(x)\setminus\{x\}$ by $(x,a)$. There are at most $n(r-1)$ pivot jumps. Let $\{u,v\}$ be a pair with $\dG(u,v)=c$. We say that a pivot $p_u\in\mathcal{P}(u)$ is \emph{aligned with $v$} if $P(u,p_u)$ and $P(u,v)$ share their first $t$ edges. We say that the pair $\{u,v\}$ is \emph{bad} if there is a pivot jump $(x,a)$ with $\dG(x,a)\le kc$ such that $P(x,a)$ and $P(u,v)$ share at least $\rho$ edges, or if a pivot of $u$ is aligned with $v$, or a pivot of $v$ is aligned with $u$. Otherwise, the pair is \emph{good}. Note that for a pivot jump $(x,a)$ with $\dG(x,a)\le kc$, the edge set $P(x,a)\cup P(u,v)$ has at most $kc+c<g$ edges, and thus contains no cycle, so the edges shared by the two paths form a contiguous segment of both.

\begin{claim}\label{cl:pivb-good}
There are at least $\tfrac14\,n\,d(d-1)^{c-1}$ good pairs.
\end{claim}
\begin{proof}
By \Cref{lem:tree}, there are $N=\tfrac12\,n\,d(d-1)^{c-1}$ pairs at distance $c$.
Fix a pivot jump $(x,a)$ with $\dG(x,a)\le kc$. We bound the number of pairs $\{u,v\}$ at distance $c$ for which $P(u,v)$ and $P(x,a)$ share at least $\rho$ edges. The first $\rho$ shared edges of any such pair, ordered from $u$, form a \emph{window}: a subpath of $P(x,a)$ with $\rho$ edges, together with a direction. Since $P(x,a)$ has at most $kc$ edges, a window is determined by its starting offset and its direction, so $P(x,a)$ has at most $2(kc+1)$ windows. Fix a window, and let $q$ and $q'$ be its first and last vertices. If it is contained in $P(u,v)$, in this direction, then $P(u,v)=P(u,q)\cup P(q,q')\cup P(q',v)$ and $\dG(u,q)+\rho+\dG(q',v)=c$. By \Cref{lem:tree}, for each value of $\dG(u,q)$ there are at most $(d-1)^{\dG(u,q)}$ choices for $u$ and at most $(d-1)^{c-\rho-\dG(u,q)}$ choices for $v$. Since $\dG(u,q)$ ranges over $\{0,\dots,c-\rho\}$, each window is contained in $P(u,v)$, in this direction, for at most $(c+1)(d-1)^{c-\rho}$ pairs $\{u,v\}$.
Since there are at most $n(r-1)$ pivot jumps, by \eqref{eq:pivb-r} the number of such pairs is at most
\begin{align*}
  n(r-1)\cdot 2(kc+1)\cdot(c+1)(d-1)^{c-\rho}&\le\frac{2n(r-1)(kc+1)(c+1)(d-1)^{c-1}}{\Lambda}
  \\&=\frac{n(r-1)(d-1)^{c-1}}{4r}<\tfrac14\,n\,(d-1)^{c-1}.
\end{align*}
Moreover, exactly as in the proof of \Cref{cl:piv-good}, the number of pairs with an aligned pivot at either endpoint is at most
\[
  n(r-1)(d-1)^{c-t}\le\frac{n(r-1)(d-1)^{c-1}}{4r}<\tfrac14\,n\,(d-1)^{c-1}.
\]
Thus fewer than $\tfrac12\,n\,(d-1)^{c-1}\le N/2$ pairs are bad, and at least $N/2=\tfrac14\,n\,d(d-1)^{c-1}$ pairs are good, as required.
\end{proof}

First, we charge every good pair to a hop. Let $\{u,v\}$ be a good pair.
By \Cref{def:piv-ado-beta}, $d_{H,\mathcal{P}}^{\le\beta}(u,v)$ is realized by a walk with at most $\beta$ hops, and hence with at most $2\beta$ pivot jumps. Since the weights of $H$ never underestimate distances and the oracle has stretch $k$, we have:
\begin{equation}\label{eq:pivb-walk}
  \sum_{(y,z)\text{ a step of the walk}}\dG(y,z)\;\le\;d_{H,\mathcal{P}}^{\le\beta}(u,v)\;\le\;kc.
\end{equation}
Replacing every step $(y,z)$ of this walk by the shortest path $P(y,z)$ in $G$, in order, gives a walk $W$ in $G$ from $u$ to $v$. By \eqref{eq:pivb-walk}, $W$ has at most $kc$ edges, so the edge set $T=W\cup P(u,v)$ has at most $kc+c=(k+1)c<g$ edges. Hence $T$ contains no cycle, and since it is connected, it is a tree. Thus $P(u,v)$ is the unique simple $u$--$v$ path in $T$, and the walk $W$ from $u$ to $v$ in $T$ traverses every edge of it. That is, $P(u,v)\subseteq W$, and the walk covers all $c$ edges of $P(u,v)$. Moreover, the shortest path of each step is a path in $T$, so the edges of $P(u,v)$ that it covers form a contiguous (possibly empty) segment of $P(u,v)$.

By \eqref{eq:pivb-walk}, every pivot jump $(y,z)$ of the walk has $\dG(y,z)\le kc$, and since $\{u,v\}$ is good, it covers at most $\rho-1$ edges of $P(u,v)$. The two pivot jumps at the endpoints cover even fewer edges: $P(u,p_u)$ and $P(u,v)$ are paths of the tree $T$ that start at $u$, so they share a prefix and are disjoint after it, and since $p_u$ is not aligned with $v$, this prefix has at most $t-1$ edges. The same holds for the jump at $v$. Hence the pivot jumps of the walk cover at most $2(t-1)+2(\beta-1)(\rho-1)$ edges of $P(u,v)$ in total.

The hops of the walk cover the remaining edges of $P(u,v)$, of which there are at least $c-2(t-1)-2(\beta-1)(\rho-1)$. Since the walk has at most $\beta$ hops, some hop $(x,y)$ covers a contiguous segment of $P(u,v)$ of length at least $\frac{c-2(t-1)-2(\beta-1)(\rho-1)}{\beta}$, and hence, as this length is an integer, at least $\sigma$. If $\beta=2$, then every hop of the walk has an endpoint that is the pivot $p_u$ or the pivot $p_v$ used by the walk. In this case, orienting the hop and swapping the names of $u$ and $v$ if necessary, we may assume that $x=p_u$.

Let $\pi$ be the first $\sigma$ edges of that segment, ordered from $u$, so that $|\pi|=\sigma$ and $\pi\subseteq P(u,v)\cap P(x,y)$.
We charge the good pair $\{u,v\}$ to the hop $(x,y)$ through the window $\pi$. See \Cref{fig:pivot-beta}.

Second, we bound the number of good pairs that charge a single hop through a single window.

\begin{claim}\label{cl:pivb-count}
Each hop $(x,y)$ is charged through at most $2(kc+1)$ windows $\pi$, and through each window $\pi$ it is charged by at most $(c+1)(d-1)^{c-\sigma}$ good pairs. Moreover, if $\beta=2$, then each hop is charged through each window by at most $t(d-1)^{c-\sigma}$ good pairs.
\end{claim}
\begin{proof}
By \eqref{eq:pivb-walk}, every charged hop $(x,y)$ has $\dG(x,y)\le kc$, and a window charged to it is a subpath of $P(x,y)$ with $\sigma$ edges, together with a direction. Such a window is determined by its starting offset and its direction, which gives at most $2(kc+1)$ windows.

For the second part, fix the hop $(x,y)$ and the window $\pi$, and let $q$ and $q'$ be the first and the last vertices of $\pi$. If a good pair $\{u,v\}$ is charged to $(x,y)$ through $\pi$, then $P(u,v)=P(u,q)\cup\pi\cup P(q',v)$ and $\dG(u,q)+\sigma+\dG(q',v)=c$. By \Cref{lem:tree}, for each value of $\dG(u,q)$ there are at most $(d-1)^{\dG(u,q)}$ choices for $u$ and at most $(d-1)^{c-\sigma-\dG(u,q)}$ choices for $v$. See \Cref{fig:pivot-beta}.
Since $\dG(u,q)$ ranges over $\{0,\dots,c-\sigma\}$, the hop $(x,y)$ is charged through $\pi$ at most $(c-\sigma+1)(d-1)^{c-\sigma}\le(c+1)(d-1)^{c-\sigma}$ times, as required.

Finally, let $\beta=2$, so that $x=p_u$. Let $z_u$ be the last vertex of the common prefix of $P(u,p_u)$ and $P(u,v)$. As $p_u$ is not aligned with $v$, we have $\dG(u,z_u)\le t-1$. In the tree $T$, every path from $p_u$ to a vertex of $P(u,v)$ passes through $z_u$, so the segment of $P(u,v)$ that is covered by $P(p_u,y)$ has $z_u$ as one of its endpoints. Hence this segment, and with it $\pi$, starts at distance at most $t-1$ from $u$. Thus $\dG(u,q)$ ranges only over $\{0,\dots,t-1\}$, and the hop $(x,y)$ is charged through $\pi$ at most $t(d-1)^{c-\sigma}$ times, as required.
\end{proof}

Third, by \Cref{cl:pivb-good}, there are at least $\tfrac14\,n\,d(d-1)^{c-1}$ good pairs, each charged to one of the $|E_H|$ hops through one of at most $2(kc+1)$ windows, and by \Cref{cl:pivb-count} each hop is charged through each window at most $(c+1)(d-1)^{c-\sigma}$ times. Since $kc<g=O(\log n)$, both $kc+1$ and $c+1$ are $O(\log n)$. Moreover, by \eqref{eq:pivb-t} and \eqref{eq:pivb-r}, $(d-1)^{t-1}=O(r)$ and $(d-1)^{\rho-1}=O(r\log^2 n)$, and we get:
\begin{align*}
  2(kc+1)(c+1)(d-1)^{c-\sigma}\,|E_H| &\ge \tfrac14\,n\,d(d-1)^{c-1} = \Omega\big(n(d-1)^{c}\big)
  \\|E_H| &= \Omega\!\left(\frac{n(d-1)^{\sigma}}{(kc+1)(c+1)}\right)
  \\&\stackrel{(\ref{eq:pivb-sigma})}{=}
  \Omega\!\left(\frac{n(d-1)^{c/\beta}}{(kc+1)(c+1)\,(d-1)^{\frac{2(t-1)}{\beta}}\,(d-1)^{\frac{2(\beta-1)(\rho-1)}{\beta}}}\right)
  \\&=
  \Omega\!\left(\frac{n(d-1)^{c/\beta}}{r^2\log^{6-4/\beta} n}\right)
  \stackrel{(\ref{eq:pivb-c})}{=}
  \Omega\!\left(\frac{n^{1+\frac{2}{\beta(k+1)}}}{r^2\log^{6-4/\beta} n}\right),
\end{align*}
as required. For $\beta=2$, each hop is charged through each window at most $t(d-1)^{c-\sigma}$ times (\Cref{cl:pivb-count}), where $t=O(\log(r+1))$, and the same calculation gives:
\[
  |E_H| = \Omega\!\left(\frac{n(d-1)^{\sigma}}{(kc+1)\,t}\right)
  \stackrel{(\ref{eq:pivb-sigma})}{=}
  \Omega\!\left(\frac{n(d-1)^{c/2}}{(kc+1)\,t\,(d-1)^{t-1}(d-1)^{\rho-1}}\right)
  =
  \Omega\!\left(\frac{n^{1+\frac{1}{k+1}}}{r^2\log^{3} n\,\log(r+1)}\right),
\]
as required.
\end{proof}

\section{Stronger Lower Bounds for $(k,\beta)$-Hopsets}
\label{sec:dense-transfer}

In this section, we prove lower bounds for $(k,\beta)$-hopsets that are stronger than those of \Cref{lem:lb,lem:lb-beta}, on sparse graphs that are not of constant degree. We use the Erd\H{o}s girth conjecture, which states that for every integer $t\ge1$ there are graphs with girth at least $2t+2$ and $\Omega(n^{1+1/t})$ edges, to obtain a dense base graph, which we transform into a sparse graph by subdivision and the addition of private leaves.
Unlike the arguments of \Cref{sec:robust-hopsets}, this technique does not extend to pivot-based ADOs: already two pivots per vertex allow all the selected demands to be answered in linear space (\Cref{rem:transfer-pivots}). We first prove the two-hop bound, which is tight by the conversion in \Cref{sec:models}, and then extend it to a general hopbound.

\subsection{Two Hops}\label{sec:dense-two-hop}

The idea is as follows. We write $N,M$ for the numbers of vertices and edges of a dense base graph $B$, and $n,m$ for those of the final graph $G$. We subdivide each base edge into a path of constant length and attach $q=\ceil{M/N}$ private leaves to every base vertex. The resulting graph has $n,m=\Theta(M)$, so it is sparse. Every base edge $ab$ gives $q^2$ demands, namely the pairs of a leaf of $a$ and a leaf of $b$. The natural exact route $u\to a\to b\to v$ of such a demand uses three hops. With only two hops, one of the hops must start at a leaf, say $u$, and enter the subdivided edge $ab$. Such a hop fixes $u$ and the base edge $ab$, so it serves at most $q$ demands, one for each leaf of $b$. Dividing the $Mq^2$ demands by this bound gives a lower bound on the size of the hopset $F$.

\thmdensetransfer*

The proof below gives the same bound for every integer $\ell\ge2$ and every real stretch $k$ with $1\le k<2\ell-1$, and the theorem is the case $k=2\ell-3$. For every real $k>1$, this bound is stronger than \Cref{lem:lb}. Indeed, the smallest integer $\ell$ with $k<2\ell-1$ satisfies $\ell\le\frac{k+3}{2}<k+1$, so the exponent $\frac1\ell$ above $1$ is larger than $\frac1{k+1}$ from \Cref{lem:lb}. For example, at stretch $2\ell-3$ with $\ell\ge3$, \Cref{lem:lb} gives the exponent $\frac1{2\ell-2}$, while the bound here gives $\frac1\ell$. For $k=2$ the bound is $\Omega(m^{3/2})$, while \Cref{lem:lb} gives $\Omega(m^{4/3})$. For $k=1$ the two bounds coincide. In particular, under the girth conjecture, a $(k,2)$-hopset of size $\Ot(m^{1+1/(k+1)})$ cannot exist for every sparse graph when $k>1$. The price is that the graphs do not have constant degree, since every base vertex has $q$ leaves.

\colorlet{cpath}{black}              

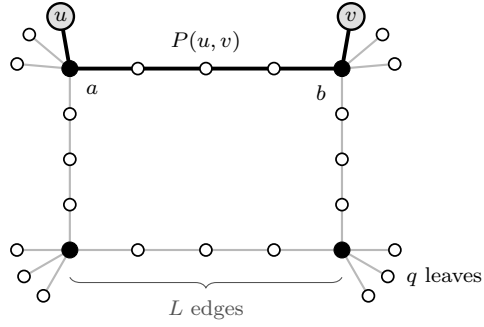
\begin{figure}[t]
\centering
\begin{tikzpicture}[
    vert/.style  = {circle, draw=black, fill=white, minimum size=4.5pt, inner sep=0pt,
                    line width=0.65pt},
    basev/.style = {circle, draw=black, fill=black, minimum size=6pt, inner sep=0pt},
    bvert/.style = {circle, draw=black, fill=black!12, minimum size=10pt, inner sep=0pt,
                    line width=0.9pt, font=\scriptsize},
    tedge/.style = {gray!55, line width=0.8pt},
    pedge/.style = {cpath, line width=1.45pt},
]

\coordinate (a) at (0,2.4);
\coordinate (b) at (3.6,2.4);
\coordinate (c) at (3.6,0);
\coordinate (d) at (0,0);

\draw[tedge] (b) -- (c) -- (d) -- (a);
\draw[pedge] (a) -- (b);
\foreach \p/\q in {a/b, b/c, c/d, d/a} {
  \foreach \j in {0.25,0.5,0.75} {\path (\p) -- (\q) node[vert, pos=\j] {};}
}

\foreach \w/\angs in {a/{140,175}, b/{5,40}, c/{-60,-30,0}, d/{180,210,240}} {
  \foreach \t in \angs {
    \draw[tedge] (\w) -- ++(\t:0.7) node[vert] {};
  }
}

\draw[pedge] (a) -- ++(100:0.7) coordinate (u);
\draw[pedge] (b) -- ++(80:0.7) coordinate (v);
\node[bvert] at (u) {$u$};
\node[bvert] at (v) {$v$};

\node[basev, label={[font=\scriptsize]below right:$a$}] at (a) {};
\node[basev, label={[font=\scriptsize]below left:$b$}] at (b) {};
\node[basev] at (c) {};
\node[basev] at (d) {};

\node[font=\scriptsize] at (1.8,2.75) {$P(u,v)$};
\draw[decorate, decoration={brace,mirror,amplitude=4pt}, gray!65!black]
  (0,-0.36) -- (3.6,-0.36)
  node[midway, below=5pt, font=\scriptsize] {$L$ edges};
\node[font=\scriptsize, anchor=west] at (4.32,-0.35) {$q$ leaves};

\end{tikzpicture}
\caption{%
The construction of \Cref{thm:dense-transfer}, with a demand $\{u,v\}$ of the base edge $ab$.%
}
\label{fig:dense-graph}
\end{figure}

\begin{proof}
\claimsof{thm:dense-transfer}%
We prove the bound for every integer $\ell\ge2$ and every real $k$ with $1\le k<2\ell-1$. Set $t=\ell-1\ge1$ and $\gamma=2\ell$. Since $k<2\ell-1$, we have $\gamma>k+1$. By the girth conjecture, there are arbitrarily large $N$-vertex
base graphs $B$ with girth at least $\gamma$ and
$M=\Omega(N^{1+1/t})$ edges. Since $\gamma>k+1$, we may choose an even constant $L\ge2$ such that
\[
  (k+1)(L+2)<\gamma L.
\]
This is the girth inequality needed for the tree argument below.

\paragraph{The graph.}
Let $q=\ceil{M/N}$. We obtain $G$ from $B$ by replacing each base edge with a path of $L$ edges through new vertices, and by attaching $q$ private leaves to each base vertex (see \Cref{fig:dense-graph}). Since $N\le M$ and $Nq\le M+N$, the graph $G$ has
\[
  n=N+(L-1)M+Nq=\Theta(M),\qquad m=LM+Nq=\Theta(M).
\]
The leaves lie on no cycle, and every cycle of $G$ is a subdivided cycle of $B$. Hence $G$ has girth at least $\gamma L$, and it is bipartite since $L$ is even. It is $2$-degenerate since the base vertices are independent and every other vertex has degree at most $2$.

\paragraph{The demands.}
For every base edge $ab$, we take the $q^2$ pairs $\{u,v\}$ where $u$ is a leaf of $a$ and $v$ is a leaf of $b$. Every other $a$--$b$ path in $G$ is a subdivided $a$--$b$ path of $B$ with at least $\gamma-1\ge3$ base edges, so its length is at least $(\gamma-1)L>L$. Hence $\dG(u,v)=L+2$, and $P(u,v)$ consists of the leaf edge $ua$, the subdivided edge $ab$, and the leaf edge $bv$. Let $c=L+2$, so that $(k+1)c<\gamma L$. A demand determines its base edge, so there are $Mq^2$ distinct demands.

Let $F$ be a $(k,2)$-hopset for $G$, and let $H$ be its hop graph. The rest of the proof follows the three steps of the proof of \Cref{lem:exact}.

First, we charge each demand to a hop. Let $\{u,v\}$ be a demand of the base edge $ab$. Consider a path with at most two hops in $H$ that realizes $d_H^{\le 2}(u,v)$, and let $x$ be its intermediate vertex ($x=v$ for a single hop). Since hop weights never underestimate distances, the walk formed by $P(u,x)$ and $P(x,v)$ has at most $kc$ edges.
Together with $P(u,v)$, it has at most $(k+1)c<\gamma L$ edges, so it forms a tree. Hence, as in the proof of \Cref{cl:trace}, the walk covers every edge of $P(u,v)$. Moreover, each of $P(u,x)$ and $P(x,v)$ meets $P(u,v)$ in a contiguous segment, since two paths in a tree intersect in a path.

In particular, one of the two hops covers an edge of the subdivided edge $ab$. Exchanging the names of $u,v$ and of $a,b$ if necessary, we may assume that this is the hop $(u,x)$ and that $u$ is the leaf of $a$. The segment $P(u,x)\cap P(u,v)$ contains $u$, so it is a prefix of $P(u,v)$. Since it reaches into $ab$, it contains the leaf edge $ua$ and the first edge of the subdivided edge $ab$ (see the following \Cref{fig:two-hop-transfer}). We charge the demand $\{u,v\}$ to the ordered hop $(u,x)$. Since $\dG(u,x)\ge2$, this hop is not an edge of $G$, so it is a hop of $F$.

\colorlet{chop}{blue!65!black}       
\colorlet{cpath}{black}              
\colorlet{cbranch}{orange!85!black}  
\colorlet{csub}{orange!80!black}     

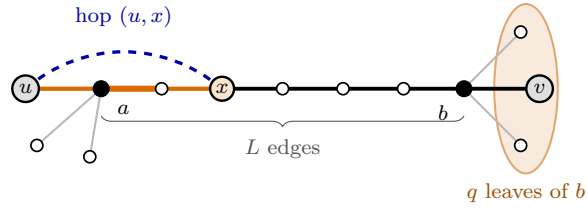
\begin{figure}[h]
\centering
\begin{tikzpicture}[
    vert/.style  = {circle, draw=black, fill=white, minimum size=4.5pt, inner sep=0pt,
                    line width=0.65pt},
    basev/.style = {circle, draw=black, fill=black, minimum size=6pt, inner sep=0pt},
    bvert/.style = {circle, draw=black, fill=black!12, minimum size=10pt, inner sep=0pt,
                    line width=0.9pt, font=\scriptsize},
    pvert/.style = {circle, draw=black, fill=white, minimum size=9pt, inner sep=0pt,
                    line width=0.9pt, font=\scriptsize},
    tedge/.style = {gray!55, line width=0.8pt},
    pedge/.style = {cpath, line width=1.45pt},
]

\foreach \i in {0,...,6} {\coordinate (p\i) at (0.8*\i,0);}

\coordinate (u)   at (-1.0,0);
\coordinate (la1) at (-0.85,-0.75);
\coordinate (la2) at (-0.15,-0.9);
\coordinate (v)   at (5.8,0);
\coordinate (lb1) at (5.55,0.75);
\coordinate (lb2) at (5.55,-0.75);

\fill[csub!13] (5.62,0) ellipse (0.42 and 1.12);
\draw[csub!58, line width=0.75pt] (5.62,0) ellipse (0.42 and 1.12);

\draw[tedge] (p0) -- (la1);  \draw[tedge] (p0) -- (la2);
\draw[tedge] (p6) -- (lb1);  \draw[tedge] (p6) -- (lb2);
\draw[pedge] (p2) -- (p6) -- (v);
\draw[cbranch, line width=1.6pt] (u) -- (p0) -- (p2);
\draw[cbranch, line width=2.3pt] (p0) -- (p1);

\draw[chop, dashed, line width=1.25pt] (u) to[bend left=40] (p2);

\foreach \i in {1,3,4,5} {\node[vert] at (p\i) {};}
\node[vert] at (la1) {}; \node[vert] at (la2) {};
\node[vert] at (lb1) {}; \node[vert] at (lb2) {};
\node[basev, label={[font=\scriptsize]below right:$a$}] at (p0) {};
\node[basev, label={[font=\scriptsize]below left:$b$}] at (p6) {};
\node[pvert, fill=cbranch!18] at (p2) {$x$};
\node[bvert] at (u) {$u$};
\node[bvert] at (v) {$v$};

\node[chop, font=\scriptsize] at (0.3,0.95) {hop $(u,x)$};
\draw[decorate, decoration={brace,mirror,amplitude=4pt}, gray!65!black]
  (0,-0.36) -- (4.8,-0.36)
  node[midway, below=5pt, font=\scriptsize] {$L$ edges};
\node[csub!75!black, font=\scriptsize] at (5.62,-1.38) {$q$ leaves of $b$};

\end{tikzpicture}
\caption{%
The first two edges of $P(u,x)$ determine the parent $a$ of $u$ and the base edge $ab$, so $v$ is one of the $q$ leaves of $b$.%
}
\label{fig:two-hop-transfer}
\end{figure}

Second, we bound the number of demands that charge a single hop.

\begin{claim}\label{cl:dense-charge}
Each ordered hop $(u,x)$ is charged by at most $q$ demands.
\end{claim}
\begin{proof}
The hop fixes the path $P(u,x)$. If a demand charges $(u,x)$, then the first edge of $P(u,x)$ is the leaf edge $ua$, and its second edge is the first edge of the subdivided edge $ab$. This edge lies on only one subdivided base edge, so the hop fixes $b$. Hence $v$ is one of the $q$ leaves of $b$.
\end{proof}

Third, each of the $Mq^2$ demands charges one of the $2|F|$ ordered hops of $F$, and each of them is charged at most $q$ times (\Cref{cl:dense-charge}). Hence
\[
  Mq^2\le2q|F|,\qquad\text{so}\qquad |F|\ge\frac{Mq}{2}\ge\frac{M^2}{2N}.
\]

It remains to express this bound in terms of $m$. Since $M=\Omega(N^{1+1/t})$, we have $N=O(M^{t/(t+1)})$.
Together with $m=\Theta(M)$, this gives
\[
  |F|=\Omega\!\left(M^{2-\frac{t}{t+1}}\right)=\Omega\!\left(m^{1+\frac1\ell}\right),
\]
as required.
\end{proof}

\begin{remark}\label{rem:degenerate}
The graphs of \Cref{thm:dense-transfer} are bipartite and $2$-degenerate, as shown in the proof, and so are the graphs of \Cref{thm:beta-dense-transfer}. Hence these lower bounds hold even for bipartite $2$-degenerate graphs, so such structural restrictions do not allow smaller hopsets.
\end{remark}

\subsection{Extension to General Hopbound}\label{sec:dense-beta-hop}
We extend \Cref{thm:dense-transfer} to every fixed hopbound
$\beta\ge2$. The demands of \Cref{thm:dense-transfer} become cheap with three hops, since the route $u\to a\to b\to v$ needs only the $M$ hops $ab$. Therefore, the parents of the two demand leaves are now $s=\beta-1$ base edges apart, along a path $a_0,a_1,\dots,a_s$ of $B$. The natural exact route $u\to a_0\to a_1\to\dots\to a_s\to v$ then uses $\beta+1$ hops, one more than allowed, just as the route $u\to a\to b\to v$ does for $\beta=2$. For $\beta=2$,
the bound below recovers \Cref{thm:dense-transfer}.

\thmbetadensetransfer*

As for \Cref{thm:dense-transfer}, the proof below gives the same bound for every integer $\ell\ge2$ and every real stretch $k$ with $1\le k<\frac{2\ell}{\beta-1}-1$, and the theorem is the case $k=\frac{2\ell-\beta-1}{\beta-1}$. Here, and throughout the proof, the hidden constants depend on $k$, $\beta$, and $\ell$, but not on $n$.
For fixed $k$ and $\beta$, this bound applies to every integer $\ell>\frac{(k+1)(\beta-1)}{2}$, and in particular to some $\ell\le\frac{(k+1)(\beta-1)+2}{2}$. Hence the exponent above $1$ is at least $\frac{2}{(k+1)(\beta-1)+2}$. For every $k>1$, this is larger than $\frac{2}{\beta(k+1)}=\frac{2}{(k+1)(\beta-1)+k+1}$ from \Cref{lem:lb-beta}. For example, for $\beta=3$ and stretch $\ell-2$ with $\ell\ge4$, \Cref{lem:lb-beta} gives the exponent $\frac{2}{3(\ell-1)}$, while the bound here gives $\frac1\ell$. For $k=1$ the two exponents are equal, but the bound here has no $\log^2 n$ loss. As in \Cref{thm:dense-transfer}, the graphs here are not of constant degree.

The idea extends the proof of \Cref{thm:dense-transfer} as in \Cref{sec:betahop}. The hops of an answer must cover the $s=\beta-1$ subdivided base edges between $u$ and $v$ (see the following \Cref{fig:dense-graph-beta}). If the first or the last hop enters one of them, then it fixes a leaf and the adjacent base edge, exactly as in the proof of \Cref{thm:dense-transfer}. Otherwise, at most $\beta-2$ internal hops cover the $\beta-1$ base edges, so one of them passes through a junction $a_i$ and fixes the two base edges around it (see \Cref{fig:beta-transfer}). This pair of base edges plays the role of the window of \Cref{sec:betahop}.

\colorlet{cpath}{black}              

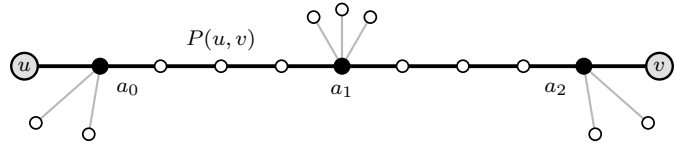
\begin{figure}[h]
\centering
\begin{tikzpicture}[
    vert/.style  = {circle, draw=black, fill=white, minimum size=4.5pt, inner sep=0pt,
                    line width=0.65pt},
    basev/.style = {circle, draw=black, fill=black, minimum size=6pt, inner sep=0pt},
    bvert/.style = {circle, draw=black, fill=black!12, minimum size=10pt, inner sep=0pt,
                    line width=0.9pt, font=\scriptsize},
    tedge/.style = {gray!55, line width=0.8pt},
    pedge/.style = {cpath, line width=1.45pt},
]

\coordinate (a0) at (0,0);
\coordinate (a1) at (3.2,0);
\coordinate (a2) at (6.4,0);
\coordinate (u)  at (-1.0,0);
\coordinate (v)  at (7.4,0);

\foreach \x/\y in {-0.85/-0.75, -0.15/-0.9} {
  \draw[tedge] (a0) -- ++(\x,\y) node[vert] {};
  \draw[tedge] (a2) -- ++(-\x,\y) node[vert] {};
}
\foreach \t in {60,90,120} {\draw[tedge] (a1) -- ++(\t:0.75) node[vert] {};}

\draw[pedge] (u) -- (a0) -- (a1) -- (a2) -- (v);
\foreach \p/\q in {a0/a1, a1/a2} {
  \foreach \j in {0.25,0.5,0.75} {\path (\p) -- (\q) node[vert, pos=\j] {};}
}

\node[basev, label={[font=\scriptsize]below right:$a_0$}] at (a0) {};
\node[basev, label={[font=\scriptsize]below:$a_1$}] at (a1) {};
\node[basev, label={[font=\scriptsize]below left:$a_2$}] at (a2) {};
\node[bvert] at (u) {$u$};
\node[bvert] at (v) {$v$};

\node[font=\scriptsize] at (1.6,0.35) {$P(u,v)$};

\end{tikzpicture}
\caption{%
A demand $\{u,v\}$ in the proof of \Cref{thm:beta-dense-transfer} with $s=2$. The parents $a_0$ and $a_2$ of $u$ and $v$ are $s$ base edges apart.%
}
\label{fig:dense-graph-beta}
\end{figure}

\begin{proof}
\claimsof{thm:beta-dense-transfer}%
We prove the bound for every integer $\ell\ge2$ and every real $k$ with $1\le k<\frac{2\ell}{\beta-1}-1$. We first make the degrees of the base graph almost equal, then define the graph and the demands, and finally charge every demand either to an endpoint hop or to an internal hop.
Set $s=\beta-1$, $t=\ell-1$, and $\gamma=2\ell$. Then $\gamma>(k+1)s$, since $k<\frac{2\ell}{s}-1$. Moreover, since $k\ge1$, we have $2\ell>(k+1)s\ge2s$, so $t\ge s\ge1$ and $\gamma\ge2s+2$. By the girth conjecture,
there are arbitrarily large $N_0$-vertex base graphs $B_0$ with
girth at least $\gamma$ and $M_0=\Omega(N_0^{1+1/t})$ edges.

\paragraph{Almost equal degrees.}
We need the base degrees to differ by at most a factor of $2$, and obtain this by deleting low-degree vertices and splitting the remaining vertices.
Let $D=\floor{M_0/(4N_0)}$, which tends to infinity, so $D\ge2$. We repeatedly delete vertices of degree less than $D$. Each deletion removes fewer than $D$ edges, so fewer than $DN_0\le M_0/4$ edges are deleted, and at least $3M_0/4$ edges remain.
Next, we split every remaining vertex $w$ into $\floor{\deg(w)/D}$ copies, and partition the edges of $w$ among its copies into groups of sizes between $D$ and $2D-1$. This is possible since $\deg(w)<(\floor{\deg(w)/D}+1)D$.
The resulting graph $B$ is simple since $B_0$ is, and all its degrees are in $[D,2D-1]$. It has $M\in[3M_0/4,M_0]$ edges, and hence $N=\Theta(M/D)=\Theta(N_0)$ vertices, since $D=\Theta(M_0/N_0)$. Its girth is still at least $\gamma$, since a cycle of $B$ maps to a closed trail of $B_0$ with distinct edges, which contains a cycle of $B_0$ that is not longer. Thus we have a base graph $B$ with
\[
  M=\Theta(ND),\qquad M=\Omega(N^{1+1/t}).
\]

\paragraph{The graph.}
Since $\gamma>(k+1)s$, we may choose an even constant $L\ge2$ such that
\[
  (k+1)(sL+2)<\gamma L.
\]
We construct $G$ from $B$ as in the proof of \Cref{thm:dense-transfer}, with $q=D$. That is, we subdivide every base edge into a path of $L$ edges and attach $D$
private leaves to every base vertex. As there, $G$ is bipartite and $2$-degenerate, it has girth at least $\gamma L$, and since $ND/2\le M\le ND$, it has
\[
 n=N+(L-1)M+ND=\Theta(M),\qquad
 m=LM+ND=\Theta(M).
\]

\paragraph{The demands.}
For every path $a_0,a_1,\dots,a_s$ with $s$ edges in $B$, we take the $D^2$ pairs $\{u,v\}$ where $u$ is a leaf of $a_0$ and $v$ is a leaf of $a_s$. The $s$ subdivided edges $a_{i-1}a_i$ are the \emph{base edges} of the demand (see \Cref{fig:dense-graph-beta}).
Every other $a_0$--$a_s$ path in $B$ has at least $\gamma-s\ge s+2$ edges, since together with $a_0,\dots,a_s$ it contains a cycle. Its subdivision therefore has length more than $sL$. Hence $\dG(u,v)=sL+2$, and $P(u,v)$ consists of the two leaf edges and the $s$ base edges. Let $c=sL+2$, so that $(k+1)c<\gamma L$.
A demand determines $a_0$ and $a_s$, and hence its base path, so the demands are distinct. Since $\gamma>s$, every walk with $s$ edges in $B$ that never returns along the edge it just used is a path. There are $N$ choices for $a_0$, at least $D$ for $a_1$, and at least $D-1$ for each of $a_2,\dots,a_s$. Every path is counted twice, once from each end. Since $D-1\ge D/2$, the number of demands is at least
\[
  \tfrac12 N D(D-1)^{s-1}\cdot D^2\ \ge\ \frac{ND^{s+2}}{2^{s}}.
\]

Let $F$ be a $(k,\beta)$-hopset for $G$, and let $H$ be its hop graph. As before, the rest of the proof has three steps.

First, we charge each demand to a hop.
Let $\{u,v\}$ be a demand with base path $a_0,\dots,a_s$. As in \Cref{sec:betahop}, let $(u=x_0,x_1),(x_1,x_2),\dots,(x_{j-1},x_j=v)$, with $j\le\beta$, be the hops of a path realizing $d_H^{\le\beta}(u,v)$, and let $P_H^{\le\beta}(u,v)$ be the concatenation of $P(x_0,x_1),\dots,P(x_{j-1},x_j)$. It has at most $kc$ edges. Since $(k+1)c<\gamma L$, the proof of \Cref{cl:trace} applies, and gives $P(u,v)\subseteq P_H^{\le\beta}(u,v)$. Moreover, as in the paragraph after \Cref{cl:trace}, each path $P(x_{i-1},x_i)$ meets $P(u,v)$ in a contiguous segment.
We say that a hop \emph{touches} a base edge if its shortest path contains an edge of that subdivided base edge. There are two cases.

\paragraph{Case 1: an endpoint hop touches a base edge.}
This case is the argument of \Cref{thm:dense-transfer}, with $a=a_0$ and $b=a_1$. Reversing the demand path if necessary, we may assume that the first hop $(u,x_1)$ touches a base edge. As there, the segment $P(u,x_1)\cap P(u,v)$ contains $u$, so it is a prefix of $P(u,v)$. Hence it contains the leaf edge $ua_0$ and the first edge of the base edge $a_0a_1$. As there, we charge the demand to the ordered hop $(u,x_1)$. We call this an \emph{endpoint charge}.

\paragraph{Case 2: no endpoint hop touches a base edge.}
Then the internal hops cover all edges of the $s$ base edges. Each of the $s$ base edges is touched by one of the at most $j-2\le s-1$ internal hops. By the pigeonhole principle, some internal hop $(x,y)$ touches two distinct base edges.
Its shortest path meets $P(u,v)$ in a contiguous segment, so it touches two consecutive base edges $a_{i-1}a_i$ and $a_ia_{i+1}$, possibly only parts of them. Thus $P(x,y)$ contains the junction $a_i$ as an internal vertex, together with the two edges of $P(u,v)$ at $a_i$ (see the following \Cref{fig:beta-transfer}). We charge the demand to the hop $(x,y)$, and call this an \emph{internal charge}. Since $P_H^{\le\beta}(u,v)$ has at most $kc$ edges, $\dG(x,y)\le kc$.

\colorlet{chop}{blue!65!black}       
\colorlet{cpath}{black}              
\colorlet{cbranch}{orange!85!black}  

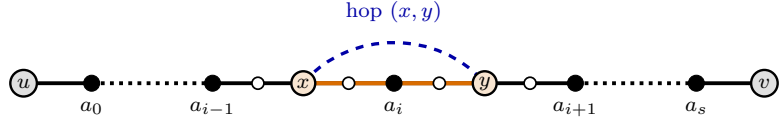
\begin{figure}[h]
\centering
\begin{tikzpicture}[
    vert/.style  = {circle, draw=black, fill=white, minimum size=4.5pt, inner sep=0pt,
                    line width=0.65pt},
    basev/.style = {circle, draw=black, fill=black, minimum size=6pt, inner sep=0pt},
    bvert/.style = {circle, draw=black, fill=black!12, minimum size=10pt, inner sep=0pt,
                    line width=0.9pt, font=\scriptsize},
    pvert/.style = {circle, draw=black, fill=white, minimum size=9pt, inner sep=0pt,
                    line width=0.9pt, font=\scriptsize},
    pedge/.style = {cpath, line width=1.45pt},
]

\coordinate (u)   at (-0.9,0);
\coordinate (a0)  at (0,0);
\coordinate (am)  at (1.6,0);    
\coordinate (ai)  at (4.0,0);    
\coordinate (ap)  at (6.4,0);    
\coordinate (as)  at (8.0,0);    
\coordinate (v)   at (8.9,0);
\foreach \j in {1,2,3} {
  \coordinate (l\j) at (1.6+0.6*\j,0);
  \coordinate (r\j) at (4.0+0.6*\j,0);
}

\draw[pedge] (u) -- (a0);
\draw[pedge, dotted] (a0) -- (am);
\draw[pedge] (am) -- (ap);
\draw[pedge, dotted] (ap) -- (as);
\draw[pedge] (as) -- (v);

\draw[cbranch, line width=1.6pt] (l2) -- (r2);

\draw[chop, dashed, line width=1.25pt] (l2) to[bend left=50] (r2);

\foreach \j in {1,3} {\node[vert] at (l\j) {}; \node[vert] at (r\j) {};}
\node[basev, label={[font=\scriptsize]below:$a_0$}] at (a0) {};
\node[basev, label={[font=\scriptsize]below:$a_{i-1}$}] at (am) {};
\node[basev, label={[font=\scriptsize]below:$a_i$}] at (ai) {};
\node[basev, label={[font=\scriptsize]below:$a_{i+1}$}] at (ap) {};
\node[basev, label={[font=\scriptsize]below:$a_s$}] at (as) {};
\node[pvert, fill=cbranch!18] at (l2) {$x$};
\node[pvert, fill=cbranch!18] at (r2) {$y$};
\node[bvert] at (u) {$u$};
\node[bvert] at (v) {$v$};

\node[chop, font=\scriptsize] at (4.0,0.95) {hop $(x,y)$};

\end{tikzpicture}
\caption{%
The path $P(x,y)$ identifies the two base edges $a_{i-1}a_i$ and $a_ia_{i+1}$ even though it traverses only part of each.%
}
\label{fig:beta-transfer}
\end{figure}

In both cases, the shortest path of the charged hop contains two edges of $P(u,v)$. Hence the charged hop is not an edge of $G$, so it is a hop of $F$.

Second, we bound the number of demands that charge a single hop, separately for Cases 1 and 2.

\begin{claim}\label{cl:dense-endpoint}
Each ordered hop receives at most $D(2D)^{s-1}$ endpoint charges.
\end{claim}
\begin{proof}
Let $(u,x)$ be an ordered hop. Exactly as in \Cref{cl:dense-charge}, the path $P(u,x)$ fixes the leaf $u$, its parent $a_0$, and the base edge $a_0a_1$. Each of $a_2,\dots,a_s$ is a neighbor of the previous base vertex in $B$, so there are fewer than $2D$ choices for each of them. Finally, there are $D$ choices for the leaf $v$ of $a_s$. Hence there are at most $D(2D)^{s-1}$ such demands.
\end{proof}

\begin{claim}\label{cl:dense-internal}
Each hop $(x,y)$ with $\dG(x,y)\le kc$ receives at most $2(kc+1)(s-1)D^2(2D)^{s-2}$ internal charges.
\end{claim}
\begin{proof}
The path $P(x,y)$ has at most $kc+1$ vertices, so there are at most $kc+1$ choices for the junction $a_i$. The two edges of $P(x,y)$ at $a_i$ lie on two subdivided base edges, which are therefore fixed, and there are $2$ choices for which of them is $a_{i-1}a_i$. There are $s-1$ choices for the index $i$. Then, as in \Cref{cl:dense-endpoint}, there are fewer than $2D$ choices for each of the other $s-2$ base vertices, and $D$ choices for each of the two leaves. Hence there are at most $2(kc+1)(s-1)D^2(2D)^{s-2}$ such demands.
\end{proof}

Third, each demand charges a hop of $F$. An unordered hop gives two ordered hops, so by \Cref{cl:dense-endpoint} and \Cref{cl:dense-internal} it is charged at most
\[
  2D(2D)^{s-1}+2(kc+1)(s-1)D^2(2D)^{s-2}=2^{s-1}\bigl(2+(kc+1)(s-1)\bigr)D^s
\]
times. Since $k$, $\beta$, and $\ell$ are fixed, so are $L$ and $c=sL+2$, and this is $O(D^s)$. Thus the choices of a junction and its position cost only a constant factor, whereas the corresponding factors in \Cref{sec:betahop} depend on $\log n$.
Comparing with the number of demands, we get
\[
  \frac{ND^{s+2}}{2^{s}}\le|F|\cdot2^{s-1}\bigl(2+(kc+1)(s-1)\bigr)D^s,\qquad\text{so}\qquad |F|=\Omega\!\left(\frac{ND^{s+2}}{D^s}\right)=\Omega(ND^2).
\]

It remains to express this bound in terms of $m$. Since $M=\Theta(ND)$, we have $|F|=\Omega(ND^2)=\Omega(M^2/N)$. As in the proof of \Cref{thm:dense-transfer}, $M=\Omega(N^{1+1/t})$ gives $N=O(M^{t/(t+1)})$. Together with $m=\Theta(M)$, this yields
\[
 |F|=\Omega\!\left(M^{2-\frac{t}{t+1}}\right)=\Omega\!\left(m^{1+\frac1\ell}\right),
\]
as required.
\end{proof}

\begin{remark}[Scope of the reduction]
\label{rem:transfer-pivots}
\Cref{thm:dense-transfer,thm:beta-dense-transfer} concern $(k,\beta)$-hopsets only. Suppose that every leaf has its parent as a pivot, in addition to itself. Then every demand of these families has an exact answer with two pivot jumps, from $u$ to $a_0$ and from $a_s$ to $v$, and one hop per base edge, that is, $\beta-1$ hops. Storing the base-edge shortcuts and the leaf-to-parent distances takes only $M+Nq=O(n)$ space. This linear-space construction answers all the selected demands. Thus these charging arguments give no superlinear lower bound for $r$-pivot $\beta$-hop ADOs (\Cref{def:piv-ado-beta}) with $r\ge2$, so the technique does not improve \Cref{thm:lb-pivot,thm:lb-pivot-beta}.
\end{remark}

\section{Discussion}\label{sec:discussion}

\paragraph{Hub labelings.}
A \emph{stretch-$k$ hub labeling} of $G$ assigns to each vertex $u$ a label
$L(u)\subseteq V$, stored together with the distances $\dG(u,x)$ for $x\in L(u)$, such that
$\min_{x\in L(u)\cap L(v)}\bigl(\dG(u,x)+\dG(x,v)\bigr)\le k\cdot\dG(u,v)$ for all $u,v\in V$.
Stretch-$k$ hub labelings and $(k,2)$-hopsets are the same object, up to a constant factor and an additive $O(n+m)$ term in size.
Given a labeling, the set $F=\{\{u,x\}:x\in L(u),\,x\ne u\}$ is a $(k,2)$-hopset of size $|F|\le\sum_{u\in V}|L(u)|$.
Given a $(k,2)$-hopset $F$ with hop graph $H$, the labels $L(u)=\{u\}\cup N_H(u)$, where $N_H(u)$ is the set of neighbors of $u$ in $H$, form a stretch-$k$ hub labeling of total size $n+2|E_H|\le n+2(m+|F|)$
(a path with the two hops $(u,x)$ and $(x,v)$ gives the common hub $x$, and a single hop $(u,v)$ gives the common hub $v$).
Hence, by \Cref{lem:lb}, every stretch-$k$ hub labeling of the graphs of that lemma has total label size
$\Omega\!\left(m^{1+1/(k+1)}\right)$, and unconditionally, on the graphs of \Cref{lem:explicit-girth},
$\Omega\!\left(m^{1+4/(5(k+1))}\right)$. On the sparse graphs of \Cref{thm:dense-transfer}, the same equivalence gives the stronger lower bound $\Omega\!\left(m^{1+1/\ell}\right)$ on total label size for every stretch $k<2\ell-1$, where $\ell\ge2$ is an integer.

\paragraph{Compact routing schemes.}
A \emph{compact routing scheme} assigns every vertex a small routing table and a short label, and the goal is to route messages from the source to the destination with little local storage and small stretch.
The best known compact routing schemes for general graphs~\cite{DBLP:conf/spaa/ThorupZ01,DBLP:conf/podc/Chechik13,DBLP:conf/podc/RodittyT15,kadria2025compact} are in fact pivot-based $1$-hop ADOs with additional restrictions. Hence, on the graphs of \Cref{thm:lb-pivot}, every such scheme with stretch $k$ and $\Ot(1)$ pivots per vertex uses $\widetilde\Omega(m^{1+\frac{2}{k+1}})$ total storage, and in particular $\widetilde\Omega(m^{\frac{2}{k+1}})$ storage at some vertex, even in sparse graphs.

\paragraph{Hopsets versus spanners.}
Our lower bounds concern structures that answer a query by combining a bounded number of stored distances, such as hopsets, hub labelings, and pivot-based ADOs.
Spanners and emulators are different in this respect. A spanner also compresses the graph, but extracting a distance from it requires a shortest-path computation over the stored subgraph, which may take linear time.
The bounds reflect this difference. The graph $G$ itself is an exact spanner with $O(m)$ edges, whereas on sparse graphs hopsets with any constant hopbound $\beta$ require superlinear size (\Cref{lem:lb-beta}).

\paragraph{On the girth hypothesis.}
The stronger hopset bounds of \Cref{sec:dense-transfer} use the standard form of the girth conjecture, described below.

\begin{remark}[The standard Erd\H{o}s girth conjecture]\label{rem:sparse}
Previous lower bounds, such as the lower bound of Thorup and Zwick~\cite{DBLP:journals/jacm/ThorupZ05}, rely on the standard Erd\H{o}s girth conjecture, which states that for every integer $\ell\ge1$ there are graphs with girth at least $2\ell+2$ and $\Omega(n^{1+1/\ell})$ edges.
By the Moore bound, a graph with girth at least $2\ell+2$ has $O(n^{1+1/\ell})$ edges. Thus the conjectured edge count matches the Moore bound up to a constant factor.

The arguments of \Cref{sec:robust-hopsets,sec:pivot-lower-bounds} also yield lower bounds under the standard conjecture, on graphs with $n^{1+\eps}$ edges for an arbitrarily small $\eps>0$. Fix a sufficiently large integer $\ell$, and let $G$ be a graph given by the standard conjecture. By the regularization lemma of Erd\H{o}s and Simonovits~\cite{ErdosSimonovits70}, in the form of Jiang and Seiver~\cite[Proposition~2.7]{JiangS12}, $G$ contains a subgraph $G'$ on $n'\to\infty$ vertices with $\Omega(n'^{1+1/\ell})$ edges and maximum degree $\Delta$ at most a constant $K=K(\ell)$ times its minimum degree.\footnote{The proposition is stated for an input density coefficient at least $1$, but its proof in~\cite[Appendix~B]{JiangS12} works for any fixed positive coefficient $a$. Retaining the factor $1/a$ in the bound on the number of iterations changes the resulting order bound only by a constant factor. In particular, $n'\to\infty$ still holds.} As a subgraph of $G$, it has girth $g\ge2\ell+2$. The Moore bound and almost regularity give $\Delta=\Theta(n'^{1/\ell})$ and $g=O(\ell)$.

Set $c=\floor{(2\ell+1)/(k+1)}$, so that $(k+1)c<g$ and $c\le\ell$. Local tree structure gives $\Omega_{\ell}(n'\Delta^c)$ pairs at distance $c$. A pivot jump of length at most $kc$ contains $O_{\ell,k}(1)$ edges, and each such edge lies on at most $O_{\ell}(\Delta^{c-1})$ shortest paths of length $c$, by fixing its position on the path and extending in both directions. Since there are at most $n'r$ pivot jumps, discarding all pairs whose shortest path shares an edge with any such jump removes at most $O_{\ell,k}(n'r\Delta^{c-1})$ pairs. For constant $r$, their fraction is $O_{\ell,k}(r/\Delta)=o(1)$. Thus we may take $t=\rho=1$: every pivot jump used by a remaining pair covers no edge of its shortest path.

The hops must therefore cover all $c$ edges, so one hop covers a window of $\sigma=\ceil{c/\beta}$ edges. For a fixed hop there are $O_{\ell,k}(1)$ windows, each contained in at most $O_{\ell}(\Delta^{c-\sigma})$ of these paths. Dividing the pair count by this bound gives $\Omega_{\ell,k,\beta}(n'\Delta^\sigma)$ stored edges. Since $\sigma/\ell\ge2/(\beta(k+1))-O_{k,\beta}(1/\ell)$, every such oracle requires $\Omega(n'^{1+\frac{2}{\beta(k+1)}-\delta})$ space, where $\delta=O_{k,\beta}(1/\ell)$. Finally, $m=\Theta(n'^{1+1/\ell})$, so expressing the bound in terms of $m$ costs another $O(1/\ell)$ in the exponent. Both losses tend to zero as $\ell$ grows.
\end{remark}

\paragraph{Open problems.}
For every integer $\ell\ge3$, $(2\ell-3,2)$-hopsets of size $O(m^{1+1/\ell})$ exist on sparse graphs (\Cref{sec:models}), so \Cref{thm:dense-transfer} is tight. For $\beta\ge3$, we do not know whether the lower bound of \Cref{thm:beta-dense-transfer} is tight.

\begin{problem}\label{prob:hopset-beta}
For integers $\beta\ge3$ and $\ell>\beta$, is there a $\left(\frac{2\ell-\beta-1}{\beta-1},\beta\right)$-hopset of size $\Ot(m^{1+1/\ell})$ for every graph with $m=O(n)$?
\end{problem}

For pivot-based ADOs, the stretch--space curve is a lower bound for one hop (\Cref{thm:lb-pivot}), but two hops can improve upon it. For $k\in\{2,3,4,5\}$, there are two-hop pivot-based ADOs that use $\Ot(n^{1+1/k})$ space, by Agarwal, Godfrey, and Har-Peled~\cite{DBLP:conf/infocom/AgarwalGH11} and by Kadria and Roditty~\cite{KadriaR25new}. This is below the curve, since $\frac1k<\frac2{k+1}$. For integers $k\ge6$, achieving $\Ot(n^{1+1/k})$ space for every sparse graph remains open, as asked by Kadria and Roditty~\cite{KadriaR25new}.

A similar question arises for pivot-based ADOs with more hops. For $\beta=1$, constructions on the stretch--space curve match our lower bound at the known achievable stretches. For $\beta\ge2$, the corresponding space bound remains open.

\begin{problem}
For every real $k\ge1$ and integer $\beta\ge2$, does every graph with $m=\Theta(n)$ edges admit an $r$-pivot $\beta$-hop $k$-stretch ADO with $r=\Ot(1)$ and space $\Ot(m^{1+\frac{2}{\beta(k+1)}})$?
\end{problem}

Even for $\beta=1$ the picture is not complete. \Cref{thm:lb-pivot} holds for every real $k\ge1$, while matching ADOs are known for every integer stretch, by Thorup and Zwick~\cite{DBLP:journals/jacm/ThorupZ05} for odd integers and by P\v{a}tra\c{s}cu, Roditty, and Thorup~\cite{PatrascuRT12} for even integers (see also~\cite{DBLP:conf/wdag/AbrahamG11}), and for the fractional stretches $k=2a+1\pm2/b$, where $a,b\ge1$ are integers~\cite{PatrascuRT12}.

\begin{problem}
Is there, for every real $k\ge1$, an $\Ot(1)$-pivot $1$-hop ADO with stretch $k$ that uses $\Ot(m^{1+\frac{2}{k+1}})$ space? Equivalently, is every point of the stretch--space curve achievable?
\end{problem}
For every $k<3$ that is not of the form $3-\frac{2}{\ell}$, P\v{a}tra\c{s}cu, Roditty, and Thorup~\cite{PatrascuRT12} show, under a set-intersection hypothesis, that oracles with stretch $k$ and $O(1)$ query time require polynomially more space than the curve on sparse graphs. For example, for $1<k<2$ such oracles require $\widetilde\Omega(n^2)$ space~\cite{PatrascuR14}.
This does not settle the problem. Their bound is conditional, and it applies only to $O(1)$ query time, whereas a pivot-based ADO may have $\Ot(1)$ pivots, and hence $\Ot(1)$ query time.

\paragraph{Use of AI.}
ChatGPT and Claude were used to assist with drafting and revising the exposition, proofreading, checking mathematical arguments, and verifying references. The authors reviewed the suggestions and take full responsibility for the final content, including all results, proofs, and citations.

\bibliographystyle{alpha}
\bibliography{articles}

\appendix
\section{Known Oracles in the Pivot Model}\label{sec:models}
In this section, we show that several known distance oracles are $r$-pivot $\beta$-hop ADOs (\Cref{def:piv-ado,def:piv-ado-beta}) with the parameters of \Cref{tab:conversions}.

\begin{table}[t]
\centering
\small
\renewcommand{\arraystretch}{1.2}
\begin{tabular}{@{}lccccl@{}}
\toprule
Reference & Stretch & Pivots & Hops & Space for $m=\Theta(n)$ & Comments\\
\midrule
\cite{DBLP:journals/jacm/ThorupZ05} & $2\ell-1$ & $\ell$ & $1$ & $\Ot(n^{1+1/\ell})$ & $\ell\ge1$\\
\cite{PatrascuR14} & $2$ & $2$ & $1$ & $O(n^{5/3})$ & --\\
\cite{PatrascuRT12} & $2a+1\pm2/b$ & $O(ab)$ & $1$ & $\Ot(n^{1+1/(a+1\pm1/b)})$ & $a,b\ge1$\\
\cite{KadriaR25new} & $2\ell-3$ & $\ell$ & $2$ & $\Ot(n^{1+1/\ell})$ & $\ell\ge3$\\
\cite{KadriaR25new} & $2\ell-5$ & $\ell$ & $2$ & $\Ot(n^{1+1/\ell})$ & $\ell\ge5$\\
\cite{KadriaR25new} & $4$ & $O(1)$ & $2$ & $\Ot(n^{5/4})$ & --\\
\cite{DBLP:conf/soda/ChechikHL25} & $\ceil{\frac{4\ell}{3}}-1$\textsuperscript{$*$} & $\ell$ & $2$ & $\Ot(n^{1+1/\ell})$ & $\ell\ge4$, unweighted\\
\cite{KadriaRW26} & $\ceil{\frac{4\ell}{3}-\frac{5}{3}}$\textsuperscript{$\dagger$} & $\ell$ & $3$ & $\Ot(n^{1+1/\ell})$ & $\ell\ge4$, unweighted\\
\cite{DBLP:conf/esa/Agarwal14} & $1+1/t$ & $2$ & $2t$ & $\Ot(n^{3/2})$ & $t\ge1$\\
\bottomrule
\end{tabular}
\par\smallskip
\caption{Known upper bounds for pivot-based ADOs. The Kadria--Roditty~\cite{KadriaR25new} and Agarwal~\cite{DBLP:conf/esa/Agarwal14} representations include additional local distances and are stated for constant-degree graphs (\Cref{sec:models}). \textsuperscript{$*$}With an additive error of $O(\ell)$. \textsuperscript{$\dagger$}The main estimates use two hops, and the exceptional distance-$3$ routine uses three.}
\label{tab:conversions}
\end{table}

\paragraph{Thorup--Zwick~\cite{DBLP:journals/jacm/ThorupZ05}.}
The oracle of Thorup and Zwick~\cite{DBLP:journals/jacm/ThorupZ05} has stretch $2\ell-1$ and uses $O(\ell n^{1+1/\ell})$ space. It samples sets $V=A_0\supseteq A_1\supseteq\dots\supseteq A_{\ell-1}$, and throughout this section, $p_i(u)$ denotes the vertex of $A_i$ closest to $u$, and $B(v)$ denotes the bunch of $v$, i.e., the set of vertices whose distances from $v$ the oracle stores~\cite{DBLP:journals/jacm/ThorupZ05}. It returns an estimate of the form
\[\dG(u,p_i(u))+\dG(p_i(u),v),\]
where $0\le i<\ell$ and $p_i(u)\in B(v)$, or the same with the roles of $u$ and $v$ exchanged. Thus, with $\mathcal{P}(u)=\{p_i(u):0\le i<\ell\}$, the oracle is an $\ell$-pivot $1$-hop ADO: $u\to p_i(u)$ is a pivot jump, and the stored distance from $p_i(u)$ to $v$ is the hop.

\paragraph{P\v{a}tra\c{s}cu--Roditty~\cite{PatrascuR14}.}
The oracle of P\v{a}tra\c{s}cu and Roditty~\cite{PatrascuR14} has stretch $2$ and uses $O(n^{4/3}m^{1/3})$ space. It returns either a stored distance $\dG(u,v)$ or an estimate of the form
\[\dG(u,a(u))+\dG(a(u),v),\]
where $a(u)$ is the landmark closest to $u$. Hence, with $\mathcal{P}(u)=\{u,a(u)\}$, the oracle is a $2$-pivot $1$-hop ADO.

\paragraph{P\v{a}tra\c{s}cu--Roditty--Thorup~\cite{PatrascuRT12}.}
For integers $a,b\ge1$, the oracles of P\v{a}tra\c{s}cu, Roditty, and Thorup~\cite{PatrascuRT12} have stretch $2a+1\pm2/b$ and space $\Ot(m^{1+1/(a+1\pm1/b)})$ when $m=\Theta(n)$. The construction combines $O(b)$ generalized Thorup--Zwick~\cite{DBLP:journals/jacm/ThorupZ05} structures, each with $O(a)$ fixed pivots per vertex. An estimate is either a stored distance, a one-pivot estimate as above, or an expression of the form
\[\dG(u,r(u))+\dG(r(u),s(v))+\dG(s(v),v),\]
where $r(u)$ and $s(v)$ are fixed pivots supplied by these structures and the middle distance is stored (or with $u,v$ exchanged). Taking all these pivots, together with the vertex itself, gives an $O(ab)$-pivot $1$-hop ADO. The case $a=1$ with the minus sign recovers stretch $3-2/b$. Taking $b=2$ gives every even integer stretch.

\paragraph{Kadria--Roditty~\cite{KadriaR25new}.}
The oracles of Kadria and Roditty~\cite{KadriaR25new} have stretches $2\ell-5$ and $2\ell-3$ and use $\Ot(m+n^{1+1/\ell})$ space. We show that, on constant-degree graphs, they are $\ell$-pivot $2$-hop ADOs. Let $C(u)=\{w:\dG(u,w)<\dG(w,p_1(w))\}$ and $C^*(u)=C(u)\cup N(C(u))$, where $N(S)$ is the set of vertices adjacent to $S$. Their choice of $A_1$ ensures $|C(u)|=\Ot(n^{1/\ell})$, so constant degree gives $|C^*(u)|=\Ot(n^{1/\ell})$. Add the distances $\dG(u,w)$ for all $w\in C^*(u)$ to $H$, using $\Ot(n^{1+1/\ell})$ additional space.

Besides one-hop estimates, the query uses estimates of the form
\[\dG(u,w)+\dG(w,p_i(w))+\dG(p_i(w),v)
  \qquad\text{or}\qquad
  \dG(u,w)+\dG(w,p_i(v))+\dG(p_i(v),v),\]
where $w\in C^*(u)$ and the other hop distance is stored in a bunch, or with $u,v$ exchanged. With $\mathcal{P}(u)=\{p_i(u):0\le i<\ell\}$, each estimate uses two hops and the indicated pivot jump. The cluster-intersection branch uses two stored hops $u\to w\to v$ with $w\in C^*(u)\cap C^*(v)$. The pivot-to-pivot branch uses one stored hop between endpoint pivots. Replacing local search estimates by exact stored distances can only decrease the answer while preserving noncontraction. Thus the stretch and asymptotic space are preserved. The same argument shows that their stretch-$4$ oracle with $\Ot(m+n^{5/4})$ space is an $O(1)$-pivot $2$-hop ADO.

For $\ell\ge3$, the stretch-$(2\ell-3)$ oracle also gives a $(2\ell-3,2)$-hopset of size $O(n^{1+1/\ell})$ when $m=O(n)$. Its analysis shows that one of the paths $u\to w\to v$ with $w\in C^*(u)\cap C^*(v)$, $u\to w\to p_1(w)\to v$ with $w\in C^*(u)$ and $p_1(w)\in B(v)$, or $x\to p_i(x)\to y$ with $\{x,y\}=\{u,v\}$ and $p_i(x)\in B(y)$ has stretch at most $2\ell-3$. We additionally store $\dG(u,p_1(w))$ for every $w\in C^*(u)$, which replaces the prefix $u\to w\to p_1(w)$ by a single hop, so each of these paths uses at most two stored hops. Without assuming constant degree, the size is still $\Ot(n^{1+1/\ell})$. Since $w\in C(u)$ implies $u\in B(w)$, we have $\sum_u|C^*(u)|\le\sum_w(1+\deg(w))|B(w)|=\Ot((n+m)n^{1/\ell})$. Only the total size matters for a hopset, not the query time. In the Thorup--Zwick~\cite{DBLP:journals/jacm/ThorupZ05} sampling, every bunch has expected size at most $\ell n^{1/\ell}$, so some choice of the sampled sets gives total size $O(\ell(n+m)n^{1/\ell})$, which is $O(n^{1+1/\ell})$ when $m=O(n)$. By \Cref{thm:dense-transfer}, this size is tight.

The earlier oracles of Agarwal, Godfrey, and Har-Peled~\cite{DBLP:conf/infocom/AgarwalGH11} and of Dalirrooyfard, Jin, Vassilevska Williams, and Wein~\cite{DBLP:conf/focs/DalirrooyfardJW22} are very similar to those of Kadria and Roditty~\cite{KadriaR25new}, which improve upon them, so we do not convert them explicitly. The same argument shows that they are $\Ot(1)$-pivot $2$-hop ADOs as well. Moreover, the stretch-$2$ oracle of Agarwal, Godfrey, and Har-Peled~\cite{DBLP:conf/infocom/AgarwalGH11} gives a $(2,2)$-hopset of size $\Ot(n^{3/2})$. If the bunches of $u$ and $v$ intersect at a vertex $w$, it uses $u\to w\to v$, and otherwise it uses $u\to p_1(u)\to v$ or $v\to p_1(v)\to u$.

\paragraph{\nPSP algorithms~\cite{DBLP:conf/soda/ChechikHL25,KadriaRW26}.}
For unweighted graphs, the \nPSP algorithm of Chechik, Hoch, and Lifshitz~\cite{DBLP:conf/soda/ChechikHL25} returns estimates that are at most $(\ceil{\tfrac{4\ell}{3}}-1)(\dG(u,v)+1)$, that is, it has stretch $\ceil{\tfrac{4\ell}{3}}-1$ with an additive error of $\ceil{\tfrac{4\ell}{3}}-1$,\footnote{Our lower bounds also hold for oracles with an additive error $b$ independent of $n$. The proofs of \Cref{sec:robust-hopsets,sec:pivot-lower-bounds} only use pairs at distance $c=\Theta(\log n)$, and replacing the condition $(k+1)c<g$ by $(k+1)c+b<g$ decreases $c$ by $O(b)$. In \Cref{sec:dense-transfer}, it suffices to choose the constant $L$ such that $(k+1)c+b<\gamma L$. This changes the resulting space bounds by only a constant factor.} and the algorithm of Kadria, Roditty, and Vassilevska Williams~\cite{KadriaRW26}, which simplifies it, has stretch $\ceil{\tfrac{4\ell}{3}-\tfrac{5}{3}}$. Both run in $\Ot(mn^{1/\ell}+n^{1+2/\ell})$ time. Apart from the short-distance cases described below, the stretch analysis of~\cite{KadriaRW26} uses only estimates of the form
\[\dG(u,p_u)+\dG(p_u,p_x)+\dG(p_x,x)+\dG(x,p'_x)+\dG(p'_x,p_v)+\dG(p_v,v),\]
where $x$ is a vertex, and $p_u$, $p_x$, $p'_x$, and $p_v$ are Thorup--Zwick~\cite{DBLP:journals/jacm/ThorupZ05} pivots of $u$, $x$, $x$, and $v$, respectively (possibly the vertex itself).\footnote{For adjacent pairs, check whether $(u,v)\in E$. For the exceptional distance-$3$ case, let $B_0(z)=\{w:\dG(z,w)<\dG(z,p_1(z))\}$, let $C_0(z)=\{w:z\in B_0(w)\}$, and set $S(z)=B_0(z)\cup C_0(z)$. The additional query in~\cite[Appendix~B, Claim~20.4.3]{KadriaRW26} minimizes $\dG(u,a)+\dG(a,b)+\dG(b,v)$ over $a\in S(u)$ and $b\in(B(a)\cup C_0(a))\cap S(v)$, where $B(a)$ is the Thorup--Zwick~\cite{DBLP:journals/jacm/ThorupZ05} bunch. In the exceptional case, a shortest path $u,a,b,v$ satisfies these conditions. Each of these distances is stored in the bunch of one of its endpoints, so this routine uses three hops.} Hence, with $\mathcal{P}(u)=\{p_i(u):0\le i<\ell\}$, the stored distances between $p_u$ and $p_x$ and between $p'_x$ and $p_v$ are the only hops of the displayed estimate, so these estimates form an $\ell$-pivot $2$-hop ADO. Including the short-distance routine gives an $\ell$-pivot $3$-hop ADO for~\cite{KadriaRW26}. It uses the Thorup--Zwick~\cite{DBLP:journals/jacm/ThorupZ05} distances and the original edges, so the space is $\Ot(m+n^{1+1/\ell})$.

\paragraph{Agarwal~\cite{DBLP:conf/esa/Agarwal14}.}
The oracle of Agarwal~\cite[Theorem~1]{DBLP:conf/esa/Agarwal14} has stretch $1+1/t$ and space $\Ot(m+n^2/\alpha)$, for every integer $t\ge1$ and parameter $1\le\alpha\le \sqrt n$. It stores distances from $\Ot(n/\alpha)$ landmarks to all vertices. Let $a(w)$ be the nearest landmark to $w$. The \emph{vicinity} of $w$ consists of the vertices closer to $w$ than $a(w)$, together with their neighbors. On constant-degree graphs, each vicinity has $\Ot(\alpha)$ vertices. Storing all vertex-to-vicinity distances therefore gives total space $\Ot(n^2/\alpha+n\alpha)=\Ot(n^2/\alpha)$, since $\alpha\le\sqrt n$.

The stretch analysis guarantees a suitable estimate represented by a walk of the form
\[x_0\to\cdots\to x_j\to v
  \qquad\text{or}\qquad
  x_0\to\cdots\to x_j\to a(x_j)\to v,\]
where $x_0=u$, $0\le j\le2t-1$ in the first form, and $0\le j\le t-1$ in the second (or with $u,v$ exchanged). Consecutive vertices of the local chain are in one another's vicinity in at least one direction. This includes the final step to $v$ only in the first form. Replacing local candidate distances by exact stored distances preserves noncontraction and can only improve the estimate. With $\mathcal{P}(w)=\{w,a(w)\}$, the landmark step is a pivot jump, so the oracle is a $2$-pivot $2t$-hop ADO. Taking $\alpha=\sqrt n$ gives the $\Ot(n^{3/2})$-space representation in \Cref{tab:conversions}.

\end{document}